\documentclass[letterpaper,11pt]{article}

\input{preamble}

\newcommand{\mmathmacro}[1]{\ensuremath{#1}\xspace}
\newcommand{\mmathfunction}[4][\empty]{%
  \ensuremath{%
    \ifthenelse{\equal{#1}{\empty}}{%
      \mathchoice{%
        \ifthenelse{\equal{#2}{R}}{
          #3\left(#4\right)%
        }{%
          #3\left[#4\right]%
        }%
      }{%
        \ifthenelse{\equal{#2}{R}}{
          #3(#4)%
        }{%
          #3[#4]%
        }%
      }{%
        \ifthenelse{\equal{#2}{R}}{
          #3(#4)%
        }{%
          #3[#4]%
        }%
      }{%
        \ifthenelse{\equal{#2}{R}}{
          #3(#4)%
        }{%
          #3[#4]%
        }%
      } 
      \xspace%
    }{%
      \mathchoice{%
        \ifthenelse{\equal{#2}{R}}{
          #3_{#1}\left(#4\right)%
        }{%
          #3_{#1}\left[#4\right]%
        }%
      }{%
        \ifthenelse{\equal{#2}{R}}{
          #3_{#1}(#4)%
        }{%
          #3_{#1}[#4]%
        }%
      }{%
        \ifthenelse{\equal{#2}{R}}{
          #3_{#1}(#4)%
        }{%
          #3_{#1}[#4]%
        }%
      }{%
        \ifthenelse{\equal{#2}{R}}{
          #3_{#1}(#4)%
        }{%
          #3_{#1}[#4]%
        }%
      }
      \xspace%
    }%
  }%
}

\newcommand{\citefull}[1]{\citeauthor*{#1}~\cite{#1}}
\newcommand{\Lineref}[1]{\hyperref[#1]{line~\ref*{#1}}}

\newcommand{\da}{\coloneqq}
\newcommand{\ad}{\eqqcolon}

\newcommand{\xcurr}{x_i}
\newcommand{\Pcurr}{P_{i - 1}}
\newcommand{\pcurr}{p_{i - 1}}
\newcommand{\Rcurr}{R_{i - 1}}
\newcommand{\rcurr}{r_{i - 1}}

\newcommand{\Plater}{P_{>i}}
\newcommand{\plater}{p_{>i}}
\newcommand{\Popt}{P_\text{opt}}
\newcommand{\popt}{p_\text{opt}}
\newcommand{\Ropt}{R_\text{opt}}
\newcommand{\ropt}{r_\text{opt}}
\newcommand{\Rsub}{R'}
\newcommand{\rsub}{r'}
\newcommand{\Rsubo}{R'}
\newcommand{\rsubo}{r'}
\newcommand{\Rsubt}{R''}
\newcommand{\rsubt}{r''}
\newcommand{\gapconst}{\gamma}

\newcommand{\pwin}[1]{p_\text{safe}(#1)}

\newcommand{\rwin}{r_\text{safe}}
\newcommand{\lowerbound}{\ell}

\newcommand{\fderiva}{\hat{f}_\alpha}
\newcommand{\fderivb}{\hat{f}_\beta}

\newcommand{\ie}{i.e.\xspace}
\newcommand{\fo}{\hat{f}} 

\newcommand{\alg}{\mmathmacro{\textsc{Alg}}}
\newcommand{\alglower}{\textsc{Alg}\xspace}
\newcommand{\symb}{\textsc{Symb}\xspace}
\newcommand{\algsymb}{\hyperref[alg:symbiosis]{\symb}\xspace}
\newcommand{\gainsymb}{\ensuremath{\mathrm{gain}}\xspace}
\newcommand{\gain}[2][\empty]{\mmathfunction[#1]{R}{\gainsymb}{#2}}
\newcommand{\Opt}{\mmathmacro{\textsc{Opt}}}
\newcommand{\opt}{\ensuremath{\mathrm{opt}}\xspace}

\newcounter{casenum}
\newenvironment{casedist}{%
  \setcounter{casenum}{0}%
  \newcommand{\case}{\stepcounter{casenum}\item[\textnormal{\textbf{Case~\arabic{casenum}.}}]}%
  \begin{description}%
}{%
  \end{description}%
}

\newcommand{\Ures}{\ensuremath{U_{\mathrm{res}}}}
\newcommand{\Urem}{\ensuremath{U_{\mathrm{rem}}}}
\newcommand{\Usymb}{\ensuremath{U_{\mathrm{symb}}}}

\newcommand{\winterval}[1]{W_{#1}}

\let\phi\varphi
\let\epsilon\varepsilon
\newcommand{\eps}{\epsilon}

\newcommand{\diff}[1]{{\frac{\textrm{d}}{\textrm{d}#1}}}

\newcommand{\xthreefunc}[1]{\chi_{#1}}
\newcommand{\casename}{Case\xspace}
\newcommand{\valname}{u}

\newcommand{\fpolynomial}[2]{\ensuremath Q_{#1,#2}}
\newcommand{\rejectyifunca}{\ensuremath\psi_{\alpha}}
\newcommand{\packyiifunc}[2][\alpha]{\ensuremath\mu(#1,#2)}
\newcommand{\rejectzifuncalpha}[1][\alpha]{\ensuremath\tau_{#1}}

\newcommand{\uniquerootyi}{\ensuremath s_{\alpha}}
\newcommand{\uniquebetayi}{\ensuremath\beta_{\alpha}}

\newcommand{\reservationCaseXii}{1\xspace}
\newcommand{\reservationCasePackXii}{2\xspace}
\newcommand{\reservationCaseRejectXii}{3\xspace}
\newcommand{\reservationCaseReserveXii}{4\xspace}
\newcommand{\reservationCaseRejectFinali}{5\xspace}
\newcommand{\reservationCasePackFinali}{6\xspace}
\newcommand{\reservationCaseRejectYi}{7\xspace}
\newcommand{\reservationCaseRejectFinalii}{8\xspace}
\newcommand{\reservationCasePackFinalii}{9\xspace}
\newcommand{\reservationCaseZi}{10\xspace}
\newcommand{\reservationCaseRejectFinaliii}{11\xspace}
\newcommand{\reservationCaseRejectZi}{12\xspace}
\newcommand{\reservationCaseReserveZi}{13\xspace}
\newcommand{\reservationCasePackFinaliii}{14\xspace}

\newcommand{\symbiosisCaseXii}{1\xspace}
\newcommand{\symbiosisCaseRejectXii}{2\xspace}
\newcommand{\symbiosisCaseReserveXii}{3\xspace}
\newcommand{\symbiosisCaseRejectXiii}{4\xspace}
\newcommand{\symbiosisCasePackXiii}{5\xspace}
\newcommand{\symbiosisCaseRejectYi}{6\xspace}
\newcommand{\symbiosisCaseReserveYi}{7\xspace}
\newcommand{\symbiosisCaseYii}{8\xspace}
\newcommand{\symbiosisCaseRejectYii}{9\xspace}
\newcommand{\symbiosisCasePackYii}{10\xspace}
\newcommand{\symbiosisCaseZi}{11\xspace}
\newcommand{\symbiosisCaseRejectZi}{12\xspace}
\newcommand{\symbiosisCaseReserveZi}{13\xspace}
\newcommand{\symbiosisCaseRejectZii}{14\xspace}
\newcommand{\symbiosisCasePackZii}{15\xspace}

\usepackage{todonotes}

\newcommand{\knapres}{\mmathmacro{\textsc{KnapRes}}}
\newcommand{\knaprem}{\mmathmacro{\textsc{KnapRem}}}
\newcommand{\knapresrem}{\mmathmacro{\textsc{KnapResRem}}}

\newcommand{\calg}{c_\alg}
\newcommand{\cknapres}{c_\knapres}
\newcommand{\cknaprem}{c_\knaprem}
\newcommand{\cknapresrem}{c_\knapresrem}

\newcommand{\cplotscale}{1}
\newcommand{\cplotxmax}{1.05}
\newcommand{\cplotymax}{5}
\pgfplotsset{cplotsettings/.append style = {
    xmin=0, xmax=\cplotxmax,
    ymin=0, ymax=\cplotymax,
    height=0.75*\textwidth,
    width=\textwidth,
    axis x line = bottom,
    axis y line = left,
    xtick align = outside,
    ytick align = outside,
    xlabel style={font=\small},
    ylabel style={font=\small},
    xticklabel style={
        text height=1.5ex,
        text depth=0.8ex
    },
    tick label style={font=\small},
    enlargelimits=false,
    clip=false,
    smooth
}}
\pgfplotsset{cplotsettings2/.append style = {
    xmin=0, xmax=\cplotxmax,
    ymin=0, ymax=\cplotymax,
    height=0.48*0.75*\textwidth,
    width=0.48*\textwidth,
    axis x line = bottom,
    axis y line = left,
    xtick align = outside,
    ytick align = outside,
    xlabel style={font=\small},
    ylabel style={font=\small},
    xticklabel style={
        text height=1.5ex,
        text depth=0.8ex
    },
    tick label style={font=\small},
    enlargelimits=false,
    clip=false,
    smooth
}}
\tikzset{clinesettings/.append style = {
    black!80,
    very thick
}}
\tikzset{cnodesettings/.append style = {
    black!80,
    circle,
    fill,
    inner sep=1pt
}}

\newcommand{\splotxmax}{1.05}
\newcommand{\splotymax}{2}
\pgfplotsset{splotsettings/.append style = {
    xmin=0, xmax=\splotxmax,
    ymin=0, ymax=\splotymax,
    height=0.75*\textwidth,
    width=0.9*\textwidth,
    axis x line = bottom,
    axis y line = left,
    xtick align = outside,
    ytick align = outside,
    xlabel style={font=\small},
    ylabel style={font=\small},
    xticklabel style={
        text height=1.5ex,
        text depth=0.8ex
    },
    tick label style={font=\small},
    enlargelimits=false,
    clip=false,
    smooth
}}
\tikzset{sequilinesettings/.append style = {
    black!80,
    very thick,
    dashed
}}
\tikzset{sequiregionsettings/.append style = {
    pattern={Lines[angle=45,distance=4pt,line width=1.2pt]},
    pattern color=black,
    draw=none
}}
\newcommand{\rescolor}{blue}
\tikzset{sresinternalsettings/.append style = {
    \rescolor,
    thick,
    dotted
}}
\tikzset{sresbordersettings/.append style = {
    \rescolor,
    very thick
}}
\tikzset{sressoftbordersettings/.append style = {
    \rescolor,
    very thick,
    dashed
}}
\tikzset{sresregionsettings/.append style = {
    \rescolor!25
}}
\newcommand{\remcolor}{orange}
\tikzset{sreminternalsettings/.append style = {
    \remcolor,
    thick,
    dotted
}}
\tikzset{srembordersettings/.append style = {
    \remcolor,
    very thick
}}
\tikzset{sremsoftbordersettings/.append style = {
    \remcolor,
    very thick,
    dashed
}}
\tikzset{sremregionsettings/.append style = {
    \remcolor!25
}}
\tikzset{sresremregionsettings/.append style = {
    pattern={Lines[angle=45,distance=4pt,line width=1.2pt]},
    pattern color=\remcolor,
    draw=none
}}
\newcommand{\symcolor}{magenta}
\tikzset{ssyminternalsettings/.append style = {
    \symcolor,
    thick,
    dotted
}}
\tikzset{ssymbordersettings/.append style = {
    \symcolor,
    very thick
}}
\tikzset{ssymsoftbordersettings/.append style = {
    \symcolor,
    very thick,
    dashed
}}
\tikzset{ssymregionsettings/.append style = {
    \symcolor!25
}}

\tikzset{itemsettings/.append style = {
    font=\large
}}
\tikzset{leafsettings/.append style = {
    font=\large,
    color=\symcolor
}}

\tikzset{numbersettings/.append style = {
    circle,
    draw,
    inner sep=0,
    minimum size=15pt,
    font=\small\bfseries
}}

\tikzset{
    ifpack/.style={below = 1cm of #1}
}

\def\rejdist{3cm}
\def\resdist{3cm}

\tikzset{
    ifreserve/.style={anchor=north,
        at=(#1.center),
        yshift=-0.6cm,
        xshift=\resdist,
        }
}
\tikzset{
    ifreservedepth2/.style={anchor=west,
            below right = 4cm and 0.1cm of #1.east}
}
\tikzset{
    ifreservedepth3/.style={anchor=west,
            below right = 6cm and 0.1cm of #1.east}
}
\tikzset{
    ifreject/.style={anchor=north,
        at=(#1.center),
        yshift=-0.6cm,
        xshift=-\rejdist,
        }
}
\tikzset{
    slightlyaboveleft/.style={above left =0.02cm and 0.3cm of #1.north}
}
\tikzset{
    slightlyaboveright/.style={above right =0.02cm and 0.3cm of #1.north}
}

\colorlet{PackColor}{LimeGreen}
\colorlet{RejectColor}{BrickRed}
\colorlet{ReserveColor}{Cerulean}

\tikzset{packarr/.append style = {
    ->,
    very thick,
    PackColor
}}

\tikzset{packimparr/.append style = {
    ->,
    dotted,
    very thick,
    green,
    rounded corners
}}
\tikzset{rejarr/.append style = {
    ->,
    very thick,
    RejectColor,
    rounded corners,
    dashed
}}
\tikzset{resarr/.append style = {
    ->,
    very thick,
    ReserveColor,
    rounded corners,
    dotted
}}
\tikzset{resarr over/.style = {
    resarr,
    preaction={draw=white, line width=6pt}
}}
\tikzset{remarr/.append style = {
    ->,
    very thick,
    orange,
    rounded corners
}}

\title{Pack, Remove, Reserve\\[0.1cm]\Large Online Knapsack with Second Thoughts}

\author{Hans-Joachim Böckenhauer}
\author{Dennis Komm}
\author{Emanuel Skodinis}
\author{Moritz Stocker}
\author{Philip~Whittington}

\affil{Department of Computer Science, ETH Zurich\\%
  \small\texttt{\{hjb, dennis.komm, emanuel.skodinis, moritz.stocker, philip.whittington\}@inf.ethz.ch}}

\date{\today}

\begin{document}

\maketitle

\begin{abstract}
    \noindent
    We study the online proportional knapsack problem with two paid forms of recourse. Items arrive one by one and must be handled immediately, without knowledge of the future: an algorithm may pack an item $x$, reject it, or \emph{reserve} it for possible later use at proportional cost $\alpha x$; additionally, it may at any time \emph{remove} previously packed items, at proportional cost $\beta y$ for each removed item $y$. Reservation and removal have each been analyzed in isolation, but their combination raises a natural question: is the better of the two mechanisms always optimal on its own, or is there a region in the parameter space spanned by $\alpha$ and $\beta$ in which they genuinely enter into a symbiosis? So far, this question has only been answered for the special case of free removal ($\beta = 0$), leaving the vast majority of the parameter space unexplored.

    We close this gap, determining matching upper and lower bounds on the competitive ratio for every pair of cost parameters $(\alpha, \beta)$ and revealing three qualitatively different regimes. In some regions, reservation alone already achieves the optimal ratio; in others, removal alone does. However, most interestingly, in the heart of the parameter space lies a \emph{symbiosis region} in which combining both mechanisms is strictly better than either one on its own. The optimal algorithm in the symbiosis region is a natural blend of the two known single-mechanism strategies: postponing commitment by reserving until a threshold is reached, then packing greedily and revising via removal.
\end{abstract}

\newpage

\section{Introduction}\label{sec:introduction}

In numerous real-world applications, algorithms must make decisions without complete information about future events. They operate, to some extent, ``in the dark,\!'' gradually adapting as new information becomes available. One central example is the \emph{online proportional knapsack problem}, where items of sizes in $(0, 1]$ arrive one by one and an algorithm aims to maximize the total size of the items it packs into a knapsack of capacity $1$; each item's value equals its size, hence the term \emph{proportional}. As future items are unknown and decisions are irrevocable, an online algorithm is at a disadvantage against an optimal offline algorithm that sees the whole sequence in advance. \emph{Competitive analysis} quantifies this disadvantage through the \emph{competitive ratio}, the worst-case ratio between the offline optimum and the online gain.

In its plain form, the problem is hopeless: no deterministic online algorithm achieves a bounded competitive ratio. A single tiny item forces a dilemma---pack it and a later item of size $1$ no longer fits, or reject it and the sequence ends---so the ratio can be made arbitrarily large.

Among the most successful ways to restore a bounded competitive ratio is to allow an algorithm to postpone a decision by reserving an item, or to revise a decision by removing a packed item. Each recourse mechanism has its own proportional cost: reserving an item $x$ incurs a cost of $\alpha x$, removing an item $y$ incurs a cost of $\beta y$. The algorithm's gain is its final packed size minus all costs it encountered while processing the input. Both relaxations restore bounded competitive ratios, and both have been analyzed thoroughly---but only in isolation.

When both mechanisms are available at once, the decision space becomes considerably richer and the two costs interact. Already the special case of \emph{free} removal ($\beta = 0$) combined with paid reservation does not reduce to either single mechanism; the proof by \citefull{BGLMR23} requires a rather involved analysis. Our paper settles the problem for every cost pair $(\alpha, \beta)$ by proposing a remarkably simple algorithm that, for some regions, uses reservation or removal only, but relies on a symbiosis of both mechanisms in a significant region of the parameter plane. Complementing this result with a matching lower bound, we determine the exact competitive ratio on the entire parameter plane.

\subsection{Preliminaries}\label{sub:preliminaries}

An instance of the \emph{online proportional knapsack problem with paid reservation and paid removal} (\knapresrem) is a sequence $I = (x_1, \dots, x_n)$ of items with $0 < x_i \leq 1$. An online algorithm is given a knapsack of capacity $1$ and the items are revealed gradually in consecutive time steps. At the arrival of an item $x$, and without any knowledge of later items, the algorithm may pack $x$ if it fits, \emph{reserve} it at a cost $\alpha x$, or reject $x$. At any time it may remove packed items; removing an item $y$ costs $\beta y$. Rejected and removed items cannot be used later. Reserved items may be packed at any later time, including after the last arrival. In particular, once the algorithm has processed the last item $x_n$, it is informed that no further items will arrive. At that point, it may still pack previously reserved items and remove packed items.

Consider an algorithm $\alg$ being executed on an instance $I$. Let $P$ be its final packed set, let $R$ be the set of all items it ever reserved, and let $D$ be the set of all items it removed. The \emph{gain} of the algorithm is the total size of the items in the final knapsack, minus all reservation and removal costs incurred. In order to formally define the gain, we introduce the following handy notation: for a set of items $S$, we denote its total size by the corresponding lower-case letter $s \da \sum_{x \in S} x$. We will use this notation, analogously for $P$, $R$ and $D$, throughout the paper. \alg's gain on $I$ is therefore
\[
    \gainsymb_\alg(I) \da p - \alpha r - \beta d\;.
\]
We denote by $\Opt$ an arbitrary but fixed optimal offline algorithm, by $\Opt(I)$ its output on $I$, and by $\opt(I)$ the total size of $\Opt(I)$. Note that $\opt(I)$ can never be larger than $1$. \alg is said to be (strictly) \emph{$c$-competitive}\footnote{The more common notion of (non-strict) competitive ratio allows for an additive constant on the right-hand side of the inequality. Since the maximum gain that can be reached for \knapresrem is bounded by the knapsack size of $1$, this more general notion is not reasonable here.} if
\[
    \opt(I) \leq c \cdot \gainsymb_\alg(I)
\]
for every input $I$, and the \emph{competitive ratio} $\calg$ of \alg is the infimum over all such $c$. If $\gainsymb_\alg(I) \leq 0$ for some instance $I$, no finite positive $c$ satisfies the inequality and we set $\calg = \infty$. Lower competitive ratios are better, and no ratio below $1$ is achievable.

When analyzing an online problem such as \knapresrem, we are interested in the \emph{best achievable} competitive ratio $\cknapresrem$, defined as the infimum of $\calg$ over all online algorithms $\alg$. Every algorithm therefore yields an upper bound on $\cknapresrem$. To obtain lower bounds, we construct \emph{adversarial inputs} that force online algorithms into unfavorable decisions. Determining $\cknapresrem$ thus requires establishing matching upper and lower bounds.

It is often convenient to argue about the reciprocal of the competitive ratio, namely the \emph{fill ratio}: \alg achieves fill ratio $f$ if
its competitive ratio is at most $1/f$.

\subsection{Related Work}

Online algorithms are one of the central topics of modern algorithmics; for a general introduction to online computation, we refer to the textbooks by 
\citefull{BE1998} and by \citefull{Komm2016}. 
The concept of competitive analysis was introduced by \citefull{ST1985}. While it still is the backbone in the analysis of online algorithms, it has often been criticized for being overly pessimistic. In recent years, many different approaches have been devised to bridge the gap between online and offline computation by either strengthening the capabilities of the algorithm or weakening the power of the adversary. An overview of these so-called \emph{semi-online} algorithms can be found in a survey by \citefull{KK2021}.

The knapsack problem is not only one of the most studied optimization problems in the offline world (see the textbook by \citefull{KPP2004} or the more recent survey by \citeauthor*{CILM2022a} [\citenum{CILM2022a,CILM2022b}] for an overview), but also one of the key problems in online computation. An overview of online knapsack results is given by \citefull{CJS2016} and more recently by \citefull{BHKRS2026}. The three main reasons for an online problem to be hard are, following the categorization in the latter survey \cite{BHKRS2026}, the instance being unknown to the algorithm, the irrevocability of decisions, and the worst-case assumption of a malicious adversary. In this paper, we focus on relaxing the irrevocability constraint by allowing the algorithm to reserve items and to remove already packed items. 

\citefull{IZ2010} considered removal in the general online knapsack problem, where every item has a size and a (possibly different) value, and the goal is to maximize the value of packed items while respecting the knapsack bound on the size. 

Several other models of removal have been considered. These include the work of \citefull{FHY2013} on a special case of the proportional case with item sizes being multiples of $1/N$ for some integer $N$, the model of removal with recourse, where a limited number of removed items can be brought back, as investigated by \citefull{BKMRSW23}, the model of flush removal, as investigated by \cite{S2025}
, the model of stack/queue removal, where the packed items are organized as a stack or queue for removal, as investigated by \citefull{RT2022}, or a two-stage model of reservation analyzed by \citefull{IKQP2021}.

In the following, we discuss in more detail the models of removal and reservation (and their combination) that are most relevant for the work presented here.

\begin{figure}
    \centering
    
    \begin{subfigure}[t]{0.48\textwidth}
        \centering
        \begin{tikzpicture}[scale=\cplotscale]
	\begin{axis}[
        cplotsettings,
        xtick={0, 0.5, 1},
        xticklabels={$0$, $0.5$, $1$},
        ytick={0, 1, 1.618, 2},
        yticklabels={$0$, $1$, $\varphi$, $2$},
        xlabel={removal cost factor $\beta$},
        ylabel={$\cknaprem (\beta)$}
        ]
        
		\node[
            cnodesettings
        ] at (axis cs:{0},{1.618}) {};

        \addplot[
            clinesettings,
            domain=0:0.5
        ]{2};
        
        \addplot[
            clinesettings,
            domain=0.5:\cplotxmax
        ]{(1 + x + sqrt(x^2 + 2*x + 5))/2};
	\end{axis}
\end{tikzpicture}
        \caption{The competitive ratio $\cknaprem$ as a function of $\beta$.}
        \label{fig:cknaprem}
    \end{subfigure}
    \hfill
    \begin{subfigure}[t]{0.48\textwidth}
        \centering
        \begin{tikzpicture}[scale=\cplotscale]
	\begin{axis}[
        cplotsettings,
        xtick={0, 0.25, 0.414, 0.618, 1},
        xticklabels={$0$, $0.25$, $\sqrt{2} - 1$, $\varphi - 1$, $1$},
        ytick={0, 1, 2},
        yticklabels={$0$, $1$, $2$},
        xlabel={reservation cost factor $\alpha$},
        ylabel={$\cknapres (\alpha)$}
        ]
        
		\node[
            black!80,
            circle,
            fill,
            inner sep=1pt
        ] at (axis cs:{0},{1}) {};

        \addplot[
            clinesettings,
            domain=0:0.25
        ]{2};
        
        \addplot[
            clinesettings,
            domain=0.25:0.414
        ]{(1 + sqrt(5 - 4*x))/(2 - 2*x)};

        \addplot[
            clinesettings,
            domain=0.414:0.618
        ]{2 + x};

        \addplot[
            clinesettings,
            domain=0.618:1,
            samples=200,
            restrict y to domain=0:\cplotymax
        ]{1/(1 - x)};
	\end{axis}
\end{tikzpicture}
        \caption{The competitive ratio $\cknapres$ as a function of $\alpha$.}
        \label{fig:cknapres}
    \end{subfigure}
    
    \caption{The best achievable competitive ratios if only one mechanism is available.}
    \label{fig:c_one_mech}
\end{figure}
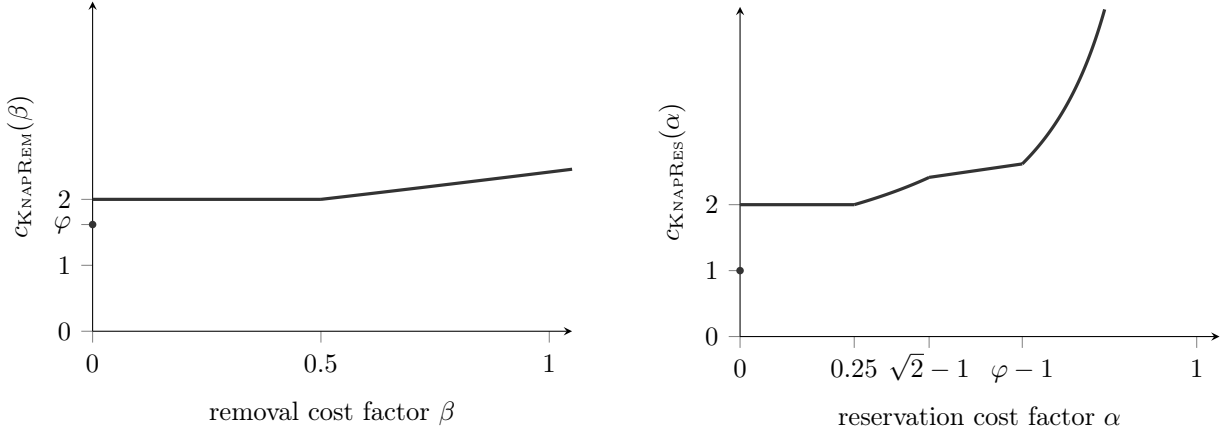

\paragraph{Removal.} \citefull{IT2002} introduced the removable online knapsack problem with \emph{free} removal and proved a tight competitive ratio of the golden ratio $\varphi = \frac{1 + \sqrt{5}}{2} \approx 1.618$. \citeauthor*{HKM2014}~\cite{HKM2012, HKM2014} extended this result to proportional \emph{removal costs}~$\beta$ (\knaprem), obtaining a tight bound which, together with the result of \citeauthor{IT2002}, yields the competitive ratio $\cknaprem$ (see \cref{fig:cknaprem}):
\begin{equation}
    \cknaprem (\beta) = \begin{cases}
                    \varphi,                                             &\text{for } \beta = 0\;, \\
                    2,                                                   &\text{for } 0 < \beta \leq \frac{1}{2}\;, \\
                    \frac{1 + \beta + \sqrt{\beta^2 + 2\beta + 5}}{2},   &\text{for } \beta > \frac{1}{2}\;.
                \end{cases} \label{eq:cknaprem}
\end{equation}

\paragraph{Reservation.} \citefull{BBHLR2021} introduced knapsack with paid reservation (\knapres), and, together with Frei, determined the tight competitive ratio $\cknapres$ for the whole range of reservation costs \cite{BBFHLR22} (see \cref{fig:cknapres}):
\begin{equation}
    \cknapres (\alpha) = \begin{cases}
                    1,                                               &\text{for } \alpha = 0\;, \\
                    2,                                               &\text{for } 0 < \alpha \leq \frac{1}{4}\;, \\
                    \frac{1 + \sqrt{5 - 4 \alpha}}{2 (1 - \alpha)},  &\text{for } \frac{1}{4} < \alpha \leq \sqrt{2} - 1\;, \\
                    2 + \alpha,                                      &\text{for } \sqrt{2} - 1 < \alpha \leq \varphi - 1\;, \\
                    \frac{1}{1 - \alpha},                            &\text{for } \varphi - 1 < \alpha < 1\;.
                \end{cases} \label{eq:cknapres}
\end{equation}

For $\alpha \geq 1$, a reserved item $x$ contributes at most $x - \alpha x \leq 0$ to the gain even if it is packed later, so reservation is never beneficial and \knapres inherits the unbounded competitive ratio of the plain problem; we accordingly set $\cknapres(\alpha) \da \infty$ for $\alpha \geq 1$.
\paragraph{Reservation and removal combined.} Closest to our work, \citefull{BGLMR23} analyzed paid reservation together with \emph{free} removal, \ie, $\beta = 0$, and proved that
\begin{equation}
    \cknapresrem (\alpha, 0) =  \begin{cases}
                                    1, &\text{for } \alpha = 0\;, \\
                                    \frac{3 - 1.5 \alpha}{2 - 1.5 \alpha}, &\text{for } 0 < \alpha \leq 1 - \frac{\sqrt{5}}{3}\;, \\
                                    \varphi, &\text{for } \alpha > 1 - \frac{\sqrt{5}}{3}\;.
                                \end{cases} \label{eq:cknapresfreerem}
\end{equation}

Crucially, for $0<\alpha<1-\sqrt{5}/3$, neither single mechanism is optimal: free removal interacts with paid reservation to beat both. This is the first evidence that the two mechanisms can be symbiotic, and it leaves open what happens once removal also carries a cost---a question we resolve in this work.

\subsection{Our Contribution}

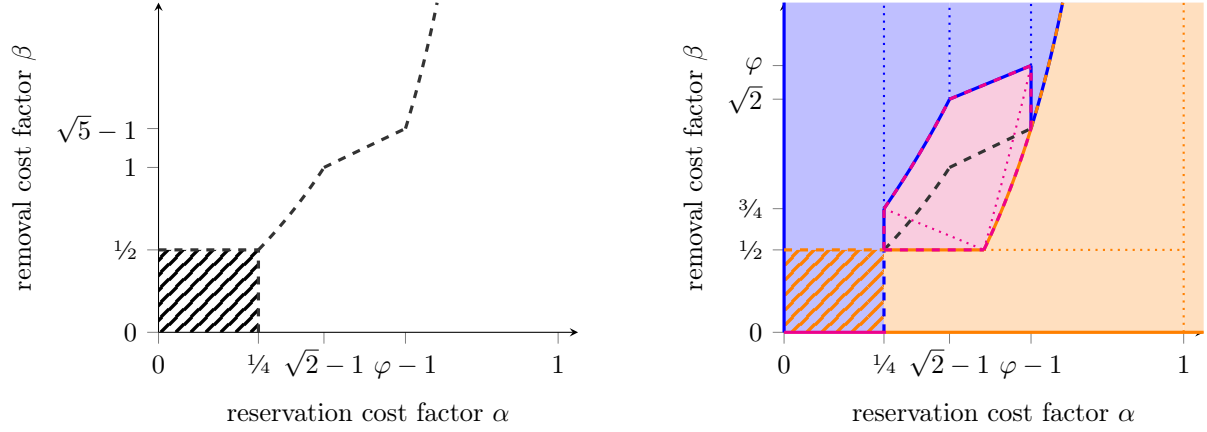
\begin{figure}
    \centering

    \begin{subfigure}[t]{0.48\textwidth}
        \centering
        \begin{tikzpicture}
	\begin{axis}[
        splotsettings,
        xtick={0, 0.25, 0.414, 0.618, 1},
        xticklabels={$0$, \sfrac{1}{4}, $\sqrt{2} - 1$, $\varphi - 1$, $1$},
        ytick={0, 0.5, 1, 2.236-1},
        yticklabels={$0$, \sfrac{1}{2}, $1$, $\sqrt{5} - 1$},
        xlabel={reservation cost factor $\alpha$},
        ylabel={removal cost factor $\beta$}
        ]

        \addplot[
            sequilinesettings,
            domain=0.25:0.414
        ]{(x*(sqrt(5 - 4*x) + 1))/(2*(1 - x))};

        \addplot[
            sequilinesettings,
            domain=0.414:0.618
        ]{(x^2 + 3*x + 1)/(x + 2)};

        \addplot[
            sequilinesettings,
            domain=0.618:1,
            restrict y to domain=0:\splotymax,
            samples=1000
        ]{(-x^2 + 3*x - 1)/(1 - x)};

        \addplot[
            name path=upper,
            sequilinesettings,
            domain=0:0.25
        ]{0.5};

        \addplot[
            sequilinesettings
        ] coordinates {(0.25,0) (0.25,0.5)};
        
        \addplot[
            name path=xaxis,
            draw=none,
            domain=0:0.25
        ]{0};

        \addplot[
            sequiregionsettings
        ] fill between[
            of=upper and xaxis
        ];
	\end{axis}
\end{tikzpicture}
        \caption{The equilibrium border separating the regions where one mechanism is better than the other.
        In the dashed area, both mechanisms perform equally well.}
        \label{fig:equilibrium_border}
    \end{subfigure}
    \hfill
    \begin{subfigure}[t]{0.48\textwidth}
        \centering
        \begin{tikzpicture}
	\begin{axis}[
        splotsettings,
        xtick={0, 0.25, 0.414, 0.618, 1},
        xticklabels={$0$, \sfrac{1}{4}, $\sqrt{2} - 1$, $\varphi - 1$, $1$},
        ytick={0, 0.5, 0.75, 1.414, 1.618},
        yticklabels={$0$, \sfrac{1}{2}, \sfrac{3}{4}, $\sqrt{2}$, $\varphi$},
        xlabel={reservation cost factor $\alpha$},
        ylabel={removal cost factor $\beta$}
        ]

        \addplot[
            sequilinesettings,
            domain=0.25:0.414
        ]{(x*(sqrt(5 - 4*x) + 1))/(2*(1 - x))};

        \addplot[
            sequilinesettings,
            domain=0.414:0.618
        ]{(x^2 + 3*x + 1)/(x + 2)};

        \addplot[
            sresinternalsettings
        ] coordinates {(0.25, 0.75) (0.25, \splotymax)};

        \addplot[
            sresinternalsettings
        ] coordinates {(0.414, 1.414) ((0.414, \splotymax)};

        \addplot[
            sresinternalsettings
        ] coordinates {(0.618, 1.618) (0.618, \splotymax)};

        \addplot[
            sresbordersettings
        ] coordinates {(0, 0) (0, \splotymax)};

        \addplot[
            sresbordersettings
        ] coordinates {(0.25, 0.5) (0.25, 0.75)};

        \addplot[
            name path=bottommidres,
            sresbordersettings,
            domain=0.25:0.414
        ] {(x*(2*x - sqrt(5 - 4*x) - 3))/(2*(x - 1))};

        \addplot[
            name path=bottomhighres,
            sresbordersettings,
            domain=0.414:0.618
        ] {1 + x};

        \addplot[
            sresbordersettings
        ] coordinates {(0.618, 1.236) (0.618, 1.618)};

        \addplot[
            name path=bottomexpensiveres,
            sresbordersettings,
            domain=0.618:1,
            restrict y to domain=0:\splotymax,
            samples=1000
        ] {(-x^2 + 3*x - 1)/(1 - x)};

        \addplot[
            sressoftbordersettings
        ] coordinates {(0.25, 0) (0.25, 0.5)};

        \addplot[
            name path=ceilcheapres,
            draw=none,
            domain=0:0.25
        ] {\splotymax};

        \addplot[
            name path=ceilmidres,
            draw=none,
            domain=0.25:0.414
        ] {\splotymax};

        \addplot[
            name path=ceilhighres,
            draw=none,
            domain=0.414:0.618
        ] {\splotymax};

        \addplot[
            name path=ceilexpensiveres,
            draw=none,
            domain=0.618:1
        ] {\splotymax};

        \addplot[
            name path=bottomcheapres,
            draw=none,
            domain=0:0.25
        ] {0};

        \addplot[
            sresregionsettings
        ] fill between[
            of=bottomcheapres and ceilcheapres
        ];

        \addplot[
            sresregionsettings
        ] fill between[
            of=bottommidres and ceilmidres
        ];

        \addplot[
            sresregionsettings
        ] fill between[
            of=bottomhighres and ceilhighres
        ];

        \addplot[
            sresregionsettings
        ] fill between[
            of=bottomexpensiveres and ceilexpensiveres
        ];

        \addplot[
            sreminternalsettings,
            domain=0.5:1
        ] {0.5};

        \addplot[
            name path=leftsuperexpensiverem,
            sreminternalsettings
        ] coordinates {(1, 0) (1, \splotymax)};

        \addplot[
            srembordersettings,
            domain=0.255:\splotxmax
        ] {0};

        \addplot[
            srembordersettings,
            domain=0.25:0.5
        ] {0.5};

        \addplot[
            srembordersettings,
            domain=0.5:0.618
        ] {(-x^2 + 3*x - 1)/(1 - x)};

        \addplot[
            sremsoftbordersettings,
            domain=0.618:1,
            restrict y to domain=0:\splotymax,
            samples=1000
        ] {(-x^2 + 3*x - 1)/(1 - x)};

        \addplot[
            name path=ceilcheapresrem,
            sremsoftbordersettings,
            domain=0:0.25
        ] {0.5};

        \addplot[
            name path=ceilcheaprem,
            draw=none,
            domain=0.25:1
        ] {0.5};

        \addplot[
            name path=bottomcheaprem,
            draw=none,
            domain=0.25:1
        ] {0};

        \addplot[
            name path=leftexpensiverem,
            draw=none,
            domain=0.5:1,
            restrict y to domain=0:\splotymax,
            samples=1000
        ] {(-x^2 + 3*x - 1)/(1 - x)};

        \addplot[
            name path=rightexpensiverem,
            draw=none
        ] coordinates {(1, 0.5) (1, \splotymax)};

        \addplot[
            name path=rightsuperexpensiverem,
            draw=none
        ] coordinates {(\splotxmax, 0) (\splotxmax, \splotymax)};

        \addplot[
            sresremregionsettings
        ] fill between[
            of=bottomcheapres and ceilcheapresrem
        ];

        \addplot[
            sremregionsettings
        ] fill between[
            of=bottomcheaprem and ceilcheaprem
        ];

        \addplot[
            sremregionsettings
        ] fill between[
            of=leftexpensiverem and rightexpensiverem
        ];

        \addplot[
            sremregionsettings
        ] fill between[
            of=leftsuperexpensiverem and rightsuperexpensiverem
        ];

        \addplot[
            name path=leftintersym,
            ssyminternalsettings,
            domain=0.25:0.5
        ] {1 - x};

        \addplot[
            name path=rightintersym,
            ssyminternalsettings,
            domain=0.5:0.618
        ] {x + (2*x - 1)/(x*(1 - x))};

        \addplot[
            ssymbordersettings,
            domain=0:0.255
        ] {0};

        \addplot[
            ssymsoftbordersettings
        ] coordinates {(0.25, 0.5) (0.25, 0.75)};

        \addplot[
            name path=leftceilsym,
            ssymsoftbordersettings,
            domain=0.25:0.414
        ] {(x*(2*x - sqrt(5 - 4*x) - 3))/(2*(x - 1))};

        \addplot[
            ssymsoftbordersettings,
            domain=0.414:0.618
        ] {1 + x};

        \addplot[
            ssymsoftbordersettings
        ] coordinates {(0.618, 1.236) (0.618, 1.618)};

        \addplot[
            name path=leftbottomsym,
            ssymsoftbordersettings,
            domain=0.25:0.5
        ] {0.5};

        \addplot[
            name path=rightbottomsym,
            ssymsoftbordersettings,
            domain=0.5:0.618
        ] {(-x^2 + 3*x - 1)/(1 - x)};

        \addplot[
            ssymregionsettings
        ] fill between[
            of=leftbottomsym and leftintersym
        ];

        \addplot[
            ssymregionsettings
        ] fill between[
            of= leftceilsym and rightintersym
        ];

        \addplot[
            ssymregionsettings
        ] fill between[
            of=rightbottomsym and rightintersym
        ];
	\end{axis}
\end{tikzpicture}
        \caption{The regions in which reservation-only dominates (blue), removal-only dominates (orange), and both mechanisms behave symbiotically (pink). The equilibrium border is drawn in black.}
        \label{fig:results}
    \end{subfigure}

    \caption{The $(\alpha, \beta)$-plane separated by the equilibrium border and categorized into regions in which one mechanism dominates or a symbiosis of both mechanisms helps.}
    \label{fig:equilib_and_results}
\end{figure}
We give a complete analysis of the proportional online knapsack problem with both paid reservation and paid removal (\knapresrem): for every pair $(\alpha, \beta)$, we determine the exact competitive ratio $\cknapresrem$, providing matching upper and lower bounds.

Before we state our results, assume for a moment that, for every $(\alpha, \beta)$, the optimal competitive ratio were exactly $\min(\cknapres, \cknaprem)$, so that each parameter pair were governed by whichever single mechanism is more effective. Where the two guarantees coincide, \ie, where
\begin{equation}
    \cknapres (\alpha) = \cknaprem (\beta)\;, \label{eq:implicit_equilibrium_border}
\end{equation}
the natural boundary between a reservation-dominated and a removal-dominated regime is located; we call it the \emph{equilibrium border}. Solving \cref{eq:implicit_equilibrium_border} with the help of \cref{eq:cknaprem,eq:cknapres} yields the explicit description
\begin{equation}\label{eq:explicit_equilibrium_border}
    \beta(\alpha) = \begin{cases}
                        \frac{\alpha \cdot (\sqrt{5 - 4\alpha} + 1)}{2 (1 - \alpha)}, &\text{for } \frac{1}{4} \leq \alpha < \sqrt{2} - 1\;, \\
                        \frac{\alpha^2 + 3\alpha + 1}{\alpha + 2}, &\text{for } \sqrt{2} - 1 \leq \alpha < \varphi - 1\;, \\
                        \frac{3\alpha - \alpha^2 - 1}{1 - \alpha}, &\text{for } \varphi - 1 \leq \alpha < 1\;,
                    \end{cases}
\end{equation}
together with the degenerate area $\{0 < \alpha \leq \frac{1}{4},\; 0 < \beta \leq \frac{1}{2}\}$, on which both ratios equal $2$ (see \cref{fig:equilibrium_border}). Note that for $\alpha = 0$, reservation is better than removal for any $\beta \geq 0$, and for $\beta = 0$, removal is better than reservation for any $\alpha > 0$. The equilibrium border defines the region
\[
    \Ures \da \left\{ (\alpha, \beta) \; \Big| \; \left( \alpha = 0 \land \beta \geq 0 \right) \lor \left( 0 < \alpha \leq \frac{1}{4} \land \beta > 0 \right) \lor \left( \frac{1}{4} < \alpha < 1 \land \beta \geq \beta(\alpha) \right) \right\}\;,
\]
in which reservation is at least as good as removal, and the region
\[
    \Urem \da \left\{ (\alpha, \beta) \; \Big| \; \left( 0 < \alpha \leq \frac{1}{4} \land \beta \leq \frac{1}{2} \right) \lor \left( \frac{1}{4} < \alpha < 1 \land \beta \leq \beta(\alpha) \right) \lor \left( \alpha \geq 1 \land \beta \geq 0 \right) \right\}\;,
\]
in which removal is at least as good as reservation.

We show that, once both reservation and removal are available, this single-mechanism picture remains accurate only on part of the parameter plane. Where at least one mechanism is cheap or at least one is expensive---concretely, before the first kink of the equilibrium border at $\alpha = \frac{1}{4}$ and after the third kink at $\alpha = \varphi - 1$---the border is sharp: one mechanism dominates the other, as we will show in \cref{sec:dominance}. In the intermediate band, however, there is a \emph{symbiosis region} in which an optimal algorithm must use reservation and removal together, achieving a ratio strictly better than either mechanism can achieve on its own; this will be shown in \cref{sec:symbiosis}. Before we summarize our results, we need to formalize two concepts.

\begin{definition}\label{def:Usymb-f}
    We define the \textbf{main symbiosis region}
    \begin{equation}\label{def:Usymb}
        \Usymb \da \left\{ (\alpha, \beta) \;\middle|\;
        \begin{aligned}
            &\tfrac{1}{4} < \alpha < \varphi - 1 \text{ and}\\
            &\max\left( \tfrac{1}{2}, \tfrac{3\alpha -\alpha^2 - 1}{1 - \alpha} \right) < \beta < \min\left( \tfrac{\alpha\cdot (3 - 2\alpha + \sqrt{5 - 4\alpha})}{2\cdot(1 - \alpha)},\; 1 + \alpha \right)
        \end{aligned}
        \right\}\;,
    \end{equation}
    and the \textbf{target fill ratio}
    \begin{equation}\label{def:f}
        f (\alpha, \beta) \da \min \left( \fo(\alpha, \beta),\; \frac{1}{2},\; 1 - \alpha \right),
    \end{equation}
    where, for $\beta \leq 1 + \alpha$, $\fo(\alpha, \beta)$ is the unique positive root of the polynomial
    \begin{equation}\label{def:Q}
        Q_{\alpha, \beta}(x) \da (1 + \alpha - \beta)\, x^2 + (1 + \beta)\, x - 1\;.
    \end{equation}
    For $\beta < 1 + \alpha$, this root is given explicitly by
    \begin{equation}\label{def:fo}
        \fo(\alpha, \beta) = \frac{- (1 + \beta) + \sqrt{(1 + \beta)^2 + 4\cdot(1 + \alpha - \beta)}}{2\cdot(1 + \alpha - \beta)}\;.
    \end{equation}
    Rearranging $Q_{\alpha, \beta}(\fo) = 0$, using $1 - \beta \fo > 0$ (see \cref{sub:fill_ratio}), yields the identity 
    \begin{equation}\label{eq:fo}
        \frac{\fo^2}{1 - \beta \fo} = \frac{1 - \fo}{1 + \alpha}\;,
    \end{equation}
    which we use frequently. Whenever $(\alpha, \beta)$ is fixed or clear from context, we write $f$ and $\fo$ for short.
\end{definition}

The target fill ratio $f$ is non-increasing in $\alpha$ and $\beta$ (\cref{lem:f_non-increasing}), and its minimum in \cref{def:f} is attained piecewise across $\Usymb$ as follows (\cref{lem:f_piecewise}):
\[
    f = \begin{cases}
            \fo, &\text{for } \beta > \max \left( 1 - \alpha, \alpha + \frac{2\alpha - 1}{\alpha\cdot(1 - \alpha)} \right), \\
            \frac{1}{2}, &\text{for } \beta \leq 1 - \alpha\;, \\
            1 - \alpha, &\text{for } \beta \leq \alpha + \frac{2\alpha - 1}{\alpha\cdot(1 - \alpha)}\;.
        \end{cases}
\]

On $\Usymb$ itself, $1/f$ is strictly smaller than both $\cknapres$ and $\cknaprem$, corresponding to a strictly better performance (\cref{lem:f_strictly_better}). By \cref{thm:results} and \cref{lem:f_strictly_better}, for $\beta > 0$ the symbiosis region is precisely the set of cost pairs at which the optimal competitive ratio strictly improves upon both single mechanisms.

The following theorem summarizes our results for $\beta > 0$, which together with \cref{eq:cknapresfreerem} analyzes the entire parameter plane; this is visualized in \cref{fig:results}.

\begin{restatable}{theorem}{theoremMain}\label{thm:results}
    For $\alpha \geq 0$ and $\beta > 0$, we have
    \[
        \cknapresrem (\alpha, \beta) =  \begin{cases}
                            \frac{1}{f(\alpha, \beta)}, &\text{for } (\alpha, \beta) \in \Usymb\;, \\
                            \cknapres (\alpha),   &\text{for } (\alpha, \beta) \in \Ures \setminus \Usymb\;, \\
                            \cknaprem (\beta),   &\text{for } (\alpha, \beta) \in \Urem \setminus \Usymb\;.
                        \end{cases}
    \]
\end{restatable}

\begin{proof}[sketch]
    For $\alpha = 0$, reservation is free, so an algorithm can reserve every arriving item and, once the instance has ended, pack an optimal subset; hence $\cknapresrem(0, \beta) = \cknapres(0) = 1$. For $\alpha \geq 1$, reservation is never beneficial (a reserved item $x$ contributes at most $x - \alpha x \leq 0$ to the gain even if it is packed later), so the problem reduces to \knaprem and inherits both its upper and its lower bound, \ie, $\cknapresrem(\alpha, \beta) = \cknaprem(\beta)$. For $0 < \alpha < 1$, the upper bounds are given by \cref{thm:alg_achieves_f} on $\Usymb$ and by \cref{obs:single_mechanism_upper} everywhere else; they are matched from below by \cref{thm:globaltwo,thm:lowerboundsymbiosistrivial,thm:lowerboundf} on $\Usymb$ and by \cref{thm:globaltwo,thm:lowerbound_reservation_small,thm:lowerbound_reservation_medium,thm:lower-equib} on the dominance regions, each applied on the subregion where its bound is tight. The full, region-by-region assembly is given in \cref{sec:assembly}.
\end{proof}

Outside $\Usymb$, the common boundary of $\Ures$ and $\Urem$ is the equilibrium border; the symbiosis region is the area carved out of its neighborhood where the border fails to be sharp.

\section{Symbiosis}\label{sec:symbiosis}

We present an algorithm for the symbiosis region $\Usymb$ strictly outperforming every reservation-only and every removal-only strategy, and show that it is optimal by providing a matching lower bound.

Throughout this section, we fix $(\alpha, \beta) \in \Usymb$ and write $f = f(\alpha, \beta)$. For a reserve $R$ of total size $r \geq 0$, we define the \emph{safety threshold} $\pwin{r}$ as follows. Note that $1 - \beta f > 0$; see \cref{eq:1-bf>0_1-a-af>=(1-a)^2>0}.
\begin{equation}\label{def:pwin}
    \pwin{r} \da \frac{f^2 - (1 - \alpha - \alpha f)\cdot r}{1 - \beta f}\;.
\end{equation}

\begin{algorithm}\caption{Symbiosis algorithm \symb}\label{alg:symbiosis}
    \begin{algorithmic}[1]
    \State $P \gets \emptyset$, $R \gets \emptyset$\label{line:init}
    \Statex \emph{Phase 1 (reserve):}
    \While{a next item $x$ arrives and $x < \pwin{r}$}\label{line:phase1check}
        \State reserve $x$: $R \gets R \cup \{x\}$\label{line:reserve}
    \EndWhile
    \Statex \emph{Phase 2 (pack, remove, or reject):}
    \For{each remaining item $x$, starting with the one that ended Phase~1}\label{line:phase2for}
        \If{there are $P' \subseteq P$ and $R' \subseteq R$ such that $f + \alpha r + \beta\cdot(p - p') \leq p' + r' + x \leq 1$}\label{line:checkwin}
            \State remove $P \setminus P'$, pack $R' \cup \{x\}$, and reject all future items \Comment{\symb wins early}\label{line:win}
        \ElsIf{$p + x \leq 1$}\label{line:ifpackable}
            \State pack $x$: $P \gets P \cup \{x\}$\label{line:pack}
        \Else
            \State reject $x$\label{line:reject}
        \EndIf
    \EndFor
    \Statex \emph{Finalization:}
    \If{\symb did not win early}\label{line:ifover}
        \State pack all reserved items: $P \gets P \cup R$\label{line:packall}
    \EndIf
    \end{algorithmic}
\end{algorithm}

The algorithm \algsymb (see \cref{alg:symbiosis}) itself is simple and proceeds in two phases. In Phase~1, it reserves \emph{every} arriving item, deliberately postponing commitment until a sufficiently large item of size at least $\pwin{r}$ arrives and ends the phase. In Phase~2, no further items are reserved; for each item $x$, beginning with the very item that ended Phase~1, \algsymb first checks whether it can immediately attain a gain of at least $f$ by packing $x$ and a subset of the reserved items, removing previously packed items from the knapsack if necessary. The condition tested in \Lineref{line:checkwin} is checking exactly that, accounting for the reservation costs $\alpha r$ and the removal costs $\beta\cdot (p - p')$. When this is possible, we say that \algsymb \emph{wins early} as it attains a fill ratio of at least $f$. If \algsymb cannot win early in this step, it tries to pack $x$ greedily and rejects it otherwise. Should the input end before \algsymb has won early, it packs its entire reserve and achieves the target fill ratio this way.

Although \algsymb is this simple, its analysis is rather involved. Clearly, if \algsymb ever executes \Lineref{line:win}, it achieves its target fill ratio $f$. Otherwise, the analysis is based on three main ingredients. The first is the \emph{packability invariant}: the whole reserve is packable at all times, \ie, $p + r \leq 1$, so in particular, \Lineref{line:packall} is always executable. The second is the \emph{critical rejection} assertion: at most one item rejected by \algsymb is packed by \Opt. The third is \emph{safe rejection}: even when \algsymb rejects such an item, it still reaches its target fill ratio $f$.

It is worth noting that \algsymb is a natural blend of the two optimal single-mechanism algorithms. The reservation-only algorithm of \citefull{BBFHLR22} \emph{reserves} greedily, checks at each arrival whether it can already win, and, should the input end first, packs its entire reserve; the removal-only algorithm of \citefull{HKM2014} \emph{packs} greedily and checks at each arrival whether it can win early by removing packed items in favor of the arriving one. \algsymb runs the former as its Phase~1 and switches to the latter as soon as its filling guarantees safe rejection, which is exactly its Phase~2. This symbiosis between postponing commitment through reservation and revising decisions through removal is exactly what enables \algsymb to outperform any single-mechanism algorithm.

\subsection{Upper Bound}\label{sub:symbiosis:upper_bound}
In this subsection, we show that \algsymb has a competitive ratio of at most $\frac{1}{f(\alpha, \beta)}$.

\begin{theorem}\label{thm:alg_achieves_f}
    For every $(\alpha, \beta) \in \Usymb$, \algsymb achieves a fill ratio of at least $f(\alpha, \beta)$.
\end{theorem}

We prove \cref{thm:alg_achieves_f} through three lemmas, one for each of the ingredients identified above: the packability invariant, critical rejection, and safe rejection. 
All three rely on the following common setup.
Let $(\alpha, \beta) \in \Usymb$, let $I = (x_1, \dots, x_n)$ be an instance of \knapresrem, and consider the run of \algsymb on $I$. 
For each $i$, let $P_i$ and $R_i$ be the sets of packed and reserved items, respectively, immediately \emph{after} \algsymb has processed $x_i$.  
It is convenient to include the boundary index $0$ as well: $P_0 = R_0 = \emptyset$ is the empty state before the first item, so all indices range over $0 \leq i \leq n$. Again, we denote the total size of the sets $P_i$ and $R_i$ by $p_i$ and $r_i$, respectively. The finalization step, in which \algsymb may still pack its reserve, receives no index of its own; where it matters, we account for it explicitly.
Finally, we assume throughout that \algsymb never wins early, as that situation already guarantees a fill ratio of $f$.

\begin{restatable}[Packability Invariant]{lemma}{LemmaPackability}\label{lem:packability_invariant}
    The packed items and the reserve together always fit into the knapsack; that is, $p_i + r_i \leq 1$, for every $0 \leq i \leq n$.
\end{restatable}

\begin{restatable}[Critical Rejection]{lemma}{LemmaCriticalRejection}\label{lem:critical_rejection}
    At most one item rejected by \algsymb is packed by \Opt.
\end{restatable}

\begin{restatable}[Safe Rejection]{lemma}{LemmaSafeRejection}\label{lem:safe_rejection}
    Assume \algsymb rejects an item $x_i$ that is packed by \Opt. If $\pcurr \geq \pwin{\rcurr}$, then \algsymb achieves a fill ratio of $f$ on $I$.
\end{restatable}

We first apply the three lemmas to prove \cref{thm:alg_achieves_f}.

\begin{proof}[of \cref{thm:alg_achieves_f}]
    We may assume that \algsymb does not win early  as that would already guarantee the desired fill ratio.
    Then \cref{lem:packability_invariant} gives the packability invariant, $p_i + r_i \leq 1$, at every step; in particular $p_n + r_n \leq 1$, so the closing step (\Lineref{line:packall}) in which \algsymb packs its entire reserve is feasible. 
    Since \algsymb does not win early, it does not remove anything, so the final packing has size $p_n + r_n$ and incurs only reservation cost $\alpha r_n$, for a total gain of $p_n + (1 - \alpha) r_n$.

    Suppose first that $\Opt(I)$ does not contain any items rejected by \algsymb. 
    Then, every item of $\Opt(I)$ is packed or reserved by \algsymb; we write $\Popt$ and $\Ropt$ for its packed and reserved parts, such that $\opt(I) = \popt + \ropt$.
    As $\Popt \subseteq P_n$ and $\Ropt \subseteq R_n$, \algsymb has a gain of
    \[
        p_n + (1 - \alpha) r_n \geq \popt + (1 - \alpha) \ropt \geq (1 - \alpha) (\popt + \ropt) \geq f \cdot \opt(I),
    \]
    because $f \leq 1 - \alpha$.

    Otherwise, some rejected item is contained in $\Opt(I)$, and by \cref{lem:critical_rejection}, there is exactly one such item $x_i$. 
    As Phase 1 rejects nothing, $x_i$ was rejected in Phase 2. 
    Let $x_s$ be the item that triggered the transition to Phase 2: it failed the Phase-1 test in \Lineref{line:phase1check}, so $x_s \geq \pwin{r_{s - 1}}$, and since \algsymb does not win early and its knapsack is empty when $x_s$ arrives, $x_s$ is packed greedily. 
    Phase~2 makes no reservations and, as \algsymb does not win early, removes nothing; thus the reserve stays fixed and the packed size never decreases from the arrival of $x_s$ on.
    Hence, immediately before $x_i$ is rejected, we have $\pcurr \geq \pwin{\rcurr}$. \cref{lem:safe_rejection} therefore applies and certifies that this rejection is safe, so \algsymb attains its target fill ratio $f$ in this case as well.
\end{proof}

We now sketch the proofs of \cref{lem:packability_invariant,lem:critical_rejection,lem:safe_rejection}; the full proofs can be found in \cref{sec:proofs}.

\begin{proof}[of \cref{lem:packability_invariant}, sketch]
    In Phase~1, nothing is packed, so the invariant reduces to $r_i \leq 1$. The safety threshold is designed so that $r + \pwin{r}$ is non-decreasing in the reserve size $r$. Since every reserved item is smaller than the current threshold, the reserve can never overtake a certain value $\rwin \leq 1$ at which the threshold drops to $0$. For Phase~2, we argue by induction and consider a greedily packed item $x_i$. Ranging over all subsets $R'$ of the reserve, the condition of \Lineref{line:checkwin} accepts every item size in the interval $[f + \alpha \rcurr - r' - \pcurr,\, 1 - r' - \pcurr]$; since reserved items are small (they passed the Phase-1 test), a combinatorial argument shows that these intervals chain together and jointly cover $[f - (1 - \alpha) \rcurr - \pcurr,\, 1 - \pcurr]$ without gaps. Since \algsymb does not win, the packed item is outside of this range; and since it fits, it must lie below it; hence $p_i + r_i \leq f + \alpha \rcurr \leq 1$. The full proof is deferred to \cref{sub:packability}.
\end{proof}

\begin{proof}[of \cref{lem:critical_rejection}, sketch]
    We show that every item $x_i$ rejected by \algsymb satisfies $x_i > \frac{1}{2}$; since two such items never fit together, \Opt packs at most one of them. \algsymb rejects only in \Lineref{line:reject}, so a rejected $x_i$ did not fit, \ie, $x_i > 1 - \pcurr$; as $f \leq \frac{1}{2}$ by \cref{def:f}, it suffices to show $\pcurr \leq f$. Consider the last item packed before the rejection: packing it together with the entire reserve was feasible by \cref{lem:packability_invariant}. Since \algsymb does not win early, the win condition in \Lineref{line:checkwin} must have failed in its lower bound, which yields $\pcurr < f$. The full proof is deferred to \cref{subsec:crit_reject}.
\end{proof}

\begin{proof}[of \cref{lem:safe_rejection}, sketch]
    Since \algsymb does not win early, it packs its reserve at the end, and in Phase~2 the reserve is frozen; so its gain is $\pcurr + \plater + (1 - \alpha) \rcurr$, where $\Plater$ denotes the items packed after the rejection of $x_i$. By \cref{lem:critical_rejection}, all other items of $\Opt(I)$ are packed or reserved by \algsymb, so $\opt(I) \leq \popt + \plater + x_i + \ropt$ for the parts $\Popt \subseteq \Pcurr$ and $\Ropt \subseteq \Rcurr$ that \Opt takes from the first $i - 1$ items. Dropping the common term $\plater$ only decreases the fill ratio, so it suffices to show $\pcurr + (1 - \alpha) \rcurr \geq (\popt + x_i + \ropt) \cdot f$. Since \Opt packs $\Popt$, $x_i$, and $\Ropt$ together, the win condition in \Lineref{line:checkwin} satisfies its upper bound for $P' = \Popt$ and $R' = \Ropt$. As \algsymb did not win on $x_i$, the lower bound of \Lineref{line:checkwin} must have failed, giving $\popt + x_i + \ropt < f + \alpha \rcurr + \beta \pcurr$. Substituting this and solving for $\pcurr$ (using $1 - \beta f > 0$) turns the sufficient inequality into exactly $\pcurr \geq \pwin{\rcurr}$, which holds by assumption; indeed, the safety threshold \cref{def:pwin} is precisely the solution of this inequality. The full proof is deferred to \cref{subsec:safe_reject}.
\end{proof}

\subsection{Lower Bound}\label{sub:symbiosis:lower_bound}

We complement \cref{thm:alg_achieves_f} with matching lower bounds. Recall from \cref{def:f} that the target fill ratio is $f(\alpha, \beta) = \min(\fo(\alpha, \beta), \frac{1}{2}, 1 - \alpha)$; accordingly, we match each branch of the minimum by a separate lower bound: \cref{thm:globaltwo} matches the branch $f(\alpha, \beta) = \frac{1}{2}$, \cref{thm:lowerboundsymbiosistrivial} the branch $f(\alpha, \beta) = 1 - \alpha$, and \cref{thm:lowerboundf} the branch $f(\alpha, \beta) = \fo(\alpha, \beta)$.

The first bound holds on the entire positive quadrant, so we will reuse it in \cref{sec:dominance}. Its short proof is deferred to \cref{sub:lb_symbiosis_proofs}.

\begin{restatable}{theorem}{thmGlobalTwo}\label{thm:globaltwo}
    For any $\alpha>0$ and $\beta>0$, no algorithm can have a better competitive ratio than $2$.
\end{restatable}

The two remaining bounds are tailored to $\Usymb$. Their proofs are adversarial decision trees: at each node, the adversary reacts to the algorithm's choice among packing, reserving, and rejecting with the next item or by ending the instance, and every leaf is shown to force a competitive ratio of at least the claimed bound. The proof of \cref{thm:lowerboundsymbiosistrivial} is a tree on three items (\cref{fig:lbsymtriv}); we provide a proof sketch here as a compact template for this technique and defer the full proof to \cref{sub:lb_symbiosis_proofs}.

\begin{figure}
    \centering
    \begin{tikzpicture}[
    font=\large,
    >=Latex
]
    \def\rejdist{3cm}
    \def\resdist{4.5cm}
    
    \node[itemsettings] (x1)
        {$x_1 \da \frac{1-\alpha}{2-\alpha}$};
    \node[itemsettings, ifpack=x1] (x2)
        {$x_2\da 1-x_1+\eps$};
    \node[itemsettings, ifpack=x2] (x3)
        {$x_3\da 1-x_1$};

    \node[leafsettings,
          ifreject=x1] (x1rej)
        {$c = \infty$};
    \node[leafsettings,
          ifreserve=x1] (x1res)
        {$c = \frac{1}{1-\alpha}$};
    \node[leafsettings,
          ifreject=x2] (x2rej)
        {$c = \frac{x_2}{x_1}$};
    \node[leafsettings,
          ifreserve=x2] (x2res)
        {$c = \min\Big(\frac{x_2}{x_1-\alpha\cdot x_2},\frac{x_2}{x_2-\beta\cdot x_2-\alpha\cdot x_1}\Big)$};
    \node[leafsettings, 
            ifreject=x3] (x3rej)
        {$c=\frac{1}{x_2-\beta\cdot x_1}$};

    \draw[rejarr] (x1.west) -| (x1rej.north);
    \draw[packarr] (x1.south) -- (x2.north);
    \draw[resarr] (x1.east) -| (x1res.north);
    \draw[rejarr] (x2.west) -| (x2rej.north);
    \draw[packarr] (x2.south) -- (x3.north);
    \draw[resarr] (x2.east) -| (x2res.north);
    \draw[rejarr] (x3.west) -| (x3rej.north);
    \draw[rejarr] (x2.west) -| (x2rej.north);
\end{tikzpicture}
    \caption{Structure of the lower bound analyzed in \cref{thm:lowerboundsymbiosistrivial}. Downward arrows (green) represent the choice to pack an item, while leftward arrows (red, dashed) and rightward arrows (blue, dotted), respectively, represent rejecting or reserving an item. }
    \label{fig:lbsymtriv}
\end{figure}
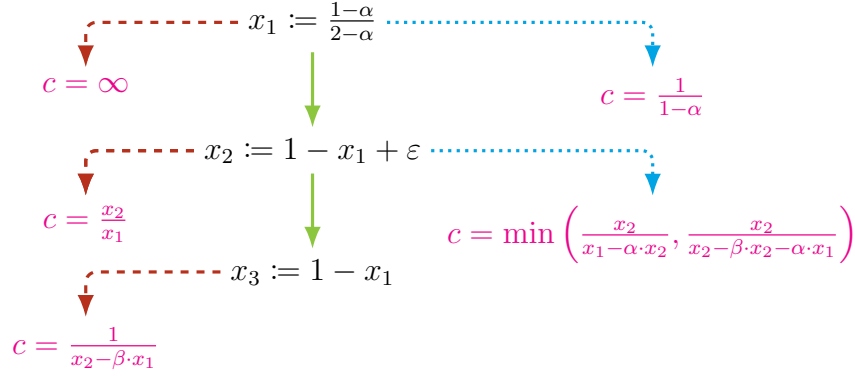

\begin{restatable}{theorem}{lowerBoundSymbiosisTrivial}\label{thm:lowerboundsymbiosistrivial}
    For any $(\alpha,\beta)$ in $\Usymb$,  no algorithm can have a better competitive ratio than $1/(1-\alpha)$.
\end{restatable}
\begin{proof}[sketch]
    Consider any deterministic algorithm \alglower. The structure of the proof is shown in \cref{fig:lbsymtriv}. \alglower first receives an item $x_1=(1-\alpha)/(2-\alpha)<1/2$. If \alglower rejects $x_1$, its competitive ratio is unbounded. If \alglower reserves $x_1$, the instance stops and its competitive ratio is at least $x_1/(x_1-\alpha\cdot x_1)=1/(1-\alpha)$.
    
    We can therefore assume that \alglower packs $x_1$. It then receives an item $x_2=1-x_1+\eps>x_1$, for some sufficiently small $\eps>0$.
    If \alglower reserves $x_2$, the instance stops. Whatever \alglower does, its gain is then less than $x_2-\alpha\cdot x_2$, so its competitive ratio is greater than $x_2/(x_2-\alpha\cdot x_2)=1/(1-\alpha)$.

    If \alglower rejects $x_2$, its competitive ratio is $x_2/x_1>1/(1-\alpha)$. We can therefore assume that \alglower packs $x_2$ and, since $x_1+x_2>1$, removes $x_1$. It then receives a final item $x_3=1-x_1$ completing $\Opt(I)$ to $1$. Since no items arrive afterwards, reserving $x_3$ is clearly suboptimal. Since $x_3<x_2$, there is also no benefit in removing $x_2$. \alglower's competitive ratio is therefore at least
    \begin{equation}\label{eq:simpleboundsketch}
    \frac{1}{x_2-\beta\cdot x_1}=\frac{1}{1-(1+\beta)\cdot x_1+\eps}>\frac{1}{1-\alpha+\eps}\xrightarrow{\eps\to\, 0}\frac{1}{1-\alpha}\;,
    \end{equation}
    since the fact that $(1+\beta)\cdot x_1>\alpha$ follows from the definition of $\Usymb$ in \cref{def:Usymb-f}.
\end{proof}

\cref{thm:lowerboundf} is the technical heart of the lower-bound side. The proof considers instances of up to four items and branches on all possible decisions an algorithm can take; the item sizes are tuned so that the defining identity \cref{eq:fo} of $\fo$ makes the decisive leaves evaluate to exactly $1/\fo$ in the limit. We state the result here and defer the full case analysis, together with the tree (\cref{fig:lbsym}), to \cref{sub:lb_symbiosis_proofs}.

\begin{restatable}{theorem}{thmLowerBoundf}\label{thm:lowerboundf}
    For any $(\alpha,\beta)\in\Usymb$ satisfying $\beta>1-\alpha$, no algorithm can have a better competitive ratio than $1/\fo(\alpha,\beta)$.
\end{restatable}

Together with \cref{thm:alg_achieves_f}, these three bounds establish $\cknapresrem(\alpha, \beta) = 1/f(\alpha, \beta)$ on $\Usymb$; the formal assembly is part of the proof of \cref{thm:results} in \cref{sec:assembly}.

\section{Dominance}\label{sec:dominance}

In this section, we discuss that, outside of \Usymb, the equilibrium border is sharp: on $\Ures \setminus \Usymb$, a reservation-only algorithm is already optimal, while on $\Urem \setminus \Usymb$, a removal-only algorithm is.

As reservation and removal can always be left unused, an algorithm for \knapresrem can in particular simulate any algorithm that uses only one of the two mechanisms. This gives us a trivial upper bound:

\begin{observation}\label{obs:single_mechanism_upper}
    For every $(\alpha, \beta)$, we have $\cknapresrem(\alpha, \beta) \leq \min (\cknapres(\alpha), \cknaprem(\beta))$.
\end{observation}

We show successively that this bound is tight for various ranges of $\alpha$ and $\beta$.
In those regions where a reservation-only algorithm or a removal-only algorithm already achieves a competitive ratio of $2$, we have the corresponding lower bound of \cref{thm:globaltwo}. The other results are more involved. We defer their proofs to \cref{sec:lower_dominance}. 
The first two, \cref{thm:lowerbound_reservation_small,thm:lowerbound_reservation_medium}, are based on the lower bounds for the reservation-only case by \citefull{BBFHLR22} and show that these bounds still hold as long as $\beta$ is not too small compared to $\alpha$:

\begin{restatable}{theorem}{theoremLowerReservationSmall}\label{thm:lowerbound_reservation_small}
    For $1/4<\alpha\leq \sqrt{2}-1$ and $\beta\geq \tilde{\beta}(\alpha)\coloneq\frac{\alpha\cdot (3+\sqrt{5-4\alpha}-2\alpha)}{2(1-\alpha)}$, no algorithm can have a better competitive ratio than 
    \[\lowerbound\da\frac{1+\sqrt{5-4\alpha}}{2(1-\alpha)}\;.\]
\end{restatable}

\begin{restatable}{theorem}{theoremLowerReservationMedium}\label{thm:lowerbound_reservation_medium}
    For $\sqrt{2} - 1 < \alpha \leq \phi - 1$ and $\beta \geq 1 + \alpha$, no algorithm can have a better competitive ratio than $\ell\da 2 + \alpha$.
\end{restatable}

The final bound, \cref{thm:lower-equib}, covers a further large region, on which either the reservation-only algorithm for $\alpha\geq\phi-1$ by \citefull{BBFHLR22} or the removal-only algorithm for $\beta\geq 1/2$ by \citefull{HKM2014} is already sufficient to achieve the best-possible competitive ratio.

\begin{restatable}{theorem}{theoremLowerEquilibrium}\label{thm:lower-equib}
    For $0 < \alpha < 1$ and $\beta \geq 1/2$, no algorithm can have a better competitive ratio than 
    \begin{align*}
        \ell\da \min \left( \frac{1}{1 - \alpha}, \frac{1 + \beta + \sqrt{\beta^2 + 2\beta + 5}}{2} \right).
    \end{align*}
\end{restatable}

\section{Exploring the Symbiosis Region}\label{sec:proofs}

This section supplies the proofs omitted from \cref{sec:symbiosis}. \Cref{sub:fill_ratio} collects the elementary properties of the fill ratio $f$ on which the analysis relies, \cref{sub:packability} proves the packability invariant (\cref{lem:packability_invariant}) at the core of the upper bound, \cref{subsec:crit_reject} proves the critical rejection assertion (\cref{lem:critical_rejection}), \cref{subsec:safe_reject} proves the safe rejection invariant (\cref{lem:safe_rejection}), and \cref{sub:lb_symbiosis_proofs} then proves the three lower bounds stated in \cref{sub:symbiosis:lower_bound}.

\subsection{Properties of the Fill Ratio}\label{sub:fill_ratio}

Throughout, we regard $f$, $\fo$, and $\pwin{\cdot}$ as functions of $(\alpha, \beta)$ on $\Usymb$, recalled from \cref{def:Usymb,def:pwin,def:f,def:fo}. We first record a few facts, all valid at every $(\alpha, \beta) \in \Usymb$. Recall from \cref{def:fo} the definition of the polynomial
\[
    Q(x) = Q_{\alpha, \beta}(x) = (1 + \alpha - \beta) x^2 + (1 + \beta) x - 1\;.
\]
Since $(\alpha, \beta) \in \Usymb$, we have $\beta < 1 + \alpha$, \ie,
\begin{equation}
    1 + \alpha - \beta > 0\;, \label{eq:1+a-b>0}
\end{equation}
so $Q$ opens upward; since $Q(0) = -1 < 0$ and $Q(1) = 1 + \alpha > 0$, it has a unique root in $(0, 1)$, namely $\fo$. Hence
\begin{equation}
    0 < \fo < 1 \quad\text{and}\quad Q(\fo) = (1 + \alpha - \beta) \fo^2 + (1 + \beta) \fo - 1 = 0\;.\label{eq:0<fo<1_Q(fo)=0}
\end{equation}
Due to \cref{eq:1+a-b>0} and \cref{eq:0<fo<1_Q(fo)=0}, we have
\[
    (1 + \beta) \fo = 1 - (1 + \alpha - \beta) \fo^2 \leq 1\;,
\]
and thus $\beta \fo < (1 + \beta) \fo \leq 1$, which implies
\begin{equation}\label{eq:1-bfo>0}
    1 - \beta \fo > 0\;.
\end{equation}
So, rearranging \cref{eq:0<fo<1_Q(fo)=0} yields
\begin{equation}
    \frac{\fo^2}{1 - \beta \fo} = \frac{1 - \fo}{1 + \alpha}\;.\label{eq:fo_identity}
\end{equation}
This proves the identity \cref{eq:fo} and the positivity claim made after \cref{def:fo} in \cref{sec:introduction}.
Finally, since $f \leq \fo$ by \cref{def:f} and $f \leq 1 - \alpha$, we have
\begin{equation}
    1 - \beta f \geq 1 - \beta \fo > 0 \quad\text{and}\quad 1 - \alpha - \alpha f \geq (1 - \alpha)^2 > 0\;. \label{eq:1-bf>0_1-a-af>=(1-a)^2>0}
\end{equation}

\begin{lemma}\label{lemma:betagreateralpha}
    For any $(\alpha,\beta)\in \Usymb$, it holds that $\beta > \alpha$.
\end{lemma}
\begin{proof}
    Recall from the definition of $\Usymb$ in \cref{def:Usymb-f} that 
    \[\beta > \max\left(\frac{1}{2},\frac{3\alpha-\alpha^2-1}{1-\alpha}\right).\]
    If $\alpha\leq 1/2$, then clearly $\beta > 1/2 \geq \alpha$. If, however, $\alpha>1/2$, then from 
    \begin{equation*}
    \alpha>\frac{1}{2}\implies 2\alpha-1>0\implies 3\alpha-\alpha^2-1> \alpha-\alpha^2\implies \frac{3\alpha-\alpha^2-1}{1-\alpha}>\alpha\;,
    \end{equation*}
    it follows that $\beta>\alpha$.
\end{proof}

\begin{lemma}\label{lem:f_non-increasing}
    The function $f(\alpha, \beta)$ is non-increasing in $\alpha$ and $\beta$ on $\Usymb$ and $\fo(\alpha, \beta)$ is strictly decreasing in $\alpha$ and $\beta$, for $\beta\leq 1+\alpha$.
\end{lemma}
\begin{proof}
    By \cref{def:f}, it suffices to show that each of $\fo$, $\frac{1}{2}$, and $1 - \alpha$ is non-increasing in $\alpha$ and $\beta$. For the latter two this is clear, so consider $\fo$ for which we even prove that it is strictly decreasing. Let $\fderiva$, $\fderivb$ denote its partial derivatives. It suffices to show that they are non-positive, which we do via \cref{eq:0<fo<1_Q(fo)=0}. Differentiating both sides of the equation $Q(\fo) = 0$ with respect to $\alpha$ gives
    \[
        \fo^2 + 2\fo \cdot (1 + \alpha - \beta) \fderiva + (1 + \beta) \fderiva = 0 \iff \fderiva = - \frac{\fo^2}{2\fo \cdot (1 + \alpha - \beta) + 1 + \beta}\;,
    \]
    which is strictly negative by \cref{eq:1+a-b>0}. Differentiating with respect to $\beta$ yields
    \[
        - \fo^2 + 2\fo \cdot (1 + \alpha - \beta) \fderivb + \fo + (1 + \beta) \fderivb = 0 \iff \fderivb = \frac{\fo^2 - \fo}{2\fo \cdot (1 + \alpha - \beta) + 1 + \beta}\;,
    \]
    which is strictly negative by \cref{eq:1+a-b>0} and since $\fo^2 - \fo < 0$ by \cref{eq:0<fo<1_Q(fo)=0}.
\end{proof}

\begin{lemma}\label{lem:f_piecewise}
    For $(\alpha, \beta) \in \Usymb$, we have
    \[
        f = \begin{cases}
                \fo, &\text{for } \beta > \max \left( 1 - \alpha, \alpha + \frac{2\alpha - 1}{\alpha\cdot(1 - \alpha)} \right), \\
                \frac{1}{2}, &\text{for } \beta \leq 1 - \alpha\;,\\
                1 - \alpha, &\text{for } \beta \leq \alpha + \frac{2\alpha - 1}{\alpha\cdot(1 - \alpha)}\;.
            \end{cases}
    \]
\end{lemma}
\begin{proof}
    The lemma identifies, for \cref{def:f}, the regions of $\beta$ in which each term attains the minimum. We first show $\fo < \frac{1}{2} \iff \beta > 1 - \alpha$. By \cref{def:fo} and \cref{eq:1+a-b>0}, we have 
    \begin{align*}
        \fo < \frac{1}{2} &\iff - (1 + \beta) + \sqrt{(1 + \beta)^2 + 4(1 + \alpha - \beta)} < 1 + \alpha - \beta \\
            &\iff \sqrt{(1 + \beta)^2 + 4(1 + \alpha - \beta)} < 2 + \alpha \\
            &\iff (1 + \beta)^2 + 4(1 + \alpha - \beta) < 4 + 4\alpha + \alpha^2 \\
            &\iff (1 - \beta)^2 < \alpha^2 \iff |1 - \beta| < \alpha \iff (1 - \beta < \alpha) \lor (1 - \beta < -\alpha)\;.
    \end{align*}
    Since $1 - \beta > -\alpha$ by \cref{def:Usymb}, we get
    \begin{equation}
        \fo < \frac{1}{2} \iff \beta > 1 - \alpha\;. \label{eq:fo<1/2_iff_b>1-a}
    \end{equation}

    Next we show $\fo < 1 - \alpha \iff \beta > \alpha + \frac{2\alpha - 1}{\alpha\cdot (1 - \alpha)}$. As $\fo$ is the unique positive root of the upward quadratic $Q$, this is equivalent to $Q(1 - \alpha) > 0$. We have
    \begin{align*}
        &Q(1 - \alpha) = (1 + \alpha - \beta)(1 - \alpha)^2 + (1 + \beta)(1 - \alpha) - 1 > 0 \\
        &\quad\iff \beta \alpha\cdot (1 - \alpha) > - \alpha^3 + \alpha^2 + 2\alpha - 1 \iff \beta > \alpha + \frac{2\alpha - 1}{\alpha\cdot (1 - \alpha)}\;.
    \end{align*}
    Together with \cref{eq:fo<1/2_iff_b>1-a,def:f}, this gives $f = \fo$ for $\beta > \max \left( 1 - \alpha, \alpha + \frac{2\alpha - 1}{\alpha\cdot(1 - \alpha)} \right)$.

    Consider $\beta \leq 1 - \alpha$. By \cref{def:Usymb}, we have $\beta > \frac{1}{2}$, which gives $\frac{1}{2} < 1 - \alpha$. So with \cref{eq:fo<1/2_iff_b>1-a}, we have $f = \frac{1}{2}$ for $\beta \leq 1 - \alpha$.
    
    Finally, consider $\beta \leq \alpha + \frac{2\alpha - 1}{\alpha\cdot (1 - \alpha)}$. Again using $\beta > \frac{1}{2}$, this gives
    \[
        \frac{1}{2} \leq \alpha + \frac{2\alpha - 1}{\alpha\cdot (1 - \alpha)}\;,
    \]
    which is equivalent to
    \[
        0 \leq \left( \alpha - \frac{1}{2} \right)\cdot \left( 1 + \frac{2}{\alpha\cdot(1 - \alpha)} \right),
    \]
    where the right-hand factor is positive. Hence $\alpha \geq \frac{1}{2}$, \ie, $1 - \alpha \leq \frac{1}{2}$. So, we have $f = 1 - \alpha$ for $\beta \leq \alpha + \frac{2\alpha - 1}{\alpha\cdot(1 - \alpha)}$, which concludes the proof.
\end{proof}
Note that, on $\Usymb$, the last two cases never overlap: since $\beta > \frac{1}{2}$, the condition $\beta \leq 1 - \alpha$ forces $\alpha < \frac{1}{2}$, while $\beta \leq \alpha + \frac{2\alpha - 1}{\alpha\cdot(1 - \alpha)}$ forces $\alpha > \frac{1}{2}$; at $\alpha = \frac{1}{2}$ itself, both conditions would require $\beta \leq \frac{1}{2}$, which is excluded.

\begin{lemma}\label{lem:f_strictly_better}
    For all $(\alpha, \beta) \in \Usymb$, we have
    \[
        \cknapresrem(\alpha, \beta) < \min (\cknapres(\alpha), \cknaprem(\beta))\;.
    \]
\end{lemma}
\begin{proof}
    At every parameter pair $(\alpha, \beta)$, the competitive ratio $\cknapresrem(\alpha, \beta)$ is bounded from above by both single-mechanism ratios $\cknapres(\alpha)$ and $\cknaprem(\beta)$, since an algorithm may ignore either mechanism. Hence
    \begin{equation}
        \cknapresrem(\alpha, \beta) \leq \min \left( \cknapres(\alpha), \cknaprem(\beta) \right). \label{eq:cresrem<=cres_crem}
    \end{equation}

    Now fix $(\alpha, \beta) \in \Usymb$. By \cref{thm:results}, $\cknapresrem(\alpha, \beta) = \frac{1}{f(\alpha, \beta)}$, and by \cref{def:f}, $f(\alpha, \beta) = \min (\fo(\alpha, \beta), \frac{1}{2}, 1 - \alpha)$, hence
    \[
        \cknapresrem(\alpha, \beta) = \frac{1}{f(\alpha, \beta)} = \max \left( \frac{1}{\fo(\alpha, \beta)}, 2, \frac{1}{1 - \alpha} \right).
    \]
    We bound each of the three terms strictly below both $\cknapres(\alpha)$ and $\cknaprem(\beta)$.

    \begin{description}
        \item[\textnormal{\textit{The term $2$.}}] We know that $\cknapres(\alpha) = 2$ on $[0, \frac{1}{4}]$ and it strictly increases for $\alpha \geq \frac{1}{4}$. Since $\alpha > \frac{1}{4}$ on $\Usymb$, we thus have $2 < \cknapres(\alpha)$. Likewise, $\cknaprem(\beta) = 2$ on $[0, \frac{1}{2}]$ and it strictly increases for $\beta \geq \frac{1}{2}$. Since $\beta > \frac{1}{2}$ on $\Usymb$, we thus have $2 < \cknaprem(\beta)$.
    
        \item[\textnormal{\textit{The term $\frac{1}{1 - \alpha}$ versus $\cknapres(\alpha)$.}}] For $\alpha \in [\frac{1}{4}, \sqrt{2} - 1]$, we have $\cknapres(\alpha) = \frac{1 + \sqrt{5 - 4\alpha}}{2(1 - \alpha)}$. So, since $\frac{1 + \sqrt{5 - 4\alpha}}{2} > 1$, we have
        \[
            \frac{1}{1 - \alpha} < \frac{1 + \sqrt{5 - 4\alpha}}{2(1 - \alpha)} = \cknapres(\alpha)\;.
        \]
        For $\alpha \in [\sqrt{2} - 1, \varphi - 1]$, we have $\cknapres(\alpha) = 2 + \alpha$. So, we need to show $\frac{1}{1 - \alpha} < 2 + \alpha$, which is equivalent to $\alpha^2 + \alpha - 1 < 0$. The latter is true for $-\varphi < \alpha < \varphi - 1$, and thus on $\Usymb$.
    
        \item[\textnormal{\textit{The term $\frac{1}{1 - \alpha}$ versus $\cknaprem(\beta)$.}}] Since $\beta > \frac{1}{2}$ in $\Usymb$, we have $\cknaprem(\beta) = \frac{1 + \beta + \sqrt{\beta^2 + 2\beta + 5}}{2}$, which is the positive root of the upward quadratic
        \[
            C(x) \da x^2 - (1 + \beta) x - 1\;.
        \]
        So, to prove $\frac{1}{1 - \alpha} < \cknaprem(\beta)$, it suffices to show $C(\frac{1}{1 - \alpha}) < 0$, that is
        \[
            \left( \frac{1}{1 - \alpha} \right)^2 - (1 + \beta) \frac{1}{1 - \alpha} - 1 < 0\;.
        \]
        Solving this for $\beta$ gives $\beta > \frac{- \alpha^2 + 3\alpha - 1}{1 - \alpha}$, which holds on $\Usymb$.
    
        \item[\textnormal{\textit{The term $\frac{1}{\fo}$.}}] By \cref{lem:f_non-increasing}, $\fo$ is strictly decreasing in $\alpha$ and $\beta$, so $\frac{1}{\fo}$ is strictly increasing in each. Suppose $\cknapres(\alpha) \leq \cknaprem(\beta)$; the case $\cknaprem(\beta) \leq \cknapres(\alpha)$ is analogous. As $\Usymb$ is open and $\cknapres$ does not depend on $\beta$, we can choose $\varepsilon > 0$ with $(\alpha, \beta + \varepsilon) \in \Usymb$. Then, using that $\frac{1}{\fo(\alpha, \beta)} \leq \max \left( \frac{1}{\fo(\alpha, \beta)}, 2, \frac{1}{1 - \alpha} \right) = \cknapresrem(\alpha, \beta)$, we have
        \[
            \frac{1}{\fo(\alpha, \beta)} < \frac{1}{\fo(\alpha, \beta + \varepsilon)} \leq \cknapresrem(\alpha, \beta + \varepsilon) \leq \cknapres(\alpha) \leq \cknaprem(\beta)\;,
        \]
        where the third inequality holds by \cref{eq:cresrem<=cres_crem}.
    \end{description}
    Each of the three terms is therefore strictly below both $\cknapres(\alpha)$ and $\cknaprem(\beta)$, so their maximum is as well: $\cknapresrem(\alpha, \beta) < \min \left( \cknapres(\alpha), \cknaprem(\beta) \right)$, for all $(\alpha, \beta) \in \Usymb$.
\end{proof}

\subsection{The Packability Invariant}\label{sub:packability}

In this subsection, we prove \cref{lem:packability_invariant}, which we restate here:
\LemmaPackability*
We prove this for each of the two phases of \algsymb. The reserve
\begin{equation}
    \rwin \da \frac{f^2}{1 - \alpha - \alpha f}\;, \label{def:rwin}
\end{equation}
for which the safety threshold $\pwin{\cdot}$ becomes $0$, \ie, $\pwin{\rwin} = 0$, will be useful in the following proofs.

Let us begin by proving the packability invariant in Phase 1, in which no items get packed by \algsymb, reducing the invariant to $r_i \leq 1$. We need the following technical lemma.

\begin{lemma}\label{lem:r+pwin(r)_increasing_in_r}
    For $(\alpha, \beta) \in \Usymb$, the term $r + \pwin{r}$ is non-decreasing in $r$.
\end{lemma}
\begin{proof}
    By \cref{def:pwin}, we have
    \[
        r + \pwin{r} = \frac{f^2 - (1 - \alpha - \alpha f) r + (1 - \beta f) r}{1 - \beta f} = \frac{f^2 + (\alpha + \alpha f - \beta f) r}{1 - \beta f}\;.
    \]
    Since $1 - \beta f > 0$ by \cref{eq:1-bf>0_1-a-af>=(1-a)^2>0}, it suffices to prove
    \begin{equation}
        \alpha + \alpha f - \beta f \geq 0. \label{suf:a+af-bf>=0}
    \end{equation}
    We distinguish three cases according to \cref{lem:f_piecewise}.
    \begin{casedist}
        \case Assume $\beta \leq 1 - \alpha$, so we have $f = \frac{1}{2}$. Then \cref{suf:a+af-bf>=0} simplifies to $3\alpha \geq \beta$ and, since $\beta \leq 1 - \alpha$, it suffices that $3\alpha \geq 1 - \alpha$, \ie, $\alpha \geq \frac{1}{4}$, which holds in $\Usymb$.
        \case Assume $\beta \leq \alpha + \frac{2 \alpha - 1}{\alpha\cdot(1 - \alpha)}$, so we have $f = 1 - \alpha$. Then \cref{suf:a+af-bf>=0} is equal to $\alpha + \alpha\cdot(1 - \alpha) - \beta\cdot(1 - \alpha) \geq 0$. Since $\beta \leq \alpha + \frac{2\alpha - 1}{\alpha\cdot(1 - \alpha)}$, we have
        \begin{align*}
            \alpha + \alpha\cdot(1 - \alpha) - \beta\cdot(1 - \alpha) &\geq \alpha + \alpha\cdot(1 - \alpha) - \left( \alpha + \frac{2\alpha - 1}{\alpha\cdot(1 - \alpha)} \right) (1 - \alpha) \\
                &= \alpha - \frac{2\alpha - 1}{\alpha} = \frac{(\alpha - 1)^2}{\alpha} \geq 0\;,
        \end{align*}
        which was to be proven.
        \case Assume $\beta > \max \left( 1 - \alpha, \alpha + \frac{2\alpha - 1}{\alpha\cdot(1 - \alpha)} \right)$, so we have $f = \fo$. Since $\beta > \alpha$ by \cref{lemma:betagreateralpha}, then \cref{suf:a+af-bf>=0} simplifies to $\frac{\alpha}{\beta - \alpha} \geq \fo$. As $\fo$ is the unique positive root of the upward quadratic $Q$ from \cref{def:fo}, with $Q(0) = -1 < 0$, it suffices that $Q \left( \frac{\alpha}{\beta - \alpha} \right) \geq 0$. We have
        \begin{align*}
            Q \left( \frac{\alpha}{\beta - \alpha} \right) &= \frac{(1 + \alpha - \beta) \alpha^2 + (1 + \beta)(\beta - \alpha) \alpha - (\beta - \alpha)^2}{(\beta - \alpha)^2} \\
                &= \frac{- \alpha^2 + \alpha^3 - 2\alpha^2 \beta + 3\alpha \beta + \alpha \beta^2 - \beta^2}{(\beta - \alpha)^2}\;.
        \end{align*}
        So, it suffices that $- \alpha^2 + \alpha^3 - 2\alpha^2 \beta + 3\alpha \beta + \alpha \beta^2 - \beta^2 \geq 0$, that is,
        \[
            (1 - \alpha)\beta^2 + (2\alpha^2 - 3\alpha)\beta + \alpha^2 - \alpha^3 \leq 0\;.
        \]
        The left-hand site is an upward quadratic in $\beta$, hence nonpositive exactly between its roots, which are, by the quadratic formula,
        \[
            \beta_\pm = \frac{3\alpha - 2\alpha^2 \pm \alpha \sqrt{5 - 4\alpha}}{2(1 - \alpha)} = \frac{\alpha \cdot (3 - 2\alpha \pm \sqrt{5 - 4\alpha})}{2(1 - \alpha)}\;.
        \]
        So, it suffices to prove $\beta_- \leq \beta \leq \beta_+$. Since $(\alpha, \beta) \in \Usymb$, we have $\beta < \frac{\alpha \cdot (3 - 2\alpha + \sqrt{5 - 4\alpha})}{2(1 - \alpha)} = \beta_+$, proving the second inequality. Again, since $(\alpha, \beta) \in \Usymb$, we have $\beta > \frac{1}{2}$, so to prove the first inequality, it suffices to show $\beta_- \leq \frac{1}{2}$, which simplifies to
        \[
            0 \leq 1 - 4\alpha + 2\alpha^2 + \alpha \sqrt{4 - 4\alpha + 1}\;.
        \]
        Since $4 - 4\alpha + 1 \geq 4 - 4\alpha + \alpha^2 = (2 - \alpha)^2$, we have
        \[
            1 - 4\alpha + 2\alpha^2 + \alpha \sqrt{4 - 4\alpha + 1} \geq 1 - 4\alpha + 2\alpha^2 + \alpha (2 - \alpha) = 1 - 2\alpha + \alpha^2 = (1 - \alpha)^2 \geq 0\;,
        \]
        which was to be proven.
    \end{casedist}
\end{proof}

The safety threshold was chosen so that the reserve can never grow past $\rwin$; the next lemma makes this precise.
\begin{lemma}\label{lem:ri<=rwin<=1}
    Let $(\alpha, \beta) \in \Usymb$ and assume \algsymb does not win early on an instance $I = (x_1, \dots, x_n)$. Then, at every step $i$, we have
    \[
        r_i \leq \rwin \leq 1\;.
    \]
\end{lemma}
\begin{proof}
    We begin by proving $\rwin \leq 1$. Because of \cref{def:rwin}, this simplifies to
    \[
        (f + \alpha) f \leq 1 - \alpha\;.
    \]
    Since $f \leq 1 - \alpha$ by \cref{def:f}, we have $f + \alpha\leq 1$, so $(f + \alpha) f \leq f \leq 1 - \alpha$, as needed.

    Next, we prove that $r_i \leq \rwin$ for every step $i$. Let $x_s$ be the item that triggers \algsymb to switch from Phase 1 to Phase 2, defining $s = n + 1$ if no switch occurs. In particular, every item before $x_s$ is reserved by \algsymb and no others. \algsymb does not win early, so from time step $s$ on, the reserve stays the same, \ie, $r_i = r_{s - 1}$, for all $i \geq s$. So, it suffices to prove $r_i \leq \rwin$ for every step $i < s$, which we do by induction on $i$.
    
    For $i = 0$, we have $r_0 = 0 \leq \rwin$. Assume $\rcurr \leq \rwin$, for some fixed $0 < i < s$. As $x_i$ is reserved, the test in \Lineref{line:phase1check} gives $x_i < \pwin{\rcurr}$. So, we have
    \[
        r_i = \rcurr + x_i < \rcurr + \pwin{\rcurr}\;.
    \]
    By \cref{lem:r+pwin(r)_increasing_in_r}, $r + \pwin{r}$ is non-decreasing in $r$, and $\rcurr \leq \rwin$ by induction hypothesis. So, we have
    \[
        r_i < \rwin + \pwin{\rwin} = \rwin\;,
    \]
    where the equality uses $\pwin{\rwin} = 0$ from \cref{def:rwin}; this closes the induction.
\end{proof}

As in Phase 1 no items get packed by \algsymb, the packability invariant is reduced to $r_i \leq 1$, which is a direct corollary from \cref{lem:ri<=rwin<=1}.

\begin{corollary}\label{cor:packability_invariant_phase_1}
    Let $(\alpha, \beta) \in \Usymb$ and assume that \algsymb does not win early. Then the packability invariant $p_i + r_i \leq 1$ holds at every step $i$ of Phase 1. \qed
\end{corollary}

Let us now prove the packability invariant in Phase 2. For Phase 2, we record how the option of winning early depends on the size of the current item. Removing nothing and adding a subset $\Rsub \subseteq \Rcurr$, \algsymb wins with $\Rsub$ exactly when
\[
    f + \alpha \rcurr \leq \pcurr + \rsub + x_i \leq 1\;.
\]
So, we can define a corresponding \emph{pack-and-win interval} of item sizes
\begin{equation}
    \winterval{\Rsub} \da [f + \alpha \rcurr - \rsub - \pcurr, 1 - \rsub - \pcurr]\;. \label{def:winterval}
\end{equation}
The crucial structural fact to be proven is that these intervals leave no gaps as $\Rsub$ ranges over all subsets of the reserve. Before we do so, we need the following two technical lemmas.

\begin{lemma}\label{lem:xcurr<=1-f-arcurr}
    Let $(\alpha, \beta) \in \Usymb$. For every item $x_i$ that got reserved by \algsymb, we have
    \[
        x_i \leq 1 - f - \alpha r_i\;.
    \]
\end{lemma}
\begin{proof}
    Since $r_i = \rcurr + \xcurr$, the claim is equivalent to
    \begin{equation}
        (1 + \alpha) x_i \leq 1 - f - \alpha \rcurr\;. \label{suf:lem:xcurr<1-f-arcurr:(1+a)xi<=1-f-arcurr}
    \end{equation}
    As $x_i$ was reserved, the test in \Lineref{line:phase1check} of \algsymb gives
    \begin{equation}
        x_i < \pwin{\rcurr} = \frac{f^2 - (1 - \alpha - \alpha f) \rcurr}{1 - \beta f}\;. \label{eq:lem:xcurr<=1-f-arcurr:xi_upper_bound}
    \end{equation}
    Substituting this into \cref{suf:lem:xcurr<1-f-arcurr:(1+a)xi<=1-f-arcurr}, it suffices to show
    \begin{equation}
        (1 + \alpha) \frac{f^2 - (1 - \alpha - \alpha f) \rcurr}{1 - \beta f} \leq 1 - f - \alpha \rcurr\;. \label{suf:lem:xcurr<1-f-arcurr:substituted}
    \end{equation}
    We distinguish three cases according to \cref{lem:f_piecewise}.
    \begin{casedist}
        \case Assume $\beta \leq 1 - \alpha$, so $f = \frac{1}{2}$. In \cref{eq:lem:xcurr<=1-f-arcurr:xi_upper_bound}, the numerator $f^2 - (1 - \alpha - \alpha f) \rcurr = \pwin{\rcurr} (1 - \beta f)$ is positive, since $\pwin{\rcurr} > x_i > 0$ by \cref{eq:lem:xcurr<=1-f-arcurr:xi_upper_bound} and $1 - \beta f > 0$ by \cref{eq:1-bf>0_1-a-af>=(1-a)^2>0}. With $f = \frac{1}{2}$ fixed, only the denominator $1 - \beta f = 1 - \frac{\beta}{2}$ depends on $\beta$; hence the left-hand side of \cref{suf:lem:xcurr<1-f-arcurr:substituted} increases in $\beta$, while the right-hand side is independent of $\beta$. Replacing $\beta$ by its maximum $1 - \alpha$ thus yields the hardest case, and it suffices to show
        \[
            (1 + \alpha) \frac{\frac{1}{4} - (1 - \frac{3}{2} \alpha) \rcurr}{\frac{1}{2} (1 + \alpha)} \leq \frac{1}{2} - \alpha \rcurr\;,
        \]
        which simplifies to
        \[
            \frac{1}{2} - (2 - 3\alpha) \rcurr \leq \frac{1}{2} - \alpha \rcurr\;.
        \]
        For $\rcurr = 0$, the inequality holds trivially. For $\rcurr > 0$, it simplifies to
        \[
            \alpha \leq \frac{1}{2}\;,
        \]
        which holds since $\frac{1}{2} < \beta \leq 1 - \alpha$ by \cref{def:Usymb}.
        \case Assume $\beta \leq \alpha + \frac{2\alpha - 1}{\alpha\cdot(1 - \alpha)}$, so $f = 1 - \alpha$. Substituting $f = 1 - \alpha$ in \cref{suf:lem:xcurr<1-f-arcurr:substituted}, it suffices to show
        \[
            (1 + \alpha)\cdot \frac{(1 - \alpha)^2 - (1 - \alpha - \alpha(1 - \alpha)) \rcurr}{1 - \beta\cdots (1 - \alpha)} \leq 1 - (1 - \alpha) - \alpha \rcurr\;,
        \]
        which simplifies to
        \[
            (1 + \alpha) (1 - \rcurr) \frac{(1 - \alpha)^2}{1 - \beta (1 - \alpha)} \leq \alpha(1 - \rcurr)\;.
        \]
        For $\rcurr = 1$, this holds trivially. Otherwise, by \cref{lem:ri<=rwin<=1}, $1 - \rcurr \geq 0$, and $1 - \beta(1 - \alpha) = 1 - \beta f > 0$ by \cref{eq:1-bf>0_1-a-af>=(1-a)^2>0}. So, the inequality is equivalent to
        \[
            (1 + \alpha) (1 - \alpha)^2 \leq \alpha(1 - \beta(1 - \alpha))\;,
        \]
        which rearranges to $\alpha \beta \cdot (1 - \alpha) \leq \alpha - (1 + \alpha)(1 - \alpha)^2 = (2\alpha - 1) + \alpha^2\cdot (1 - \alpha)$ and thus to exactly the case hypothesis
        \[
            \beta \leq \alpha + \frac{2\alpha - 1}{\alpha\cdot(1 - \alpha)}\;.
        \]
        \case Assume $\beta > \max \left( 1 - \alpha, \alpha + \frac{2 \alpha - 1}{\alpha\cdot (1 - \alpha)} \right)$, so $f = \fo$. Dividing \cref{suf:lem:xcurr<1-f-arcurr:substituted} by $1 + \alpha > 0$ and splitting both sides, it suffices to show
        \[
            \frac{\fo^2}{1 - \beta \fo} - \frac{(1 - \alpha - \alpha \fo) \rcurr}{1 - \beta \fo} \leq \frac{1 - \fo}{1 + \alpha} - \frac{\alpha \rcurr}{1 + \alpha}\;.
        \]
        Since $\frac{\fo^2}{1 - \beta \fo} = \frac{1 - \fo}{1 + \alpha}$ by \cref{eq:fo_identity}, the first terms on each side cancel and the inequality becomes
        \[
            \frac{\alpha \rcurr}{1 + \alpha} \leq \frac{(1 - \alpha - \alpha \fo) \rcurr}{1 - \beta f}\;.
        \]
        For $\rcurr = 0$, this holds trivially. For $\rcurr > 0$, it simplifies to
        \[
            \alpha \leq (1 - \alpha - \alpha \fo) \frac{1 + \alpha}{1 - \beta \fo}\;.
        \]
        Again by \cref{eq:fo_identity}, we have $\frac{1 + \alpha}{1 - \beta \fo} = \frac{1 - \fo}{\fo^2}$, so it suffices to show
        \[
            \alpha \leq (1 - \alpha - \alpha \fo) \frac{1 - \fo}{\fo^2}\;.
        \]
        By multiplying with $\fo^2 > 0$ and rearranging, the inequality simplifies to
        \[
            \fo \leq 1 - \alpha\;,
        \]
        which holds because $\fo = f \leq 1 - \alpha$ by \cref{def:f}.
    \end{casedist}
\end{proof}

The following purely combinatorial lemma shows that the subset sums of a set whose elements are individually small in the sense below, leave no gaps larger than $\gapconst$; applied to the reserve, it will let the pack-and-win-intervals chain into a single interval.

\begin{lemma}\label{lem:subset_gaps}
    Let $S = \{x_1, \dots, x_n\}$ be a set of $n$ positive reals, not necessarily distinct, and $\gapconst \geq 0$ a constant with $x_i \leq \gapconst + \sum_{j > i} x_j$ for all $i$. Then any enumeration $S_1, \dots, S_{2^n}$ of the subsets of $S$ with $s_1 \leq \cdots \leq s_{2^n}$ has consecutive gaps of at most $\gapconst$, \ie, $s_{k + 1} - s_k \leq \gapconst$ for every $k$.
\end{lemma}
\begin{proof}
    For $n = 0$, there is nothing to show, so assume $n \geq 1$. We prove, by induction on $i$ from $n$ down to $1$, the following: an enumeration of the subsets of $\{x_i, \dots, x_n\}$ with non-decreasing sizes has consecutive size gaps of at most $\gapconst$. Taking $i = 1$ gives the claim.

    For $i = n$, the subsets of $\{x_n\}$ are $\emptyset$ and $\{x_n\}$. Their gap is $x_n \leq \gapconst + \sum_{j > n} x_j = \gapconst$, using the assumption of the lemma. So the enumeration $\emptyset, \{x_n\}$ fulfills the wanted property.

    For a fixed $i < n$, assume every enumeration $T_1, \dots, T_M$ of the subsets of $T \da \{x_{i + 1}, \dots, x_n\}$ with $t_1 \leq \cdots \leq t_M$ has consecutive gaps of at most $\gapconst$, \ie, $t_{k + 1} - t_k \leq \gapconst$ for all $k$, where $M = 2^{n - i}$. Every subset of $\{x_i, \dots, x_n\}$ is either some $T_k$ or some $T'_k \da T_k \cup \{x_i\}$.
    
    Any enumeration of the subsets of $\{x_i, \dots, x_n\}$ with non-decreasing sizes arises by merging the two chains $T_1, \dots, T_M$ and $T'_1, \dots, T'_M$ in order of non-decreasing size: the enumeration restricted to each chain is itself non-decreasing, so it induces a valid ordering of that chain, and the induction hypothesis applies to the ordering induced on $T_1, \dots, T_M$. Note also that the primed chain has the same internal gaps as the unprimed one, since $t'_{k + 1} - t'_k = t_{k + 1} - t_k \leq \gapconst$. We prove that any such merged enumeration has consecutive gaps of at most $\gapconst$.

    Consider an arbitrary such enumeration. Consecutive elements of the merged enumeration come from the two sorted chains, each of which has internal gaps of at most $\gapconst$, as observed above. A new gap can only occur where the merge switches chains, \ie, between some $T_a$ and some $T'_b = T_b \cup \{x_i\}$ that are adjacent in the merge.
    
    To arrive at a contradiction, assume two sets $T_a$ and $T'_b$ that are adjacent in the merge have a gap of more than $\gapconst$.
    
    First, assume $t_a \leq t'_b$. The set $T_{a + 1}$ has a gap of at most $\gapconst$ to $T_a$ by induction hypothesis and thus $t_a \leq t_{a + 1} < t'_b$, since the gap between $T_a$ and $T'_b$ is larger than $\gapconst$. But $T_{a + 1}$ succeeds $T_a$ in the merge, so it would have to lie strictly between $T_a$ and $T_b$ in the enumeration, contradicting their adjacency. Hence $T_{a + 1}$ does not exist, implying $a = M$. With a similar argumentation, we get $b = 1$. So, the gap between $T_a$ and $T'_b$ is
    \[
        t'_b - t_a = t'_1 - t_M = x_i - \sum_{j > i} x_j \leq \gapconst\;,
    \]
    which is a contradiction.
    
    Now, assume $t_a > t'_b$. A similar argumentation as above gives $a = 1$ and $b = M$, so $t_a = t_1 = 0$ and $t'_b = t'_M = \sum_{j \geq i} x_j$, which is a contradiction to $t_a > t'_b$.
\end{proof}

\begin{lemma}\label{lem:wintervals_cover_without_gaps}
    Let $(\alpha, \beta) \in \Usymb$, assume \algsymb does not win early on an instance $I = (x_1, \dots, x_n)$, and let step $i$ be in Phase 2. Then, the intervals $\winterval{\Rsub}$ with $\Rsub \subseteq \Rcurr$ cover the interval
    \[
        [f - (1 - \alpha) \rcurr - \pcurr, 1 - \pcurr]
    \]
    without gaps.
\end{lemma}
\begin{proof}
    For subsets $\Rsubo$ and $\Rsubt$ with $\rsubo \leq \rsubt$, the interval $\winterval{\Rsubt}$ is componentwise below $\winterval{\Rsubo}$, so the two overlap if and only if $f + \alpha \rcurr - \rsubo - \pcurr \leq 1 - \rsubt - \pcurr$, which simplifies to
    \begin{equation}
        \rsubt - \rsubo \leq 1 - f - \alpha \rcurr \ad \gapconst\;. \label{eq:cond_for_winterval_overlap}
    \end{equation}
    Note that $\gapconst \geq 0$ and that every interval $\winterval{\Rsub}$ is nonempty: by \cref{lem:ri<=rwin<=1}, $\rcurr \leq \rwin$, so $f + \alpha \rcurr \leq f + \alpha \rwin = \frac{(1 - \alpha) f}{1 - \alpha - \alpha f} \leq 1$, where the last inequality is equivalent to $f \leq 1 - \alpha$.
    Since \algsymb reserves only in Phase 1, the reserve is $\Rcurr = \{x_1, \dots, x_{s - 1}\}$, where $x_s$ is the item that triggered \algsymb to switch from Phase 1 to Phase 2. By \cref{lem:xcurr<=1-f-arcurr}, every reserved item satisfies $x_j \leq 1 - f - \alpha r_j$. Using $(1 - \alpha) \sum_{k = j + 1}^{s - 1} x_k \geq 0$ and $\rcurr = r_j + \sum_{k = j + 1}^{s - 1} x_k$, this gives
    \[
        x_j \leq 1 - f - \alpha r_j + (1 - \alpha) \sum_{k = j + 1}^{s - 1} x_k = 1 - f - \alpha \rcurr + \sum_{k = j + 1}^{s - 1} x_k = \gapconst + \sum_{k = j + 1}^{s - 1} x_k\;,
    \]
    so the elements $x_1, \dots, x_{s - 1}$ of $\Rcurr$ satisfy the assumption of \cref{lem:subset_gaps} with the constant $\gapconst$. By \cref{lem:subset_gaps}, any enumeration $\emptyset = S_1, S_2, \dots, S_m = \Rcurr$ of the subsets of $\Rcurr$ with $s_1 \leq \cdots \leq s_m$ has consecutive size gaps at most $\gapconst$. Thus, by \cref{eq:cond_for_winterval_overlap}, the corresponding intervals overlap consecutively, so their union is connected; it runs from the top end $1 - \pcurr$ of $\winterval{\emptyset}$ down to the bottom end $f + \alpha \rcurr - \rcurr - \pcurr = f - (1 - \alpha) \rcurr - \pcurr$ of $\winterval{\Rcurr}$, which is the claim.
\end{proof}

\begin{lemma}\label{lem:packability_invariant_phase_2}
    Let $(\alpha, \beta) \in \Usymb$ and suppose that \algsymb does not win early. Then the packability invariant $p_i + r_i \leq 1$ holds at every step $i$ in Phase 2.
\end{lemma}
\begin{proof}
    Let $x_s$ be the item that triggers the switch from Phase 1 to Phase 2; the Phase-2 steps are $i \geq s$. We use induction on $i$, with the last Phase-1 step as base: at step $s - 1$, \cref{cor:packability_invariant_phase_1} gives $p_{s - 1} + r_{s - 1} \leq 1$.
    
    For a fixed $i \geq s$, assume the packability invariant holds for step $i - 1$, \ie, $\pcurr + \rcurr \leq 1$. Since \algsymb does not win, it either rejects or packs item $x_i$. If $x_i$ is rejected, then $p_i + r_i = \pcurr + \rcurr \leq 1$ by induction hypothesis. If $x_i$ is packed, we get by \cref{lem:wintervals_cover_without_gaps},
    \begin{equation}
        x_i \notin [f - (1 - \alpha) \rcurr - \pcurr, 1 - \pcurr]\;. \label{eq:lem:packability_invariant_phase_2:xi_notin_wintervals}
    \end{equation}
    Since $x_i$ gets packed, it must fit into the knapsack, \ie, $\pcurr + x_i \leq 1$ or equivalently $x_i \leq 1 - \pcurr$. So \cref{eq:lem:packability_invariant_phase_2:xi_notin_wintervals} simplifies to
    \[
        x_i < f - (1 - \alpha) \rcurr - \pcurr\;.
    \]
    So, we have
    \[
        p_i + r_i = \pcurr + x_i + \rcurr \leq f + \alpha \rcurr \leq f + \alpha\cdot  \rwin = \frac{(1 - \alpha) f}{1 - \alpha - \alpha f} \leq 1\;,
    \]
    where the second inequality holds by \cref{lem:ri<=rwin<=1} and the last is equivalent to $f \leq 1 - \alpha$, which holds by \cref{def:f}.
\end{proof}

We can now use the results from \cref{lem:packability_invariant_phase_2,cor:packability_invariant_phase_1} to prove \cref{lem:packability_invariant}.

\begin{proof}[of \cref{lem:packability_invariant}]
    Recall that the lemma asserts $p_i + r_i \leq 1$ for every $0 \leq i \leq n$, provided that \algsymb does not win early. \cref{cor:packability_invariant_phase_1} covers the steps of Phase 1, including the initial step $0$, and \cref{lem:packability_invariant_phase_2} covers those of Phase 2; together this is every step $0 \leq i \leq n$.
\end{proof}

\subsection{Critical Rejection}\label{subsec:crit_reject}

With the packability invariant established, we can supply the deferred proof of \cref{lem:critical_rejection}, which we restate here.

\LemmaCriticalRejection*

\begin{proof}
       We show that every item \algsymb rejects has size greater than $\frac{1}{2}$. Two items of size greater than $\frac{1}{2}$ cannot both fit into a knapsack of capacity $1$, so $\Opt(I)$ contains at most one of them, which is the claim.

    Let $x_i$ be an item rejected by \algsymb. Phase~1 rejects nothing, so $x_i$ is rejected in Phase 2; and since \algsymb does not win, $x_i$ was rejected only because it did not fit, i.e., $\pcurr + x_i > 1$. It therefore suffices to show $\pcurr \leq f$, as then $x_i > 1 - \pcurr \geq 1 - f \geq \frac{1}{2}$, using $f \leq \frac{1}{2}$.

    Note that $\pcurr > 0$ and thus $\Pcurr$ is nonempty since otherwise $x_i$ would be greedily packed by \algsymb. 
    Let $x_j$ be the last item packed before $x_i$ is rejected. 
    No items are packed or reserved between the packing of $x_j$ and the rejection of $x_i$, and, as \algsymb does not win early, none are removed either; so the packed and reserved sizes are unchanged over this span, and in particular $\rcurr = r_{j - 1}$ and $\pcurr + \rcurr = p_{j - 1} + x_j + r_{j - 1}$.
    By \cref{lem:packability_invariant}, the packability invariant holds at time step $i - 1$, \ie, $\pcurr + \rcurr \leq 1$. So, we have
    \begin{equation}
        p_{j - 1} + r_{j - 1} + x_j \leq 1\;.\label{eq:lem:critical_rejection:upper_bound_of_win_condition}
    \end{equation}
    Now consider the round in which $x_j$ arrives. \algsymb could have packed $x_j$ together with the whole reserve $R_{j - 1}$, removing nothing; this move satisfies the upper bound of the condition in \Lineref{line:checkwin} because of \cref{eq:lem:critical_rejection:upper_bound_of_win_condition}. If the lower bound would also hold, namely $f + \alpha r_{j - 1} \leq p_{j - 1} + x_j + r_{j - 1}$, then $x_j$ would have caused \algsymb to win early. As that did not happen, that lower bound fails and using $\rcurr = r_{j - 1}$ we have
    \[
        p_{j - 1} + r_{j - 1} + x_j < f + \alpha \rcurr\;.
    \]
    Rearranging and using $\alpha - 1 \leq 0$, we have
    \[
        \pcurr = p_{j - 1} + x_j < f + (\alpha - 1) \rcurr \leq f\;,
    \]
    which completes the proof. 
\end{proof}

\subsection{Safe Rejection}\label{subsec:safe_reject}

Now we prove that, after \algsymb has reached a reserve of size $\rwin$, all its rejections are safe.

\LemmaSafeRejection*

\begin{proof}
    Since \algsymb does not win early, it packs its reserve after the arrival of the last item (\Lineref{line:packall}); so we have
    \[
        \gain{\symb(I)} = p_n + (1 - \alpha) r_n\;.
    \]
    
    Consider the time step at which $x_i$ is rejected. \algsymb has packed $\Pcurr$, reserved $\Rcurr$, and packs a (possibly empty) set $\Plater$ of further items afterwards. As $x_i$ is rejected, \algsymb is in Phase 2, so the reserve is frozen from here on, \ie, $r_n = \rcurr$, and we have
    \begin{equation}
        \gain{\symb(I)} = \pcurr + \plater + (1 - \alpha) \rcurr\;.\label{eq:lem:safe_rejection:gain(alg)}
    \end{equation}

    By \cref{lem:critical_rejection}, $x_i$ is the only item packed by \Opt and rejected by \algsymb. Thus, every other item of $\Opt(I)$ is packed or reserved by \algsymb. In particular, let $\Popt \subseteq \Pcurr$ and $\Ropt \subseteq \Rcurr$ denote the part \Opt takes from the first $i - 1$ items. Since \Opt's remaining items lie among $\Plater$, we have
    \begin{equation}
        \opt(I) \leq \popt + \plater + x_i + \ropt\;.\label{eq:lem:safe_rejection:opt}
    \end{equation}
    
    The fill ratio $\gain{\symb(I)}/\opt(I)$ is positive and at most $1$, so subtracting the common term $\plater$ from \cref{eq:lem:safe_rejection:gain(alg)} and \cref{eq:lem:safe_rejection:opt} only decreases it; hence it suffices to show
    \begin{equation}
        \pcurr + (1 - \alpha) \rcurr \geq (\popt + x_i + \ropt)\cdot f\;.\label{suf:lem:safe_rejection:reduced_fraction}
    \end{equation}
    
    Since \algsymb does not win early on $x_i$, the condition (\Lineref{line:checkwin}) fails for $P' = \Popt$, $R' = \Ropt$. As \Opt packs $\Popt$, $x_i$, and $\Ropt$ together, we have $\popt + x_i + \ropt \leq 1$, so the upper bound of the condition holds and the lower bound must fail:
    \[
        \popt + x_i + \ropt < f + \alpha \rcurr + \beta (\pcurr - \popt) \leq f + \alpha \rcurr + \beta \pcurr\;.
    \]
    Substituting this in \cref{suf:lem:safe_rejection:reduced_fraction} only strengthens its claim, so it suffices to show
    \[
        \pcurr + (1 - \alpha) \rcurr \geq (f + \alpha \rcurr + \beta \pcurr)\cdot f\;.
    \]
    Since $1 - \beta f > 0$ (see \cref{eq:1-bf>0_1-a-af>=(1-a)^2>0}), solving this inequality for $\pcurr$ yields
    \[
        \pcurr \geq \frac{f^2 - (1 - \alpha - \alpha f) \rcurr}{1 - \beta f} \overset{\cref{def:pwin}}{=} \pwin{\rcurr}\;,
    \]
    which holds by assumption.  
\end{proof}

\subsection{Lower Bounds}\label{sub:lb_symbiosis_proofs}
In this section, we restate and prove the lower bounds from \cref{thm:globaltwo,thm:lowerboundsymbiosistrivial,thm:lowerboundf} as already stated in \cref{sec:symbiosis}.
\thmGlobalTwo*
\begin{proof}
    Let $\eps>0$ be sufficiently small. Now consider any deterministic algorithm \alglower.
    \alglower receives items of size $1/2+\eps^i$, for $1\leq i\leq N$, in this order, where $N$ is a large value to be determined later. 
    If \alglower ever rejects an item $1/2+\eps^i$, the instance immediately ends after a final item of size $1/2-\eps^i$. The optimal solution has value $1$, while \alglower cannot pack any two items together, because $(1/2+\eps^j)+(1/2-\eps^i)>1$, for any $j<i$. Its competitive ratio is therefore at least
    \[
        \frac{1}{1/2+\eps}\xrightarrow{\eps\to 0} 2\;.
    \]
    If \alglower never rejects such an item, it must either reserve or pack all of them. Since it cannot pack two of the items simultaneously, it must remove any item it packs, except for the last. If $N_\alpha$ is the number of items that \alglower reserves and $N_\beta$ is the number of items it removes, we therefore know that $N_\alpha+N_\beta\geq N-1$. Since every item has size greater than $1/2$, the total cost paid for reservation and removal is therefore at least
    \[\frac{\alpha\cdot N_\alpha+\beta\cdot N_\beta}{2} \geq \frac{\min(\alpha,\beta)\cdot (N-1)}{2} > 1\;,\]
    for $N>2/\min(\alpha,\beta)+1$. Whatever \alglower packs, its gain is therefore negative and its competitive ratio is unbounded.
\end{proof}
\lowerBoundSymbiosisTrivial*
\begin{proof}
    Consider any deterministic algorithm \alglower. The structure of the proof is shown in \cref{fig:lbsymtriv}. \alglower first receives an item $x_1=(1-\alpha)/(2-\alpha)$. Because $0<\alpha<1$, we know that $0<x_1<1/2$. If \alglower rejects $x_1$, its competitive ratio is unbounded. If it reserves it, the instance stops and its competitive ratio is at least
    \[\frac{x_1}{x_1-\alpha\cdot x_1}=\frac{1}{1-\alpha}\;.\]
    We can therefore assume that \alglower packs $x_1$. It then receives an item $x_2=1-x_1+\eps$, for some sufficiently small $\eps>0$. Note that, since $x_1<1/2$, we know that $x_2>x_1$. 
    If \alglower reserves $x_2$, the instance stops. \alglower's gain is then either
    $x_1-\alpha\cdot x_2<x_2-\alpha\cdot x_2$ if it does not pack $x_2$, or $x_2-\beta\cdot x_1-\alpha\cdot x_2<x_2-\alpha\cdot x_2$ if it removes $x_1$ and packs $x_2$. In both cases, its competitive ratio is therefore greater than 
    \[\frac{x_2}{x_2-\alpha\cdot x_2}=\frac{1}{1-\alpha}\;.\]
    If \alglower rejects $x_2$, its competitive ratio is 
    \begin{align*}
    \frac{x_2}{x_1}&>\frac{1-x_1}{x_1}=\frac{1-\frac{1-\alpha}{2-\alpha}}{\frac{1-\alpha}{2-\alpha}}=\frac{1}{1-\alpha}\;.
    \end{align*}
    We can therefore assume that \alglower packs $x_2$. Since $x_1+x_2=1+\epsilon>1$, it has to remove $x_1$. It then receives a final item $x_3=1-x_1$. Since no items arrive afterwards, reserving $x_3$ is clearly suboptimal. Since $x_1<1/2$, we know that $x_3>1/2$ and thus $x_2+x_3$ cannot both be packed together. Since $x_2>x_3$, \alglower also does not benefit from packing $x_3$ instead of $x_2$. \alglower's gain is therefore highest if it rejects $x_3$. In this case, its competitive ratio is
    \begin{equation}\label{eq:simplebound}
    \frac{1}{x_2-\beta\cdot x_1}=\frac{1}{1-(1+\beta)\cdot x_1+\eps}\;.
    \end{equation}
    Recall from the definition of $\Usymb$ in \cref{def:Usymb-f} that
    \[\beta > \frac{3\alpha-\alpha^2-1}{1-\alpha}\;.\]
    This means that
    \begin{align*}
    (1+\beta)\cdot x_1&=(1+\beta)\cdot \frac{1-\alpha}{2-\alpha}\\
    &> \left(1+\frac{3\alpha-\alpha^2-1}{1-\alpha}\right)\cdot \frac{1-\alpha}{2-\alpha}\\
    &=\frac{2\alpha-\alpha^2}{1-\alpha}\cdot \frac{1-\alpha}{2-\alpha}\\
    &=\alpha\;.
    \end{align*}
    Thus, by \eqref{eq:simplebound}, \alglower's competitive ratio is at least
    \[\frac{1}{1-\alpha+\eps}\xrightarrow{\eps\to 0} \frac{1}{1-\alpha}\;.\]
\end{proof}

\thmLowerBoundf*
\begin{proof}
We first give an overview of the structure of the proof. The proof is dependent on whether or not $\alpha\geq 1/3$, but is identical up to a certain point. The structure is illustrated in \cref{fig:lbsym}. We then check all relevant conditions separately. They are therefore presented in the order in which they are analyzed later, which may make the enumeration initially seem somewhat arbitrary in places.
\begin{figure}
    \centering
    \begin{tikzpicture}[
    font=\large,
    >=Latex
]
    \def\rejdist{3cm}
    \def\resdist{3.5cm}
    
    \node[itemsettings] (x1)
        {$x_1 \da \frac{1-\fo}{1+\alpha} + \eps$};
    \node[itemsettings, 
          ifpack=x1] (x2)
        {$x_2 \da \fo + \beta x_1$};
    \node[itemsettings,
          ifpack=x2] (x3)
        {$x_3 \da 1 - x_1$};
    \node[itemsettings,
          ifreservedepth3=x1] (y1)
        {$y_1 \da 1 - x_1 + \eps$};
    \node[itemsettings,
          ifpack=y1] (y2)
        {$y_2 \da 1 - x_1$};
    \node[itemsettings,
          ifreservedepth2=y1] (z1)
        {$z_1 \da \fo + \alpha\cdot (x_1 + y_1)$};
    \node[itemsettings,
          ifpack=z1] (z2)
        {$z_2 \da 1 - y_1$};
    
    \node[leafsettings,
          ifreject=x1] (x1rej)
        {$c = \infty$};
    \node[leafsettings,
          ifreject=x2] (x2rej)
        {$c = \frac{x_2}{x_1}$};
    \node[leafsettings,
            anchor=west,
          below right = 0.5cm and 1.7cm of x2.east] (x2res)
        {$c = \min \left( \frac{x_2}{x_1 - \alpha x_2}, \frac{x_2}{x_2 - \alpha x_2 - \beta x_1} \right)$};
    \node[leafsettings,
          ifreject=x3] (x3rej)
        {$c \geq \frac{1}{x_2 - \beta x_1}$};
    \node[leafsettings,
          ifpack=x3] (x3pack)
        {$c \geq \frac{1}{x_3 - \beta \cdot(x_1 + x_2)}$};
    \node[leafsettings,
          ifreject=y1] (y1rej)
        {$c = \frac{y_1}{x_1 - \alpha x_1}$};
    \node[leafsettings,
          anchor=west, below right = 1cm and 3cm of y1.east] (y1res)
        {$c = \frac{y_1}{y_1 - \alpha\cdot (x_1 +y_1)}$};
            
    \node[leafsettings,
          ifreject=y2] (y2rej)
        {$c \geq \frac{1}{y_1 - \alpha x_1}$};
    \node[leafsettings,
          ifpack=y2] (y2pack)
        {$c \geq \frac{1}{1 - \alpha x_1 - \beta y_1}$};
    \node[leafsettings,
          ifreserve=z1] (z1res)
        {$c = \frac{z_1}{z_1 - \alpha\cdot (x_1 + y_1 + z_1)}$};
    \node[leafsettings,
          ifreject=z1] (z1rej)
        {$c = \frac{z_1}{y_1 - \alpha \cdot(x_1 + y_1)}$};
    \node[leafsettings,
          ifreject=z2] (z2rej)
        {$c \geq \frac{1}{z_1 - \alpha \cdot(x_1 + y_1)}$};
    \node[leafsettings,
          ifpack=z2] (z2pack)
        {$c \geq \frac{1}{1 - \alpha \cdot(x_1 + y_1) - \beta z_1}$};


    \node[at=(y1.east), yshift=0.1cm] (y1uppereast) {};
    \node[at=(y1.east), yshift=-0.1cm] (y1lowereast) {};
    \node[at=(y1uppereast), anchor=south west, color=ReserveColor, font=\normalsize] {$\alpha\geq 1/3$};
    \node[at=(y1lowereast), anchor=north west, color=ReserveColor, font=\normalsize] {$\alpha< 1/3$};
    
    \draw[rejarr] (x1.west) -| (x1rej.north);
    \draw[packarr] (x1.south) -- (x2.north);
    \draw[resarr] (x1.east) -| (y1.north);
    \draw[rejarr] (x2.west) -| (x2rej.north);
    \draw[resarr over] (x2.east) -| (x2res.north);
    \draw[packarr] (x2.south) -- (x3.north);
    \draw[rejarr] (x3.west) -| (x3rej.north);
    \draw[packarr] (x3.south) -- (x3pack.north);
    \draw[rejarr] (y1.west) -| (y1rej.north);
    \draw[packarr] (y1.south) -- (y2.north);
    \draw[resarr] (y1lowereast) -| (z1.north);
    \draw[rejarr] (y2.west) -| (y2rej.north);
    \draw[packarr] (y2.south) -- (y2pack.north);
    \draw[resarr] (y1uppereast) -| (y1res.north);
    \draw[rejarr] (z1.west) -| (z1rej.north);
    \draw[resarr] (z1.east) -| (z1res.north);
    \draw[packarr] (z1.south) -- (z2.north);
    \draw[rejarr] (z2.west) -| (z2rej.north);
    \draw[packarr] (z2.south) -- (z2pack.north);


    \node[numbersettings, slightlyaboveleft=x2] {\symbiosisCaseXii};
    \node[numbersettings, slightlyaboveleft=x2rej] {\symbiosisCaseRejectXii};
    \node[numbersettings, slightlyaboveright=x2res] {\symbiosisCaseReserveXii};
    \node[numbersettings, slightlyaboveleft=x3rej] {\symbiosisCaseRejectXiii};
    \node[numbersettings, slightlyaboveleft=x3pack] {\symbiosisCasePackXiii};
    \node[numbersettings, slightlyaboveleft=y2] {\symbiosisCaseYii};
    \node[numbersettings, slightlyaboveleft=y1rej] {\symbiosisCaseRejectYi};
    \node[numbersettings, slightlyaboveleft=y2rej] {\symbiosisCaseRejectYii};
    \node[numbersettings, slightlyaboveleft=y2pack] {\symbiosisCasePackYii};
    \node[numbersettings, slightlyaboveleft=z1] {\symbiosisCaseZi};
    \node[numbersettings, slightlyaboveleft=z1rej] {\symbiosisCaseRejectZi};
    \node[numbersettings, slightlyaboveright=z1res] {\symbiosisCaseReserveZi};
    \node[numbersettings, slightlyaboveright=y1res] {\symbiosisCaseReserveYi};
    \node[numbersettings, slightlyaboveleft=z2rej] {\symbiosisCaseRejectZii};
    \node[numbersettings, slightlyaboveleft=z2pack] {\symbiosisCasePackZii};

\end{tikzpicture}
    \caption{Structure of the lower bound analyzed in \cref{thm:lowerboundf}. Downward arrows (green) represent the choice to pack an item, while leftward arrows (red, dashed) and rightward arrows (blue, dotted), respectively, represent rejecting or reserving an item.}
    \label{fig:lbsym}
\end{figure}
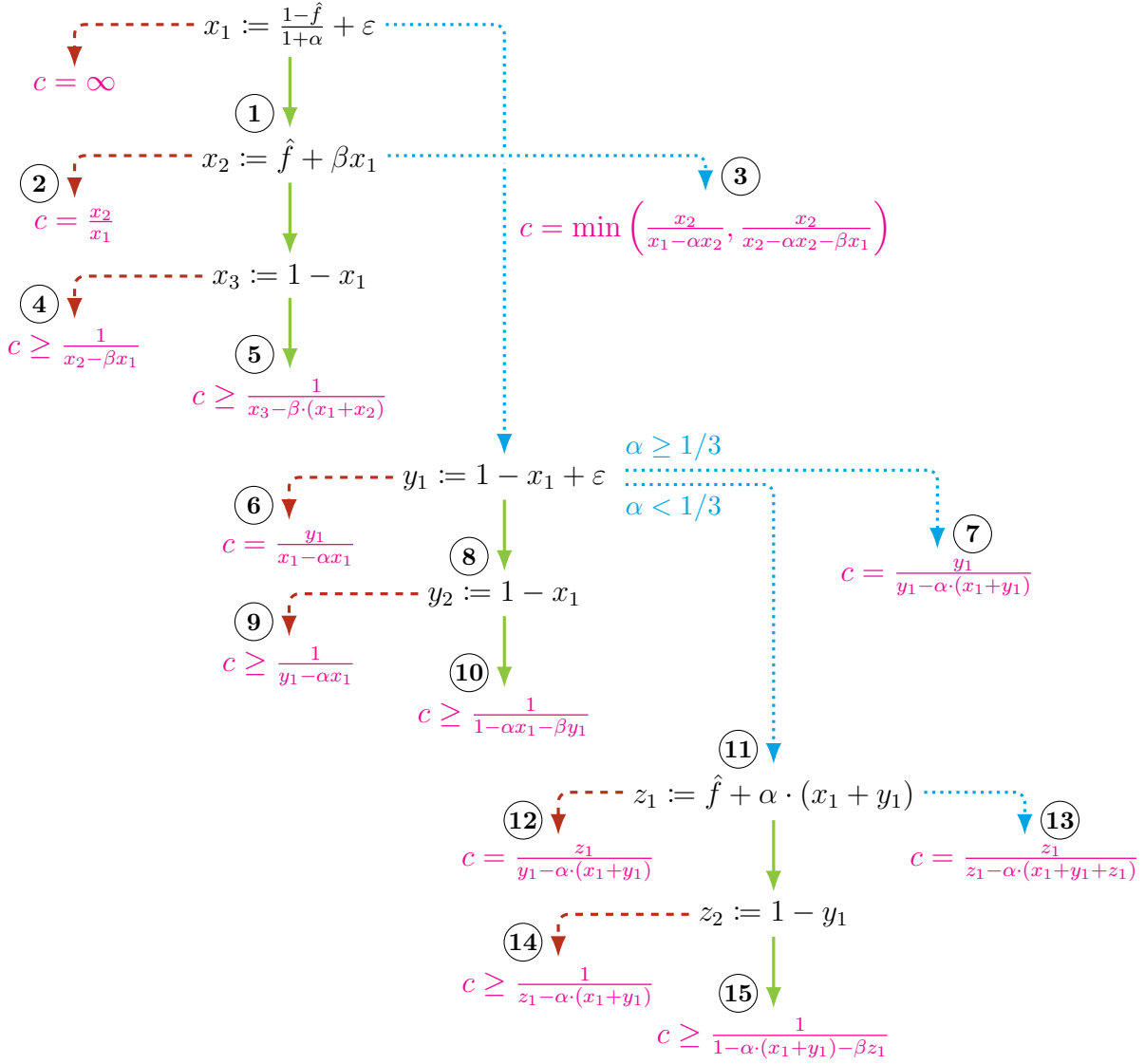

Consider any deterministic algorithm \alglower. 
\alglower first receives an item $x_1\da (1-\fo)/(1+\alpha)+\eps$. Note that, because $0<\fo<1$ and $\alpha>0$, this is a valid item size. If \alglower rejects this item, its competitive ratio is unbounded.  If \alglower packs $x_1$, it receives an item $x_2\da \fo+\beta\cdot x_1$. We will check in \casename~\symbiosisCaseXii that this is indeed a valid item size and that $x_1+x_2>1$, so the two items cannot be packed together. If \alglower rejects $x_2$, the instance stops. We will show in \casename~\symbiosisCaseRejectXii that its competitive ratio is at least $1/\fo$ in this case. If it reserves $x_2$, the instance also ends and we check its competitive ratio in \casename~\symbiosisCaseReserveXii. If \alglower removes $x_1$ and packs $x_2$, it receives a final item $x_3\da 1-x_1$. Since this is the last item, reserving it is clearly suboptimal, so we only check the cases where $x_3$ is rejected (\casename~\symbiosisCaseRejectXiii) or packed (\casename~\symbiosisCasePackXiii).

If $x_1$ is reserved instead, \alglower receives an item $y_1\da 1-x_1+\eps$. If $y_1$ is rejected, the instance ends (\casename~\symbiosisCaseRejectYi). If it is packed, it receives a final item $y_2\da 1-x_1$. We check in \casename~\symbiosisCaseYii that $y_1+y_2>1$, so the two items cannot be packed together. Once again, since $y_2$ is the last item, reserving it is clearly suboptimal, so we only check the cases where it is rejected (\casename~\symbiosisCaseRejectYii) or packed (\casename~\symbiosisCasePackYii).

If $y_1$ is reserved, the way the proof continues depends on the value of $\alpha$. If $\alpha\geq 1/3$, the instance ends and we check in \casename~\symbiosisCaseReserveYi that \alglower's competitive ratio is at least $1/\fo$. If $\alpha<1/3$, \alglower receives an item $z_1\da \fo+\alpha\cdot (x_1+y_1)$. We check in \casename~\symbiosisCaseZi that this is indeed a valid item.  If $z_1$ is rejected (\casename~\symbiosisCaseRejectZi) or reserved (\casename~\symbiosisCaseReserveZi), the instance ends. Finally, if $z_1$ is packed, the instance receives a final item $z_2\da 1-y_1$. Once again, since this is the last item, we only check the cases where $z_2$ is rejected (\casename~\symbiosisCaseRejectZii) or packed (\casename~\symbiosisCasePackZii).

We now prove that, in each of these cases, the necessary conditions hold and that, however the instance ends, \alglower's competitive ratio is at least $1/\fo$ (in some cases, this holds only in the limit $\eps\to 0$).

\begin{description}
    \item[\casename~\symbiosisCaseXii.] We first need to check that $x_2=\fo+\beta\cdot x_1$ is in fact a valid item size. 
        Clearly, $x_2\geq 0$. Then recall that, by definition of $\Usymb$ in \cref{def:Usymb-f}, it holds that $\beta < 1+\alpha$, so from 
        \[
            \beta < 1+\alpha\implies\frac{\beta}{1+\alpha} < 1\implies \beta\cdot \frac{1-\fo}{1+\alpha} < 1-\fo\implies \fo+\beta\cdot\frac{1-\fo}{1+\alpha}< 1
        \]
        it follows that $x_2=\fo+\beta\cdot \frac{1-\fo}{1+\alpha}+\beta\cdot\eps\leq 1$ for sufficiently small $\eps$.
        We also check that $x_1+x_2>1$. This is the case because
        \begin{align*}
            x_1+x_2&=\frac{1-\fo}{1+\alpha}+\fo+\beta\cdot \frac{1-\fo}{1+\alpha}+\beta\eps\\
            &>\frac{(1-\fo)\cdot (1+\beta)+(1+\alpha)\cdot \fo}{1+\alpha}\\
            &=\frac{1+\beta-\fo\cdot (\beta-\alpha)}{1+\alpha}\\
            &>\frac{1+\beta-(\beta-\alpha)}{1+\alpha}\\
            &=\frac{1+\alpha}{1+\alpha}\;,
        \end{align*}
        because $\fo<1$ and $\beta>\alpha$ by \cref{lemma:betagreateralpha}.
    \item[\casename~\symbiosisCaseRejectXii.] In this situation, \alglower has packed $x_1$ and rejected $x_2$. 
        Its competitive ratio is therefore at least 
        \[
            \frac{x_2}{x_1}=\frac{\fo}{x_1}+\beta\xrightarrow{\eps\to 0}\frac{\fo}{\frac{1-\fo}{1+\alpha}}+\beta\;.
        \]
        By \cref{eq:fo}, this in turn is equal to
        \[
            \frac{\fo}{\frac{1-\fo}{1+\alpha}}+\beta=\frac{\fo}{\frac{\fo^2}{1-\beta\cdot \fo}}+\beta=\frac{1-\beta\cdot \fo}{\fo}+\beta=\frac{1}{\fo}\;.
        \]
    \item[\casename~\symbiosisCaseReserveXii.] In this situation, \alglower has packed $x_1$ and reserved $x_2$. 
        Its gain is thus either $x_1-\alpha\cdot x_2$ if it does not pack $x_2$, or $x_2-\alpha\cdot x_2-\beta\cdot x_1$ if it removes $x_1$ and packs $x_2$ (note that we know from \casename~\symbiosisCaseXii that $x_1+x_2>1$). Clearly, $x_1-\alpha\cdot x_2\leq x_1$. We show that the same holds for $x_2-\alpha\cdot x_2-\beta\cdot x_1$. This is the case because by \eqref{eq:fo<1/2_iff_b>1-a} we know that $\fo<1/2$, and from the implication
        \begin{align*}
            \fo<1/2&\implies\fo<1-\fo\\
            &\implies (1-\alpha^2)\cdot \fo<1-\fo\\
            &\implies (1-\alpha^2)\cdot \fo<(1+\alpha\cdot \beta)\cdot(1-\fo)\\
            &\implies (1-\alpha)\cdot \fo<(1+\alpha\cdot \beta)\cdot\frac{1-\fo}{1+\alpha}\\
            &\implies (1-\alpha)\cdot \fo<(1+\alpha\cdot \beta)\cdot x_1\\
            &\implies (1-\alpha)\cdot \fo-\alpha\cdot \beta\cdot x_1<x_1\\
            &\implies (1-\alpha)\cdot \fo-\alpha\cdot \beta\cdot x_1+\beta\cdot x_1<x_1+\beta\cdot x_1\\
            &\implies (1-\alpha)\cdot (\fo+\beta\cdot x_1)<(1+\beta)\cdot x_1\\
            &\implies (1-\alpha)\cdot x_2<(1+\beta)\cdot x_1\\
            &\implies x_2-\alpha\cdot x_2-\beta\cdot x_1<x_1\;.
        \end{align*}
        In either case, \alglower's competitive ratio is therefore at least $x_2/x_1$, which we proved in \casename~\symbiosisCaseRejectXii was at least $1/\fo$. 

    \item[\casename~\symbiosisCaseRejectXiii.] In this situation, \alglower has packed $x_1$, removed $x_1$, packed $x_2$ and rejected $x_3$. 
        Its gain is therefore $x_2-\beta\cdot x_1$, while the optimal solution is equal to $x_1+x_3=1$. Its competitive ratio is therefore at least
        \[
            \frac{1}{x_2-\beta\cdot x_1}=\frac{1}{\fo+\beta\cdot x_1-\beta\cdot x_1}=\frac{1}{\fo}
        \;.\]
    \item[\casename~\symbiosisCasePackXiii.] In this situation, \alglower has packed $x_1$, removed $x_1$, packed $x_2$, and packed $x_3$. 
        Since we showed in \casename~\symbiosisCaseRejectXii that $x_2/x_1\to 1/\fo>1$ as $\eps\to 0$, we know that $x_2>x_1$, so $x_3+x_2>1$, so \alglower must have removed $x_2$. Its gain is therefore $x_3-\beta\cdot (x_1+x_2)$, while the optimal solution is equal to $x_1+x_3=1$.
        Now recall that $x_1+x_2>1$. This means that $x_3=1-x_1<x_2<x_2\cdot (1+\beta)$, so $x_3-\beta\cdot x_2<x_2$. This means that \alglower's competitive ratio is at least as high as in \casename~\symbiosisCaseRejectXiii, which we already showed was at least $1/\fo$.
    \item[\casename~\symbiosisCaseRejectYi.] In this situation, \alglower has reserved $x_1$ and then rejected $y_1$. 
        Its competitive ratio is therefore at least
        \[
            \frac{y_1}{x_1-\alpha\cdot x_1}\xrightarrow{\eps\to 0} \frac{1-\frac{1-\fo}{1+\alpha}}{(1-\alpha)\cdot \frac{1-\fo}{1+\alpha}}=\frac{\alpha+\fo}{(1-\alpha)\cdot (1-\fo)}
        \;.\]
        Our goal is therefore to show that 
        \[
            \frac{\alpha+\fo}{(1-\alpha)\cdot (1-\fo)}\geq \frac{1}{\fo}\;.
        \]
        Note that 
        \begin{equation}\label{eq:toshow}
            \frac{\alpha+\fo}{(1-\alpha)\cdot (1-\fo)}\geq \frac{1}{\fo}\iff \fo\cdot (\alpha+\fo)\geq (1-\alpha)\cdot (1-\fo)\iff \fo^2+\fo+\alpha-1\geq 0\;.
        \end{equation}
        Now recall from \cref{def:Usymb-f} that, for $\beta<1+\alpha$, the value $\fo$ is the unique root of 
        \[
            \fpolynomial{\alpha}{\beta}(x)=(1+\alpha-\beta)\cdot x^2+(1+\beta)\cdot x-1\;.
        \]
        We also  interpret \cref{eq:toshow} as a function of $\fo$:
        \begin{align*}
            \rejectyifunca(x)&=x^2+x+\alpha-1\;.
        \end{align*}
        Then by \cref{eq:toshow}, we need to show that $\rejectyifunca(\fo)\geq 0$. 
        Note that $\rejectyifunca$ has a unique positive root 
        \[
            \uniquerootyi\coloneq\frac{\sqrt{5-4\alpha}-1}{2}\;.
        \]
        Since $0<\alpha<1$, we know that $0<\uniquerootyi< 1$. This means that we can define $\uniquebetayi$ as the value of $\beta$ such that $\fpolynomial{\alpha}{\beta}(\uniquerootyi)=0$, \ie, 
        \begin{equation*}
            \uniquebetayi\cdot (\uniquerootyi-\uniquerootyi^2)+(1+\alpha)\uniquerootyi^2+\uniquerootyi-1=0\;,
        \end{equation*}
        or equivalently,
        \begin{equation*}
            \uniquebetayi=\frac{1-\uniquerootyi-(1+\alpha)\cdot\uniquerootyi^2}{\uniquerootyi-\uniquerootyi^2}\;,
        \end{equation*}
        which, using the fact that $\psi_{\alpha}(\uniquerootyi)=0$, can be simplified to
        \begin{align*}
            \frac{1-\uniquerootyi-(1+\alpha)\uniquerootyi^2}{\uniquerootyi-\uniquerootyi^2}&=\frac{1-\uniquerootyi-(1+\alpha)(1-\alpha-\uniquerootyi)}{\uniquerootyi-(1-\alpha-\uniquerootyi)}\\
            &=\frac{\alpha\cdot \uniquerootyi+\alpha^2}{2\uniquerootyi+\alpha-1}\\
            &=\frac{\alpha\cdot (\uniquerootyi+\alpha)}{2\uniquerootyi+\alpha-1}\cdot \frac{2\uniquerootyi+3-\alpha}{2\uniquerootyi+3-\alpha}\\
            &=\frac{\alpha\cdot (2\uniquerootyi^2+\uniquerootyi(2\alpha+3-\alpha)+(3-\alpha)\cdot \alpha)}{4\uniquerootyi^2+2\uniquerootyi\cdot (3-\alpha+\alpha-1)+(3-\alpha)(\alpha-1)}\\
            &=\frac{\alpha\cdot (2\cdot (1-\alpha-\uniquerootyi)+\uniquerootyi(3+\alpha)+(3-\alpha)\cdot \alpha)}{4\uniquerootyi^2+4\uniquerootyi+4\alpha-\alpha^2-3}\\
            &=\frac{\alpha\cdot (\uniquerootyi\cdot (\alpha+1)-(\alpha^2-\alpha-2))}{4\uniquerootyi^2+4\uniquerootyi+4\alpha-4+1-\alpha^2}\\
            &=\frac{\alpha\cdot (\uniquerootyi\cdot (\alpha+1)-(\alpha+1)\cdot (\alpha-2))}{1-\alpha^2}\\
            &=\frac{\alpha\cdot (\uniquerootyi\cdot (\alpha+1)-(\alpha+1)\cdot (\alpha-2))}{(1+\alpha)(1-\alpha)}\\
            &=\frac{\alpha\cdot (\uniquerootyi- \alpha+2)}{1-\alpha}\;.
        \end{align*}
        Now recall that 
        \[
            \uniquerootyi=\frac{\sqrt{5-4\alpha}-1}{2}
        \;,\]
        so 
        \begin{equation}\label{eq:expressionuniquebetayi}
            \uniquebetayi=\frac{\alpha\cdot (\uniquerootyi- \alpha+2)}{1-\alpha}=\frac{\alpha\cdot(3-2\alpha+\sqrt{5-4\alpha})}{2(1-\alpha)}
        \;.\end{equation}
        By definition of $\Usymb$ in \cref{def:Usymb-f}, this means that $\beta< \uniquebetayi$.
        
        Now observe that $\fpolynomial{\alpha}{\beta}(\uniquerootyi)$ is linear in $\beta$ with coefficient $\uniquerootyi-\uniquerootyi^2=\uniquerootyi\cdot (1-\uniquerootyi)>0$, and that it vanishes at $\beta=\uniquebetayi$ by the definition of $\uniquebetayi$. Hence $\beta<\uniquebetayi$ implies $\fpolynomial{\alpha}{\beta}(\uniquerootyi)<0$. Since $\fo$ is the unique positive root of $\fpolynomial{\alpha}{\beta}(x)$, this gives $\fo > \uniquerootyi$, which in turn implies $\rejectyifunca(\fo) \geq 0$.
    \item[\casename~\symbiosisCaseReserveYi.] In this situation, \alglower has reserved $x_1$ and $y_1$, and $\alpha\geq 1/3$, so the instance stops. 
        If \alglower packs $x_1$, its gain is $x_1-\alpha\cdot (x_1+y_1)<x_1-\alpha\cdot x_1$, so its competitive ratio is at least
        \[
            \frac{y_1}{x_1-\alpha\cdot x_1}\geq \frac{1}{\fo}\;,
        \]
        as shown in \casename~\symbiosisCaseRejectYi. If, on the other hand, \alglower packs $y_1$, its gain is
        \[
            y_1-\alpha\cdot (x_1+y_1)=1-x_1+\eps-\alpha\cdot (1+\eps) \xrightarrow{\eps\to 0} 1-\frac{1-\fo}{1+\alpha}-\alpha
        \;.\]
        Since the optimal solution is at least $y_1=1-x_1+\eps=1-(1-\fo)/(1+\alpha)$, \alglower's competitive ratio (as $\eps\to 0$) is at least
        \[
            \frac{1-\frac{1-\fo}{1+\alpha}}{1-\frac{1-\fo}{1+\alpha}-\alpha}=\frac{\frac{\fo+\alpha}{1+\alpha}}{\frac{\fo+\alpha-\alpha-\alpha^2}{1+\alpha}}=\frac{\fo+\alpha}{\fo-\alpha^2}\;.
        \]
        If $\fo\leq \alpha^2$, the gain is nonpositive in the limit $\eps\to 0$ and \alglower's competitive ratio is unbounded, so assume $\fo>\alpha^2$. Our goal is to show that the ratio above, too, is at least $1/\fo$. This means that we need to show that
        \[
            \fo\cdot (\fo+\alpha)\geq \fo-\alpha^2
        \;,\]
        which is equivalent to
        \[
            \fo^2+\fo\cdot (\alpha-1)+\alpha^2\geq 0\;.
        \]
        This however is the case because $\alpha\geq 1/3$:
        \[
            \fo^2+\fo\cdot (\alpha-1)+\alpha^2\geq \fo^2-\fo\cdot \frac{2}{3}+\frac{1}{9}=\left(\fo-\frac{1}{3}\right)^2\geq 0
        \;.\]
    \item[\casename~\symbiosisCaseYii.] In this situation, \alglower has reserved $x_1=(1-\fo)/(1+\alpha)+\eps$ and packed $y_1=1-x_1+\eps$. It then receives an item $y_2=1-x_1$. 
        We need to convince ourselves that $y_1>x_1$, so $y_1$ and $y_2$ cannot be packed simultaneously. 
        Note the equivalence
        \begin{equation}\label{eq:f_condition_CaseYii}
            1-\frac{1-\fo}{1+\alpha}> \frac{1-\fo}{1+\alpha}\iff \frac{\fo+\alpha}{1+\alpha}>\frac{1-\fo}{1+\alpha}\iff 2\fo>1-\alpha\;,
        \end{equation}
        which means that, as long as $2\fo>1-\alpha$, for sufficiently small values of $\eps$, we can guarantee that $y_1>x_1$. 
        This we do by using \cref{lem:f_non-increasing}, which states that $\fo$ is decreasing in $\beta$. Since $\beta<1+\alpha$ on $\Usymb$, we know that $\fo(\alpha,\beta)>\fo(\alpha,1+\alpha)$. Recall that $\fo$ is the unique solution to the equation 
        \[
            (1+\alpha-\beta)\cdot \fo^2+(1+\beta)\cdot \fo-1=0\;,
        \]
        so $\fo(\alpha,1+\alpha)=1/(2+\alpha)$, which means that
        \begin{equation}\label{eq:lowerboundf}
            \fo(\alpha,\beta)>\frac{1}{2+\alpha}\;,
        \end{equation}
        for any $(\alpha,\beta)\in\Usymb$.
        Therefore, 
        \[
            2\fo> \frac{2}{2+\alpha}> 1-\alpha\;,
        \]
        because
        \[
            (2+\alpha)\cdot (1-\alpha)=2-\alpha^2-\alpha<2\;,
        \]
        for $\alpha>0$. Therefore, \cref{eq:f_condition_CaseYii} holds and we know that $y_1+y_2>1$.
    \item[\casename~\symbiosisCaseRejectYii.] In this situation, \alglower has reserved $x_1$, packed $y_1$, and rejected $y_2$. 
        The optimal solution has value $x_1+y_2=1$. If \alglower removes $y_1$ and packs $x_1$, its gain is $x_1-\alpha\cdot x_1-\beta\cdot y_1<x_1-\alpha\cdot x_1$. Its competitive ratio is therefore at least
        \[
            \frac{1}{x_1-\alpha\cdot x_1}\geq \frac{y_1}{x_1-\alpha\cdot x_1}\geq \frac{1}{\fo}\;,
        \]
        as shown in \casename~\symbiosisCaseRejectYi, where \alglower rejected $y_1$. 

        If \alglower does not remove $y_1$, its gain is 
        \[
            y_1-\alpha\cdot x_1\xrightarrow{\eps\to 0} 1-\frac{1-\fo}{1+\alpha}-\alpha\cdot \frac{1-\fo}{1+\alpha}=1-(1+\alpha)\cdot \frac{1-\fo}{1+\alpha}=\fo\;,
        \]
        so its competitive ratio is at least $1/\fo$.
    \item[\casename~\symbiosisCasePackYii.] In this situation, \alglower has reserved $x_1$, packed $y_1$, and packed $y_2$. 
        Since we know from \casename~\symbiosisCaseYii that $y_1+y_2>1$, this means that it must have removed $y_1$. The optimal solution has value $x_1+y_2=1$, while \alglower's gain is at most $x_1+y_2-\alpha\cdot x_1-\beta\cdot y_1=1-\alpha\cdot x_1-\beta\cdot y_1$.
        We now show that this gain is no larger than the gain of $y_1-\alpha\cdot x_1$ \alglower achieved in \casename~\symbiosisCaseRejectYii, where it rejected $y_2$. In other words, we show that
        \[
            1-\alpha\cdot x_1-\beta\cdot y_1\leq y_1-\alpha\cdot x_1\;.
        \]
        Now note the equivalence
        \begin{align*}
            1-\alpha\cdot x_1-\beta\cdot y_1\leq y_1-\alpha\cdot x_1&\iff y_1\geq \frac{1}{1+\beta}\\
            &\iff 1-\frac{1-\fo}{1+\alpha}\geq \frac{1}{1+\beta}\\
            &\iff \frac{\fo+\alpha}{1+\alpha}\geq \frac{1}{1+\beta}\\
            &\iff \fo\geq \frac{1+\alpha}{1+\beta}-\alpha=\frac{1-\alpha\cdot \beta}{1+\beta}\;.
        \end{align*}
        Because of this, combined with the fact that $\fo$ is the only positive root of the polynomial $\fpolynomial{\alpha}{\beta}(x)=(1+\alpha-\beta)\cdot x^2+(1+\beta)\cdot x-1$, and $\fpolynomial{\alpha}{\beta}(0)<0$, it suffices to show that
        \[
            (1+\alpha-\beta)\cdot\frac{(1-\alpha\cdot\beta)^2}{(1+\beta)^2}+(1-\alpha\cdot\beta)-1\leq 0\;,
        \]
        which is equivalent to
        \begin{equation}\label{eq:showremoveYii}
            (1+\alpha-\beta)\cdot(1-\alpha\cdot\beta)^2- \alpha\cdot\beta\cdot (1+\beta)^2\leq 0\;.
        \end{equation}
        We define
        \begin{align*}
            \packyiifunc{\beta}&\da (1+\alpha-\beta)\cdot (1-\alpha\cdot\beta)^2- \alpha\cdot\beta\cdot (1+\beta)^2\\
            &=\beta^3\cdot(-\alpha^2-\alpha)+\beta^2\cdot(\alpha^3+\alpha^2)+\beta\cdot(-2\alpha^2-3\alpha-1)+\alpha+1\;.
        \end{align*}    
        
        We first show that $\packyiifunc{\beta}$ is decreasing in $\beta$ for $(\alpha,\beta)\in\Usymb$ by considering the derivative:
        \begin{equation}\label{eq:diffrho}
            \diff{\beta}\,\packyiifunc{\beta}=3\beta^2\cdot (-\alpha^2-\alpha)+2\beta\cdot(\alpha^3+\alpha^2)-2\alpha^2-3\alpha-1\;.
        \end{equation}
        For $1/4\leq \alpha\leq 1/2$, we know that $\beta<1+\alpha\leq 3/2$, so we can bound the derivative by
        \begin{align*}
            \diff{\beta}\,\packyiifunc{\beta}&<2\beta\cdot(\alpha^3+\alpha^2)-3\alpha-1\\
            &\leq 3\cdot \left(\frac{1}{8}+\frac{1}{4}\right)-\frac{3}{4}-1\\
            &=-\frac{5}{8}<0\;.
        \end{align*}
        For $1/2\leq \alpha\leq \phi-1<7/10$, we know that $1/2\leq \beta\leq 1+\alpha\leq 17/10$, so
        \begin{align*}
            \diff{\beta}\,\packyiifunc{\beta}&=3\beta^2\cdot (-\alpha^2-\alpha)+2\beta\cdot(\alpha^3+\alpha^2)-2\alpha^2-3\alpha-1\\
            &\leq \frac{3}{4}\cdot \left(-\frac{1}{4}-\frac{1}{2}\right)+4\beta\cdot \alpha^2-2\alpha^2-3\alpha-1\\
            &\leq -\frac{9}{16}+\frac{34}{5}\cdot \alpha^2-2\alpha^2-3\alpha-1\\
            &\leq\frac{24}{5}\cdot \alpha^2-\frac{3}{2}-\frac{25}{16}\\
            &\leq \frac{24}{5}\cdot \frac{49}{100}-\frac{49}{16}\\
            &\leq\frac{12}{5}-\frac{49}{16}=-\frac{53}{80}<0\;.
        \end{align*}
        We now make a second case distinction between $\alpha\leq 1/2$ and $\alpha>1/2$. For $\alpha\leq 1/2$, since $\packyiifunc{\beta}$ is decreasing in $\beta$ and $\beta\geq 1-\alpha$, it suffices to show that $\packyiifunc{1-\alpha}\leq 0$. We compute
        \begin{alignat*}{2}
            \packyiifunc{1-\alpha}&=\;&&(1-\alpha)^3\cdot (-\alpha^2-\alpha)+(1-\alpha)^2\cdot (\alpha^3+\alpha^2)\\
            &&&+(1-\alpha)\cdot (-2\alpha^2-3\alpha-1)+\alpha+1\\
            &=\;&&2\alpha^5-3\alpha^4+\alpha^3+4\alpha^2-2\alpha\\
            &=\;&&\alpha\cdot (2\alpha^4-3\alpha^3+\alpha^2+4\alpha-2)\;.
        \end{alignat*}
        Now the function $2\alpha^4-3\alpha^3+\alpha^2+4\alpha-2$ is increasing for $0<\alpha<1/2$, since its derivative satisfies
        \[
            8\alpha^3-9\alpha^2+2\alpha+4>4-9\alpha^2>4-9/4>0\;.
        \]
        Since, additionally,
        \[
            2\cdot \left(\frac{1}{2}\right)^4-3\cdot\left(\frac{1}{2}\right)^3+\left(\frac{1}{2}\right)^2+4\cdot \frac{1}{2}-2=0\;,
        \]
        we know that $\packyiifunc{1-\alpha}\leq \packyiifunc[1/2]{1/2}=0$, for any $1/4\leq \alpha\leq 1/2$. 
        
        For $\alpha>1/2$, recall that $\beta>\alpha$ by \cref{lemma:betagreateralpha}. Since $\packyiifunc{\beta}$ is decreasing in $\beta$, it suffices to show that $\packyiifunc{\alpha}\leq 0$. We compute
        \begin{alignat*}{2}
            \packyiifunc{\alpha}&=\;&&\alpha^3\cdot (-\alpha^2-\alpha)+\alpha^2\cdot (\alpha^3+\alpha^2)\\
            &&&+\alpha\cdot (-2\alpha^2-3\alpha-1)+\alpha+1\\
            &=\;&&-2\alpha^3-3\alpha^2+1\\
            &=\;&&(\alpha+1)^2\cdot (1-2\alpha)\;,
        \end{alignat*}
        which is negative for $\alpha>1/2$. 
    \item[\casename~\symbiosisCaseZi.] In this situation, \alglower has reserved $x_1$ and $y_1$ and because $\alpha< 1/3$, it received an item $z_1=\fo+\alpha\cdot (x_1+y_1)$. 
        We convince ourselves in this case that $z_1$ is a valid item size. Clearly, $z_1\geq 0$. Additionally, we know that $\fo\leq 1/2$ because of \cref{lem:f_piecewise}, more specifically because of \cref{eq:fo<1/2_iff_b>1-a}. Therefore, 
        \[
            z_1=\fo+\alpha\cdot (x_1+y_1)=\fo+\alpha\cdot (1+\eps)\leq \frac{1}{2}+\frac{1}{3}\cdot (1+\eps)<1\;,
        \]
        for sufficiently small values of $\eps$.
        Additionally, recall that 
        \[
            y_1=1-x_1+\eps=\frac{\fo+\alpha}{1+\alpha}=\fo+\alpha\cdot \frac{1-\fo}{1+\alpha}=\fo+\alpha\cdot (x_1-\eps)\leq \fo+\alpha\cdot x_1\;.
        \]
        Therefore, $z_1=\fo+\alpha\cdot (x_1+y_1)> y_1$.
    \item[\casename~\symbiosisCaseRejectZi.] In this situation, \alglower has reserved $x_1$ and $y_1$, and rejected $z_1$. Also note that $\alpha<1/3$.
        Now recall that we showed in \casename~\symbiosisCaseRejectXii that 
        \[
            \frac{x_2}{x_1} \xrightarrow{\eps\to 0} \frac{1}{\fo}>2\;,
        \]
        since we know by \cref{lem:f_piecewise} that $\fo<1/2$ for $\beta>1-\alpha$ which we assumed. Hence $x_2>2x_1$ for sufficiently small $\varepsilon$, and since $x_2 \leq 1$ by \casename~\symbiosisCaseXii, this implies that $x_1<1/2$, and thus $y_1=1-x_1+\eps>1/2>x_1$. This means that \alglower's gain is at most $y_1-\alpha\cdot (x_1+y_1)\leq y_1-\alpha$. If $y_1-\alpha\leq 0$, \alglower's competitive ratio is unbounded. Otherwise, its competitive ratio is at least
        \[
            \frac{z_1}{y_1-\alpha}\xrightarrow{\eps\to 0} \frac{\fo+\alpha}{1-\frac{1-\fo}{1+\alpha}-\alpha}=\frac{\fo+\alpha}{\frac{\fo+\alpha}{1+\alpha}-\alpha}\;.
        \]
        Our goal is therefore to show 
        \[
            \frac{\fo+\alpha}{\frac{\fo+\alpha}{1+\alpha}-\alpha}\geq \frac{1}{\fo}\;,
        \]
        which is equivalent to
        \[
            \fo\cdot(\fo+\alpha)\geq \frac{\fo+\alpha}{1+\alpha}-\alpha
        \]
        or
        \[
            (1+\alpha)\cdot \fo^2+(\alpha^2+\alpha-1)\cdot \fo+\alpha^2\geq 0\;.
        \]
        Now, for any fixed value of $\alpha$, the polynomial
        \[
            \rejectzifuncalpha(x)=(1+\alpha)\cdot x^2+(\alpha^2+\alpha-1)\cdot x+\alpha^2
        \]
        is increasing in $x\geq 1/(2+\alpha)$, as we can see by considering the derivative:
        \begin{align*}
            \diff{x}\,\rejectzifuncalpha(x)&=2x\cdot (1+\alpha)+\alpha^2+\alpha-1\geq \frac{2(1+\alpha)}{2+\alpha}+\alpha^2+\alpha-1\\
            &> 1+\alpha^2+\alpha-1>0\;.
        \end{align*}
        Since we know from \cref{eq:lowerboundf} that $\fo\geq 1/(2+\alpha)$, it suffices to check that $\rejectzifuncalpha(1/(2+\alpha))>0$:
        \begin{align*}
            \rejectzifuncalpha\left(\frac1{2+\alpha}\right)&=(1+\alpha)\cdot \left(\frac{1}{2+\alpha}\right)^2+(\alpha^2+\alpha-1)\cdot \frac{1}{2+\alpha}+\alpha^2\\
            &=\left(\frac{1}{2+\alpha}\right)^2\cdot (1+\alpha+(2+\alpha)\cdot  (\alpha^2+\alpha-1)+(2+\alpha)^2\cdot \alpha^2)\\
            &=\left(\frac{1}{2+\alpha}\right)^2\cdot (\alpha^4+5\alpha^3+7\alpha^2+2\alpha-1)\;.
        \end{align*}
        The first factor of this term is positive, while the second is clearly increasing in $\alpha$, so since $\alpha\geq 1/4$, it suffices to check that 
        \[\left(\frac{1}{4}\right)^4+5\left(\frac{1}{4}\right)^3+7\left(\frac{1}{4}\right)^2+2\cdot \frac{1}{4}-1=\frac{5}{256}>0\;.\]
    \item[\casename~\symbiosisCaseReserveZi.] In this situation, \alglower has reserved $x_1$, $y_1$, and $z_1$. Also note that $\alpha<1/3$.    
        Recall from \casename~\symbiosisCaseZi that $z_1> y_1$, and from \casename~\symbiosisCaseRejectZi that $y_1> x_1$. Since $x_1+y_1>1$, this also means that no two of these items can be packed simultaneously. \alglower's gain is therefore at most
        $z_1-\alpha\cdot (x_1+y_1+z_1)$. We also showed in \casename~\symbiosisCaseZi that $y_1=\fo+\alpha\cdot (x_1-\eps)$, so
        \[
            z_1-y_1=\fo+\alpha\cdot (x_1+y_1)-(\fo+\alpha\cdot (x_1-\eps))=\alpha\cdot y_1+\alpha\cdot \eps<\alpha\cdot z_1\;,
        \]
        for sufficiently small values of $\eps$.
        Therefore, \alglower's gain satisfies
        \[
            z_1-\alpha\cdot (x_1+y_1+z_1)\leq z_1-\alpha\cdot (x_1+y_1)-(z_1-y_1)=y_1-\alpha\cdot (x_1+y_1)\;,
        \]
        and its competitive ratio is at least
        \[
            \frac{z_1}{y_1-\alpha\cdot (x_1+y_1)}\geq \frac{1}{\fo}\;,
        \]
        as shown in \casename~\symbiosisCaseRejectZi.
    \item[\casename~\symbiosisCaseRejectZii.] In this situation, \alglower has reserved $x_1$ and $y_1$, packed $z_1$ and rejected the final item $z_2$. 
        As in \casename~\symbiosisCaseReserveZi, we know that $z_1>y_1>x_1$ and no two of these items can be packed simultaneously. \alglower's gain is therefore at most
        $z_1-\alpha\cdot (x_1+y_1)=\fo$. Since the optimal solution has value $y_1+z_2=1$, \alglower's competitive ratio is at least $1/\fo$.
    \item[\casename~\symbiosisCasePackZii.] In this situation, \alglower has reserved $x_1$ and $y_1$, packed $z_1$, and packed $z_2=1-y_1$. 
        Since $z_1>y_1$ as shown in \casename~\symbiosisCaseZi, this means that it had to remove $z_1$. Its gain is therefore at most 
        \[
            1-\alpha\cdot (x_1+y_1)-\beta\cdot z_1<1-\alpha\cdot x_1-\beta\cdot z_1<1-\alpha\cdot x_1-\beta\cdot y_1\;.
        \]
        Its competitive ratio is therefore at least
        \[
            \frac{1}{1-\alpha\cdot x_1-\beta\cdot y_1}\;,
        \]
        which we showed in \casename~\symbiosisCasePackYii is at least $1/\fo$.
\end{description}
\end{proof}
\section{Lower Bounds for the Dominance Region}\label{sec:lower_dominance}
In this section, we restate and prove the lower bounds from \cref{thm:lowerbound_reservation_small,thm:lowerbound_reservation_medium,thm:lower-equib} as already stated in \cref{sec:dominance}, beginning with \cref{thm:lower-equib}, whose adversarial tree is the simplest of the three. 

\theoremLowerEquilibrium*
\begin{proof}
The proof is an extension of a lower bound presented by \citefull{HKM2014}. The structure is shown in \cref{fig:lbequi}.
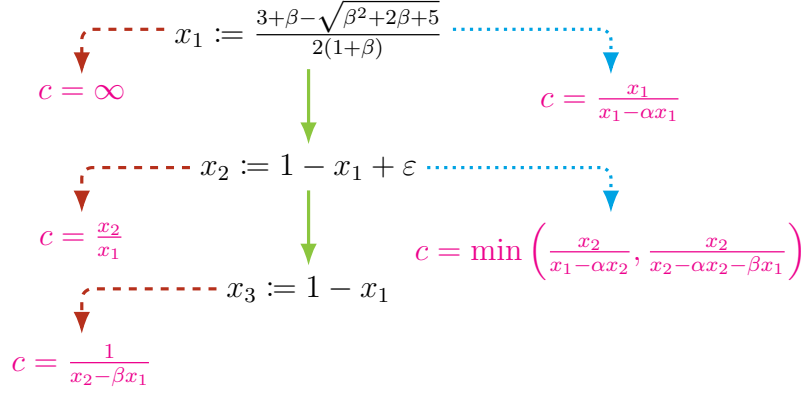
\begin{figure}
    \centering
    \begin{tikzpicture}[
    font=\large,
    >=Latex
]
    \def\rejdist{3cm}
    \def\resdist{4cm}
    \node[itemsettings] (x1)
        {$x_1 \da \frac{3 + \beta - \sqrt{\beta^2 + 2 \beta + 5}}{2 (1 + \beta)}$};
    \node[itemsettings, 
          ifpack=x1] (x2)
        {$x_2 \da 1 - x_1 + \varepsilon$};
    \node[itemsettings, 
          ifpack=x2] (x3)
        {$x_3 \da 1 - x_1$};
    
    \node[leafsettings,
          ifreject=x1] (x1rej)
        {$c = \infty$};
    \node[leafsettings,
          ifreserve=x1] (x1res)
        {$c = \frac{x_1}{x_1 - \alpha x_1}$};
    \node[leafsettings,
          ifreject=x2] (x2rej)
        {$c = \frac{x_2}{x_1}$};
    \node[leafsettings,
          ifreserve=x2] (x2res)
        {$c = \min \left( \frac{x_2}{x_1 - \alpha x_2}, \frac{x_2}{x_2 - \alpha x_2 - \beta x_1} \right)$};
    \node[leafsettings,
          ifreject=x3] (x3rej)
        {$c = \frac{1}{x_2 - \beta x_1}$};

    \draw[rejarr] (x1.west) -| (x1rej.north);
    \draw[resarr] (x1.east) -| (x1res.north);
    \draw[packarr] (x1.south) -- (x2.north);
    \draw[rejarr] (x2.west) -| (x2rej.north);
    \draw[resarr] (x2.east) -| (x2res.north);
    \draw[packarr] (x2.south) -- (x3.north);
    \draw[rejarr] (x3.west) -| (x3rej.north);
\end{tikzpicture}
    \caption{Structure of the lower bound analyzed in \cref{thm:lower-equib}. Downward arrows (green) represent the choice to pack an item, while leftward arrows (red, dashed) and rightward arrows (blue, dotted), respectively, represent rejecting or reserving an item.}
    \label{fig:lbequi}
\end{figure}
Consider any deterministic algorithm \alglower and let $\eps>0$ be sufficiently small. \alglower first receives an item 
\[x_1=\frac{3+\beta-\sqrt{\beta^2+2\beta+5}}{2(1+\beta)}\;.\]
Note that this value satisfies
\begin{equation}
    \frac{1-x_1}{x_1}=\frac{1 + \beta + \sqrt{\beta^2 + 2\beta + 5}}{2}\label{eq:equilibriumx1}
\end{equation}
and
\begin{equation}
    \frac{1}{1-x_1-\beta\cdot x_1}=\frac{1 + \beta + \sqrt{\beta^2 + 2\beta + 5}}{2}\label{eq:equilibriumx2}\;.
\end{equation}
Also note that clearly, $x_1>0$, and that, because $\beta\geq 1/2$, 
\[
    x_1=\frac{3+\beta-\sqrt{\beta^2+2\beta+5}}{2(1+\beta)}\leq \frac{3+\beta-\sqrt{\beta^2+2\beta+1}}{2(1+\beta)}=\frac{1}{1+\beta}<1\;,
\]
so $x_1$ is a valid item size. Additionally,
\[
    \frac{1 + \beta + \sqrt{\beta^2 + 2\beta + 5}}{2}> \frac{1 + \beta + \sqrt{\beta^2}}{2}\geq 1\;,
\]
so by \cref{eq:equilibriumx1}, we know that $1-x_1>x_1$. 
If \alglower rejects $x_1$, its competitive ratio is unbounded. If it reserves the item, the instance stops and its competitive ratio is at least
\[
    \frac{x_1}{x_1-\alpha\cdot x_1}=\frac{1}{1-\alpha}\geq \ell\;.
\]
We can therefore assume it packs $x_1$. It then receives an item $x_2=1-x_1+\eps$. If it rejects this item, its competitive ratio is 
\[
    \frac{1-x_1}{x_1}\geq \ell
\]
by \cref{eq:equilibriumx1}. If it reserves it, the instance stops. \alglower's gain is $x_2-\alpha\cdot x_2-\beta\cdot x_1<x_2-\alpha\cdot x_2$ if it removes $x_1$ to pack $x_2$, or $x_1-\alpha\cdot x_2<x_2-\alpha\cdot x_2$ if it does not. In either case, its competitive ratio is at least
\[
    \frac{x_2}{x_2-\alpha\cdot x_2}=\frac{1}{1-\alpha}\geq \ell\;.
\]
Finally, if \alglower packs $x_2$ and removes $x_1$, it receives a final item $x_3=1-x_1$. Since there are no further items, reserving $x_3$ is clearly suboptimal, while since $x_2>x_3$, \alglower gains no benefit from removing $x_2$ and packing $x_3$. The optimal solution has value $x_1+x_3=1$, while \alglower has a gain of at most $x_2-\beta\cdot x_1$. Its competitive ratio is therefore at least
\[
\frac{1}{x_2-\beta\cdot x_1}\xrightarrow{\eps\to 0} \frac{1}{1-x_1-\beta\cdot x_1}=\frac{1 + \beta + \sqrt{\beta^2 + 2\beta + 5}}{2}\geq \ell\;,
\]
where the last inequality holds by \cref{eq:equilibriumx2}.
\end{proof}
\theoremLowerReservationSmall*
\begin{proof}
The proof is an extension of the lower bound by \citefull{BBHLR2021} and based on the same three item sizes $\valname_1$, $\valname_2$, $\valname_3$ satisfying the system of equations
\begin{align}
\label{eq:system1}
\lowerbound=\frac{1}{\valname_2-\alpha\cdot \valname_1}&=\frac{\valname_2}{(1-\alpha)\cdot \valname_1}\\ \intertext{and}
\label{eq:system2}
\frac{\valname_3}{\valname_2-\alpha\cdot (\valname_1+\valname_2)}&=\frac{1}{\valname_3-\alpha\cdot (\valname_1+\valname_2)}\;.
\end{align}
This system of equations is satisfied for
\begin{align}
\valname_1&=\frac{2}{3+\sqrt{5-4\alpha}}\label{eq:value_xi_1}\;,\\
\valname_2&=1-\valname_1=\frac{1+\sqrt{5-4\alpha}}{3+\sqrt{5-4\alpha}}\;,\\
\valname_3&=\frac{\alpha+\sqrt{4(\valname_2-\alpha)+\alpha^2}}{2}\;.
\end{align}
We first give an overview of the structure of the proof. The structure is illustrated in \cref{fig:lbreslow}. We then check all relevant conditions separately. Some of the cases we consider build upon each other. They are therefore presented in the order in which they are analyzed later, which may make the enumeration initially seem somewhat arbitrary in places.
\begin{figure}
    \centering
    \begin{tikzpicture}[
    font=\large,
    >=Latex
]
    \def\rejdist{3cm}
    \def\resdist{4cm}
    
    \node[itemsettings] (x1)
        {$x_1\da \valname_1+\eps$};
    \node[itemsettings, 
          ifpack=x1] (x2)
        {$x_2 \da \lowerbound\cdot\valname_1$};
    \node[itemsettings,
          ifpack=x2] (f1)
        {$1$};
    \node[itemsettings,
          ifreservedepth3=x1] (y1)
        {$y_1 \da \valname_2$};
    \node[itemsettings,
          ifpack=y1] (f2)
        {$1$};
    \node[itemsettings,
          ifreservedepth2=y1,
          xshift=1cm] (z1)
        {$z_1 \da \valname_3$};
    \node[itemsettings,
          ifpack=z1] (f3)
        {$1$};
    
    \node[leafsettings,
          ifreject=x1] (x1rej)
        {$c = \infty$};
    \node[leafsettings,
          ifreject=x2] (x2rej)
        {$c = \frac{x_2}{x_1}$};
    \node[leafsettings,
          anchor=west,
          below right = 0.5cm and 1.7cm of x2.east] (x2res)
        {$c = \min \left( \frac{x_2}{x_1 - \alpha x_2}, \frac{x_2}{x_2 - \alpha x_2 - \beta x_1} \right)$};
    \node[leafsettings,
          ifreject=f1] (f1rej)
        {$c \geq \frac{1}{x_2 - \beta x_1}$};
    \node[leafsettings,
          ifpack=f1] (f1pack)
        {$c \geq \frac{1}{1 - \beta \cdot(x_1 + x_2)}$};
    \node[leafsettings,
          ifreject=y1] (y1rej)
        {$c = \frac{y_1}{x_1 - \alpha x_1}$};
    \node[leafsettings,
          ifreject=f2] (f2rej)
        {$c \geq \frac{1}{y_1 - \alpha x_1}$};
    \node[leafsettings,
          ifpack=f2] (f2pack)
        {$c \geq \frac{1}{1 - \alpha x_1 - \beta y_1}$};
    \node[leafsettings,
          ifreject=z1] (z1rej)
        {$c = \frac{z_1}{y_1 - \alpha \cdot(x_1 + y_1)}$};
    \node[leafsettings,
          ifreserve=z1] (z1res)
        {$c = \frac{z_1}{z_1 - \alpha \cdot(x_1 + y_1 + z_1)}$};
    \node[leafsettings,
          ifreject=f3] (f3rej)
        {$c \geq \frac{1}{z_1 - \alpha \cdot(x_1 + y_1)}$};
    \node[leafsettings,
          ifpack=f3] (f3pack)
        {$c \geq \frac{1}{1 - \alpha \cdot(x_1 + y_1) - \beta z_1}$};

    \draw[rejarr] (x1.west) -| (x1rej.north);
    \draw[packarr] (x1.south) -- (x2.north);
    \draw[resarr] (x1.east) -| (y1.north);
    \draw[rejarr] (x2.west) -| (x2rej.north);
    \draw[resarr over] (x2.east) -| (x2res.north);
    \draw[packarr] (x2.south) -- (f1.north);
    \draw[rejarr] (f1.west) -| (f1rej.north);
    \draw[packarr] (f1.south) -| (f1pack.north);
    \draw[rejarr] (y1.west) -| (y1rej.north);
    \draw[packarr] (y1.south) -- (f2.north);
    \draw[resarr] (y1.east) -| (z1.north);
    \draw[rejarr] (f2.west) -| (f2rej.north);
    \draw[packarr] (f2.south) -| (f2pack.north);
    \draw[rejarr] (z1.west) -| (z1rej.north);
    \draw[packarr] (z1.south) -- (f3.north);
    \draw[resarr] (z1.east) -| (z1res.north);
    \draw[rejarr] (f3.west) -| (f3rej.north);
    \draw[packarr] (f3.south) -| (f3pack.north);


    \node[numbersettings, slightlyaboveleft=x2] {\reservationCaseXii};
    \node[numbersettings, slightlyaboveleft=f1] {\reservationCasePackXii};
    \node[numbersettings, slightlyaboveleft=x2rej] {\reservationCaseRejectXii};
    \node[numbersettings, slightlyaboveleft=f1rej] {\reservationCaseRejectFinali};
    \node[numbersettings, slightlyaboveright=f1pack] {\reservationCasePackFinali};
    \node[numbersettings, slightlyaboveright=x2res] {\reservationCaseReserveXii};
    \node[numbersettings, slightlyaboveleft=y1rej] {\reservationCaseRejectYi};
    \node[numbersettings, slightlyaboveleft=f2rej] {\reservationCaseRejectFinalii};
    \node[numbersettings, slightlyaboveright=f2pack] {\reservationCasePackFinalii};
    \node[numbersettings, slightlyaboveleft=z1] {\reservationCaseZi}; 
    \node[numbersettings, slightlyaboveleft=z1rej] {\reservationCaseRejectZi};
    \node[numbersettings, slightlyaboveright=z1res] {\reservationCaseReserveZi};
    \node[numbersettings, slightlyaboveleft=f3rej] {\reservationCaseRejectFinaliii};
    \node[numbersettings, slightlyaboveright=f3pack] {\reservationCasePackFinaliii};

\end{tikzpicture}
    \caption{Structure of the lower bound analyzed in \cref{thm:lowerbound_reservation_small}.  Downward arrows (green) represent the choice to pack an item, while leftward arrows (red, dashed) and rightward arrows (blue, dotted), respectively, represent rejecting or reserving an item.}
    \label{fig:lbreslow}
\end{figure}
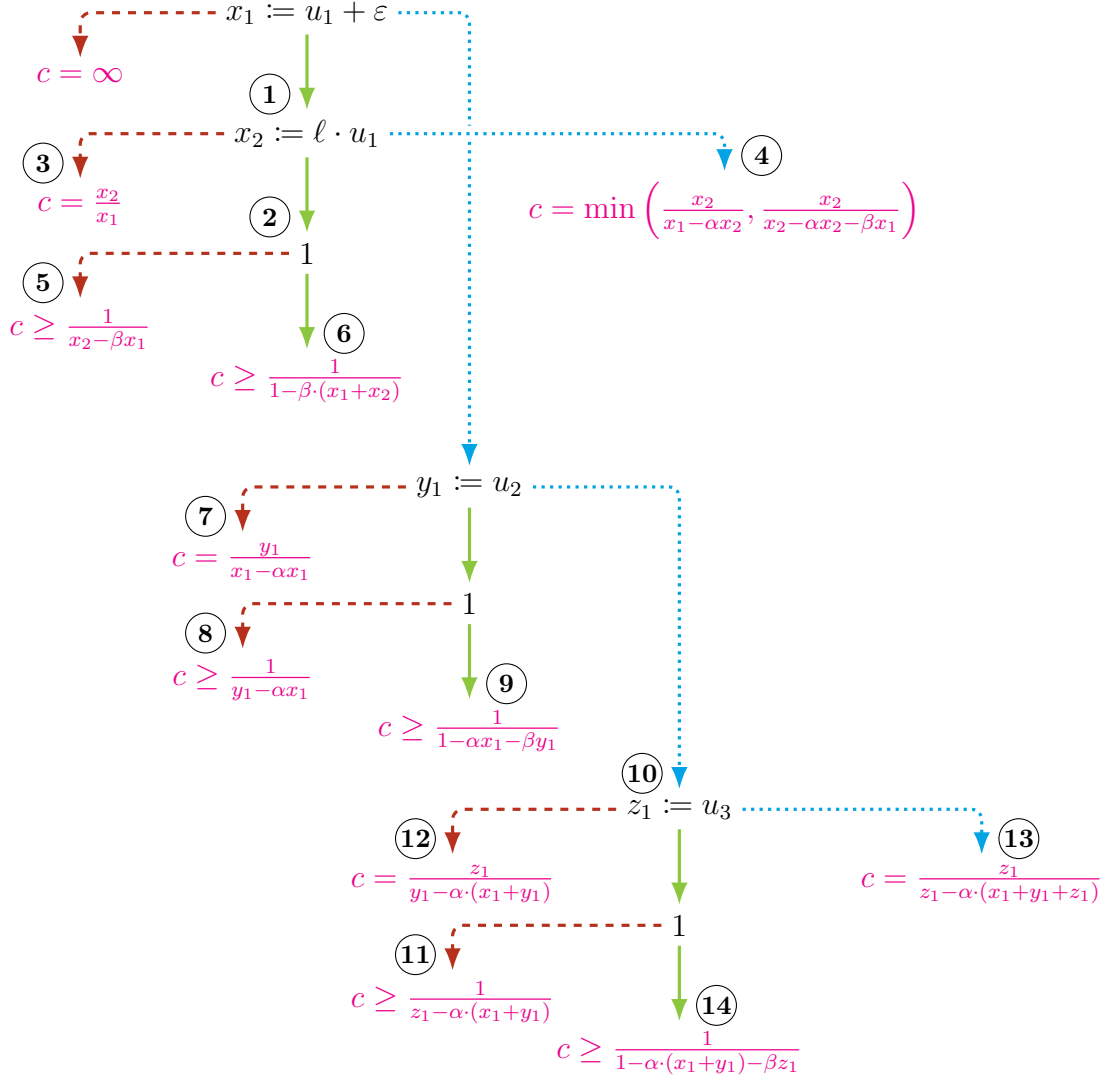

Consider any deterministic algorithm \alglower and let $\eps>0$ be sufficiently small. \alglower first receives an item $x_1\coloneq \valname_1+\eps$.
Note that, since $0<\alpha<1$, this value is well-defined and $0<x_1<1$.
If \alglower rejects $x_1$, the instance stops and \alglower's competitive ratio is unbounded.
If \alglower packs $x_1$, it receives an item $x_2=\lowerbound\cdot \valname_1$.
In \casename~\reservationCaseXii, we will check that this is indeed a valid item size. 
In \casename~\reservationCasePackXii, we consider the case where $x_2$ is packed. We check that $x_1$ and $x_2$ cannot be packed together, so $x_1$ must be removed in this case. \alglower then receives a final item of size $1$. Since this is the last item, reserving it is clearly suboptimal. We therefore only consider the cases where it is rejected (\casename~\reservationCaseRejectFinali) or packed (\casename~\reservationCasePackFinali).
If $x_2$ is rejected, the instance stops (\casename~\reservationCaseRejectXii), and the same holds if it is reserved (\casename~\reservationCaseReserveXii). 

If $x_1$ is reserved instead, \alglower receives an item $y_1\coloneq \valname_2$. If it rejects $y_1$, the instance stops (\casename~\reservationCaseRejectYi).
If it is packed, \alglower receives a final item of size $1$. Once again, since this is the last item, we only consider the case where it is rejected (\casename~\reservationCaseRejectFinalii) or packed (\casename~\reservationCasePackFinalii).

If $y_1$ is also reserved, \alglower receives an item $z_1\coloneq \valname_3$. We will check in (\casename~\reservationCaseZi) that this is a valid item. If this item is rejected, the instance stops (\casename~\reservationCaseRejectZi), and the same holds if it is reserved (\casename~\reservationCaseReserveZi). 
If it is packed, the instance receives a final item of size $1$. As before, we only consider the cases where this item is rejected (\casename~\reservationCaseRejectFinaliii) or packed (\casename~\reservationCasePackFinaliii).


We now prove that, in each of these cases, the necessary conditions hold and that, however the instance ends, \alglower's competitive ratio is at least $\lowerbound$ (in some cases, this holds only in the limit $\eps\to 0$). Recall that whenever \alglower's gain is nonpositive, its competitive ratio is unbounded; we may therefore assume throughout the case analysis that \alglower's gain is positive, so that every upper bound on the gain below is positive as well. First, we show for technical reasons that $\tilde{\beta}(\alpha)$ is increasing in $\alpha\leq \sqrt{2}-1$.
\begin{claim}\label{claim:betaboundincreasing}
The function 
\[\tilde{\beta}(\alpha)=\frac{\alpha\cdot (3+\sqrt{5-4\alpha}-2\alpha)}{2(1-\alpha)}\]
is increasing for $0<\alpha<\sqrt{2}-1$.
\end{claim}
\begin{claimproof}
The denominator of $\tilde{\beta}(\alpha)$ is clearly decreasing in $\alpha$. As for the numerator, we consider its derivative:
\begin{align*}
    \diff{\alpha}\left(\alpha\cdot (3+\sqrt{5-4\alpha}-2\alpha)\right)&=3+\sqrt{5-4\alpha}-2\alpha+\alpha\cdot \left(\frac{-4}{2\sqrt{5-4\alpha}}-2\right)\\
    &=3+\sqrt{5-4\alpha}-4\alpha-\frac{2\alpha\cdot \sqrt{5-4\alpha}}{5-4\alpha}\\
    &=3-4\alpha+\sqrt{5-4\alpha}\cdot \left(1-\frac{2\alpha}{5-4\alpha}\right).
\end{align*}
Now since $\alpha\leq \sqrt{2}-1<5/6$, the term $2\alpha/(5-4\alpha)$ is at most $1$, hence $\sqrt{5-4\alpha}\cdot (1-2\alpha/(5-4\alpha))$ is positive. Since $\alpha\leq \sqrt{2}-1<3/4$, also $3-4\alpha$ is positive. Thus the derivative is positive, which means the numerator of $\tilde{\beta}$ and thus $\tilde{\beta}$ itself are increasing with $\alpha$. 
\end{claimproof}
We now proceed with the proof itself.
\begin{description}
    \item[\casename~\reservationCaseXii.] We need to check that $x_2$ is indeed a valid item, \ie, that $0\leq x_2\leq 1$. Recall that $x_2=\lowerbound\cdot \valname_1$. Clearly, $x_2\geq 0$. Note that 
        \begin{align}
            \nonumber
            x_2&=\frac{1+\sqrt{5-4\alpha}}{2(1-\alpha)}\cdot \frac{2}{3+\sqrt{5-4\alpha}}\\ 
            \nonumber
            &=\frac{1+\sqrt{5-4\alpha}}{(1-\alpha)\cdot (3+\sqrt{5-4\alpha})}\\
            \nonumber
            &=\frac{1+\sqrt{5-4\alpha}}{(1-\alpha)\cdot (3+\sqrt{5-4\alpha})}\cdot \frac{1-\sqrt{5-4\alpha}}{1-\sqrt{5-4\alpha}}\\
            \nonumber
            &=\frac{4\alpha-4}{(1-\alpha)\cdot (3-2\sqrt{5-4\alpha}-(5-4\alpha))}\\
            \nonumber
            &=\frac{4}{2\sqrt{5-4\alpha}+2-4\alpha}\\
            &=\frac{2}{1-2\alpha+\sqrt{5-4\alpha}}\;.\label{eq:valuey1}
            \end{align}
            The expression $1-2\alpha+\sqrt{5-4\alpha}$ is decreasing in $\alpha$ for $0<\alpha<1$.
            Since $\alpha\leq \sqrt{2}-1$, it is therefore at least
            \begin{align*}
            1-2\cdot (\sqrt{2}-1)+\sqrt{5-4\cdot (\sqrt{2}-1)}&=3-2\sqrt{2}+\sqrt{9-4\sqrt{2}}\\
            &=3-2\sqrt{2}+\sqrt{(2\sqrt{2}-1)^2}\\
            &=2\;,
        \end{align*}
        so $x_2\leq 1$ is indeed a valid item size.
    \item[\casename~\reservationCasePackXii.] For this case, we need to prove that $x_2$ cannot be packed without removing $x_1$, \ie, that $x_1+x_2>1$. For this, note that
        \begin{align*}
            \valname_1+x_2&=(1+\lowerbound)\cdot\frac{2}{3+\sqrt{5-4\alpha}}\\
            &=\left(1+\frac{1+\sqrt{5-4\alpha}}{2(1-\alpha)}\right)\cdot \frac{2}{3+\sqrt{5-4\alpha}}\\
            &=\frac{3-2\alpha+\sqrt{5-4\alpha}}{2(1-\alpha)}\cdot\frac{2}{3+\sqrt{5-4\alpha}}\\
            &=\frac{3-2\alpha+\sqrt{5-4\alpha}}{2(1-\alpha)}\cdot\frac{2}{3+\sqrt{5-4\alpha}}\cdot \frac{3-\sqrt{5-4\alpha}}{3-\sqrt{5-4\alpha}}\\
            &=\frac{3-2\alpha+\sqrt{5-4\alpha}}{2(1-\alpha)}\cdot\frac{3-\sqrt{5-4\alpha}}{2+2\alpha}\\
            &=\frac{4+2\alpha\sqrt{5-4\alpha}-2\alpha}{4-4\alpha^2}\\
            &=\frac{2+\alpha\sqrt{5-4\alpha}-\alpha}{2-2\alpha^2}\\
            &\geq 1\;,
        \end{align*}
because $\sqrt{5-4\alpha}\geq 1$ and $2-2\alpha^2\leq 2$, for $0<\alpha\leq 1$.
Therefore, $x_1+x_2=\valname_1+x_2+\eps>1$.
    \item[\casename~\reservationCaseRejectXii.] In this situation, \alglower has a gain of $x_1$, while the optimal solution is at least $x_2$. \alglower's competitive ratio is therefore at least
    \[\frac{x_2}{x_1}=\frac{\lowerbound\cdot\valname_1}{\valname_1+\eps}\xrightarrow{\eps\to0} \lowerbound\;.\]
    \item[\casename~\reservationCaseReserveXii.] In this situation, \alglower has packed $x_1$ and reserved $x_2$ when the instance ends. The optimal solution has value $x_2$. If \alglower keeps $x_1$ and does not pack $x_2$, its gain is $x_1-\alpha\cdot x_2\leq x_1$, so we can see from \casename~\reservationCaseRejectXii that \alglower's competitive ratio is at least $\lowerbound$.
        If \alglower removes $x_1$ and packs $x_2$, its gain is  $x_2-\alpha\cdot x_2-\beta\cdot x_1\leq x_2-\alpha\cdot x_2-\beta\cdot \valname_1$. We will show that this, too, is at most $x_1$, allowing us to use the same comparison as before. First, we can see that
        \[
            x_2-\alpha\cdot x_2-\beta\cdot \valname_1=(1-\alpha)\cdot x_2-\beta\cdot \valname_1=((1-\alpha)\cdot \lowerbound-\beta)\cdot \valname_1\;.
        \]
        We now consider the expression
        \[
            (1-\alpha)\cdot \lowerbound-\beta\leq \frac{1+\sqrt{5-4\alpha}}{2}-\tilde{\beta}(\alpha)\;.
        \]
        The first term, $(1+\sqrt{5-4\alpha})/2$, is clearly decreasing in $0<\alpha<1$. Since $\tilde{\beta}(\alpha)$ is increasing in $\alpha\leq \sqrt{2}-1$ by \cref{claim:betaboundincreasing}, the entire expression is decreasing in $\alpha$. Since $\alpha\geq 1/4$, we thus know that
        \begin{align*}
            (1-\alpha)\cdot \lowerbound-\beta \leq 
            \frac{1+\sqrt{5-4/4}}{2}-\frac{1/4\cdot (3+\sqrt{5-4/4}-2/4)}{2(1-1/4)}=\frac{3}{4}\;.
        \end{align*}
        Hence, $((1-\alpha)\cdot \lowerbound-\beta)\cdot \valname_1\leq \valname_1\leq x_1$ and, using the bound from \casename~\reservationCaseRejectXii, \alglower has a competitive ratio of at least $\lowerbound$.
    \item[\casename~\reservationCaseRejectFinali.] In this situation, \alglower has removed $x_1$, packed $x_2$, and rejected the final item of size $1$. Its gain is therefore
        $x_2-\beta\cdot x_1\leq x_2-\beta\cdot \valname_1$.
        We therefore need to show that 
        \[\frac{1}{x_2-\beta\cdot \valname_1}\geq \lowerbound\;,\]
        which, since $\beta\geq \tilde{\beta}(\alpha)$, we will do by showing that
        \[\frac{1}{x_2-\tilde{\beta}(\alpha)\cdot \valname_1}\geq \lowerbound\;.\]
        This is the case because
        \allowdisplaybreaks\begin{align*}
            x_2-\tilde{\beta}(\alpha)\cdot \valname_1&=\lowerbound\cdot \valname_1-\tilde{\beta}(\alpha)\cdot \valname_1\\
            &= \lowerbound\cdot \frac{2}{3+\sqrt{5-4\alpha}}-\tilde{\beta}(\alpha)\cdot \frac{2}{3+\sqrt{5-4\alpha}}\\
            &=\frac{1+\sqrt{5-4\alpha}}{2(1-\alpha)}\cdot \frac{2}{3+\sqrt{5-4\alpha}}- \frac{\alpha\cdot (3+\sqrt{5-4\alpha}-2\alpha)}{2(1-\alpha)}\cdot\frac{2}{3+\sqrt{5-4\alpha}}\\
            &=\frac{1+\sqrt{5-4\alpha}}{1-\alpha}\cdot \frac{3-\sqrt{5-4\alpha}}{4+4\alpha}-\frac{\alpha\cdot (3+\sqrt{5-4\alpha}-2\alpha)}{1-\alpha}\cdot \frac{3-\sqrt{5-4\alpha}}{4+4\alpha}\\
            &=\frac{(3-\sqrt{5-4\alpha})\cdot \left(1+\sqrt{5-4\alpha}-\alpha\cdot (3+\sqrt{5-4\alpha}-2\alpha)\right)}{4\cdot (1-\alpha^2)}\\
            &=\frac{(3-\sqrt{5-4\alpha})\cdot \left(1 + \sqrt{5 - 4\alpha}\cdot (1 - \alpha) - 3 \alpha + 2 \alpha^2\right)}{4\cdot (1-\alpha^2)}\\
            &=\frac{\sqrt{5-4\alpha}\cdot (2-2\alpha^2)+3-(5-4\alpha)\cdot (1-\alpha)-9\alpha+6\alpha^2}{4\cdot (1-\alpha^2)}\\
            &=\frac{\sqrt{5-4\alpha}\cdot (2-2\alpha^2)+3-(5-4\alpha-5\alpha+4\alpha^2)-9\alpha+6\alpha^2}{4\cdot (1-\alpha^2)}\\
            &=\frac{\sqrt{5-4\alpha}\cdot (2-2\alpha^2)-2+2\alpha^2}{4\cdot (1-\alpha^2)}\\
            &=\frac{2\cdot(\sqrt{5 - 4 \alpha}-1)\cdot (1 - \alpha^2)}{4\cdot (1-\alpha^2)}\\
            &=\frac{\sqrt{5-4\alpha}-1}{2}\;,
        \end{align*}
        which means that
        \begin{align*}
            \frac{1}{x_2-\tilde{\beta}(\alpha)\cdot \valname_1}&= \frac{2}{\sqrt{5-4\alpha}-1}\\
            &=\frac{2}{\sqrt{5-4\alpha}-1}\cdot \frac{\sqrt{5-4\alpha}+1}{\sqrt{5-4\alpha}+1}\\
            &=\frac{2(\sqrt{5-4\alpha}+1)}{4-4\alpha}\\
            &=\frac{\sqrt{5-4\alpha}+1}{2-2\alpha}\\
            &=\lowerbound\;.
        \end{align*}
    \item[\casename~\reservationCasePackFinali.] In this situation, \alglower has removed $x_1$, packed $x_2$, removed $x_2$ and packed the final item of size $1$. Its gain is therefore $1-\beta\cdot (x_1+x_2)$. We now show that this gain is no larger than \alglower's gain in \casename~\reservationCaseRejectFinali, where it rejected the final item. For this, we need to show that 
        \[1-\beta\cdot (x_1+x_2)\leq x_2-\beta\cdot x_1\;,\]
        which is equivalent to
        \[x_2\geq \frac{1}{1+\beta}\;.\]
        Since $\beta\geq \tilde{\beta}(\alpha)$, it suffices to show that $x_2\geq 1/(1+\tilde\beta(\alpha))$. Since $\tilde{\beta}(\alpha)$ is increasing in $\alpha$, it is even sufficient to show that $x_2\geq 1/(1+\tilde\beta(1/4))$.
        It is clear from the expression in \eqref{eq:valuey1} that $x_2=x_2(\alpha)$ is also increasing in $\alpha$. Evaluating
        \[x_2(1/4)=\frac{2}{1-2\cdot 1/4+\sqrt{5-4\cdot1/4}}=\frac{4}{5}\]
        and
        \begin{equation}\label{eq:betatildequarter}
            \tilde{\beta}(1/4)=\frac{1/4\cdot (3+\sqrt{5-4\cdot 1/4}-2\cdot 1/4)}{2(1-1/4)}=\frac{3}{4}
        \end{equation}
        thus leads to
        \[
            x_2\geq x_2(1/4)=\frac{4}{5}> \frac{4}{7}=\frac{1}{1+3/4}=\frac{1}{1+\tilde{\beta}(1/4)}\geq\frac{1}{1+\beta}
        \;.\]
    \item[\casename~\reservationCaseRejectYi.] In this situation, \alglower has reserved $x_1$ and rejected the item $y_1=\valname_2=1-\valname_1=1-x_1+\eps$. By \cref{eq:system1}, its competitive ratio is therefore at least
        \[
            \frac{\valname_2}{(1-\alpha)\cdot (\valname_1+\eps)}\xrightarrow{\eps\to 0}\frac{\valname_2}{(1-\alpha)\cdot \valname_1}=\lowerbound\;.
        \]
    \item[\casename~\reservationCaseRejectFinalii.] In this situation, \alglower has reserved $x_1$, packed $y_1$ and rejected the final item of size $1$. Note that, since $y_1=\valname_2=1-\valname_1=1-x_1+\eps$, we know that $x_1$ and $y_1$ cannot both be packed. Additionally, since
        \[\valname_1=\frac{2}{3+\sqrt{5-4\alpha}}< \frac{2}{3+\sqrt{5-4\cdot 1}}=1/2\;,\]
        we know that
        \begin{equation}\label{eq:reservationYigreaterxi}
            \valname_2>\valname_1\;.
        \end{equation}
        \alglower's gain is therefore at most $y_1-\alpha\cdot x_1\leq \valname_2-\alpha\cdot \valname_1$ and, by \cref{eq:system1}, its competitive ratio is at least
        \[\frac{1}{\valname_2-\alpha\cdot \valname_1}=\lowerbound\;.\]
    \item[\casename~\reservationCasePackFinalii.] In this situation, \alglower has reserved $x_1$, packed $y_1$, removed $y_1$, and packed the final item of size $1$. Its gain is
        $1-\alpha\cdot x_1-\beta\cdot y_1$. We claim that this is no larger than its gain in \casename~\reservationCaseRejectFinalii, where it rejected the final item. We therefore need to show that
        \[1-\alpha\cdot x_1-\beta\cdot y_1\leq  y_1-\alpha\cdot x_1\;,\]
        which is equivalent to 
        \[y_1\geq \frac{1}{1+\beta}\;.\]
        This we do similarly to the analysis in \casename~\reservationCasePackFinali for $x_2$. Since $\beta\geq \tilde{\beta}(\alpha)$ and $\tilde{\beta}(\alpha)$ is increasing in $\alpha$, it suffices to show that $y_1\geq 1/(1+\tilde{\beta}(1/4))$. Since, additionally, $\valname_1=\valname_1(\alpha)=2/(3+\sqrt{5-4\alpha})$ is clearly increasing in $\alpha$, we know that $\valname_2(\alpha)=1-\valname_1(\alpha)$ is decreasing, so it even suffices to show that 
        $\valname_2(\sqrt{2}-1)\geq 1/(1+\tilde{\beta}(1/4))$. This is the case because
        \begin{align*}
            \valname_2(\sqrt{2}-1)&=1-\frac{2}{3+\sqrt{5-4(\sqrt{2}-1)}}
                =\frac{1+\sqrt{9-4\sqrt{2}}}{3+\sqrt{9-4\sqrt{2}}}\\
            &=\frac{1+\sqrt{(2\sqrt{2}-1)^2}}{3+\sqrt{(2\sqrt{2}-1)^2}}
                =\frac{2\sqrt{2}}{2+2\sqrt{2}}=\frac{\sqrt{2}}{1+\sqrt{2}}\\
            &=\frac{\sqrt{2}}{1+\sqrt{2}}\cdot \frac{\sqrt{2}-1}{\sqrt{2}-1}
                =2-\sqrt{2}\;,
        \end{align*}
        while $\tilde{\beta}(1/4)=3/4$ by \cref{eq:betatildequarter}, so
        \[
            y_1\geq 2-\sqrt{2}>\frac{4}{7}=\frac{1}{1+3/4}\geq \frac{1}{1+\beta}
        \;.\]
    \item[\casename~\reservationCaseZi.] In this case, we check that $z_1=\valname_3$ is a valid item size. Recall that
        \[
            \valname_3=\frac{\alpha+\sqrt{4\cdot (\valname_2-\alpha)+\alpha^2}}{2}
        \;.\]
        Since we saw in \cref{eq:reservationYigreaterxi} that $1-\valname_2=\valname_1\leq \valname_2$, we know that $\valname_2\geq 1/2\geq \alpha$, so  $z_1$ is well-defined and non-negative.  
        Additionally, since $\valname_2\leq 1$, 
        \[
            \valname_3\leq \frac{\alpha+\sqrt{4\cdot (1-\alpha)+\alpha^2}}{2}=\frac{\alpha+\sqrt{(2-\alpha)^2}}{2}=1
        \;,\]
        and, for the same reason,
        \begin{equation}\label{eq:reservationZigreaterYi}
            \valname_3\geq \frac{\alpha+\sqrt{4\valname_2\cdot(\valname_2-\alpha)+\alpha^2}}{2}=\frac{\alpha+\sqrt{(2\valname_2-\alpha)^2}}{2}=\valname_2
        \;.\end{equation}

    \item[\casename~\reservationCaseRejectFinaliii.] In this situation, \alglower has reserved $x_1$, $y_1$, packed $z_1$, and rejected the final item of size $1$. 
        Since $z_1\geq y_1\geq x_1$ by \cref{eq:reservationYigreaterxi,eq:reservationZigreaterYi} and since $y_1+x_1>1$, none of these three items can be packed together. \alglower's gain is therefore at most $z_1-\alpha\cdot (x_1+y_1)\leq \valname_3-\alpha\cdot (\valname_1+\valname_2)$ and its competitive ratio is at least
        \begin{equation}\label{eq:reject1ratio}
            \frac{1}{\valname_3-\alpha\cdot (\valname_1+\valname_2)}\;.
        \end{equation}
        Our goal is to show that this is at least $\lowerbound$. Since we already know from \eqref{eq:system1} that $1/(\valname_2-\alpha\cdot \valname_1)=\lowerbound$, it suffices to show that
        \[
            \valname_3-\alpha\cdot (\valname_1+\valname_2)\leq \valname_2-\alpha\cdot \valname_1\;,
        \]
        which is equivalent to
        \begin{equation}\label{eq:reject1condition}
            \valname_3\leq (1+\alpha)\cdot \valname_2\;.
        \end{equation}
        This we will now show. If we consider 
        \[
            \valname_3-(1+\alpha)\cdot \valname_2=\frac{\alpha+\sqrt{4(\valname_2-\alpha)+\alpha^2}}{2}-(1+\alpha)\cdot \valname_2
        \;,\]
        we can define the function
        \[
            \xthreefunc{\alpha}(x)=\frac{\alpha+\sqrt{4(x-\alpha)+\alpha^2}}{2}-(1+\alpha)\cdot x
        \;.\]
        We can then consider the derivative of $\xthreefunc{\alpha}$
        \[
            \xthreefunc{\alpha}'(x)=\frac{1}{\sqrt{4(x-\alpha)+\alpha^2}}-\alpha-1
        \;.\]
        This is clearly decreasing in $x$ as long as the numerator is positive. Since $x\geq 2-\sqrt{2}>1/2>\alpha$, this is the case. Since $x\geq 2-\sqrt{2}$ and $\alpha\geq 1/4$, we therefore know that
        \begin{equation}\label{eq:xthreefuncderiv}
            \frac{1}{\sqrt{8-4\sqrt{2}-4\alpha+\alpha^2}}-\alpha-1\leq\frac{1}{\sqrt{8-4\sqrt{2}-4\alpha+\alpha^2}}-\frac{5}{4}\;.
        \end{equation}
        To show that this is negative, consider the square of its first term,
        \begin{equation}\label{eq:xthreefuncderivpart}
            \frac{1}{8-4\sqrt{2}-4\alpha+\alpha^2}
        \;.\end{equation}
        This term is increasing in $\alpha<1/2$, because $\alpha^2-4\alpha$ is decreasing (its derivative $2\alpha-4$ is negative). Since $\alpha\leq \sqrt{2}-1$, the expression in \cref{eq:xthreefuncderivpart} is at most
        \begin{align*}
            \frac{1}{8-4\sqrt{2}-4(\sqrt{2}-1)+(\sqrt{2}-1)^2}&=\frac{1}{15-10\sqrt{2}}<\frac{25}{16}\;,
        \end{align*}
        which can be checked numerically. Therefore, 
        \[
            \frac{1}{\sqrt{8-4\sqrt{2}-4\alpha+\alpha^2}}<\frac{5}{4}\;,
        \]
        and thus by \cref{eq:xthreefuncderiv}, $\xthreefunc{\alpha}'(x)$ is negative for $x\geq 2-\sqrt{2}$, which means that $\xthreefunc{\alpha}(x)$ is decreasing. Once again, since $\valname_1=\valname_1(\alpha)$ is increasing in $\alpha$, we know that $\valname_2(\alpha)$ is decreasing, so $\valname_2\geq \valname_2(\sqrt{2}-1)=2-\sqrt{2}$. Therefore,
        \begin{align*}
            \xthreefunc{\alpha}(\valname_2)&\leq \xthreefunc{\alpha}(2-\sqrt{2})\\
            &=\frac{\alpha+\sqrt{4(2-\sqrt{2}-\alpha)+\alpha^2}}{2}-(1+\alpha)\cdot (2-\sqrt{2})\\
            &=\frac{\alpha+\sqrt{8-4\sqrt{2}-4\alpha+\alpha^2}}{2}-(1+\alpha)\cdot (2-\sqrt{2})\;.
        \end{align*}
        Now, note that 
        \[
            8-4\sqrt{2}-4\alpha<\frac{2916}{1369}-\frac{108\alpha}{37}
        \;,\] because 
        \begin{align*}
            \left(4-\frac{108}{37}\right)\cdot\alpha=\frac{40\alpha}{37}\geq \alpha\geq \frac{1}{4}>8-\frac{2916}{1369}-4\sqrt{2}=\frac{8036}{1369}-4\sqrt{2}\;,
        \end{align*}
        where the last inequality can be checked numerically.
        Therefore,
        \begin{align*}
            \xthreefunc{\alpha}(\valname_2)&\leq\frac{\alpha+\sqrt{8-4\sqrt{2}-4\alpha+\alpha^2}}{2}-(1+\alpha)\cdot (2-\sqrt{2})\\
            &< \frac{\alpha+\sqrt{\frac{2916}{1369}-\frac{108\alpha}{37}+\alpha^2}}{2}-(1+\alpha)\cdot (2-\sqrt{2})\\
            &=\frac{\alpha+\sqrt{(54/37-\alpha)^2}}{2}-(1+\alpha)\cdot (2-\sqrt{2})\\
            &=\frac{27}{37}-(1+\alpha)\cdot (2-\sqrt{2})\\
            &\leq \frac{27}{37}-\left(1+\frac{1}{4}\right)\cdot (2-\sqrt{2})\\
            &=\frac{5\sqrt{2}}{4}-\frac{131}{74}\\
            &<0\;,
        \end{align*}
        which once again can be checked numerically.
        This means that $\xthreefunc{\alpha}(\valname_2)<0$, for all $\alpha$, so $\valname_3<(1+\alpha)\cdot \valname_2$ and \cref{eq:reject1condition} holds. Therefore,
        \begin{equation}\label{eq:reject1ratioholds}
            \frac{1}{\valname_3-\alpha\cdot (\valname_1+\valname_2)}\geq \lowerbound
        \end{equation}
        as desired.
    \item[\casename~\reservationCaseRejectZi.] In this situation, \alglower has reserved $x_1$ and $y_1$, and then rejected $z_1$. 
        Since $x_1+y_1>1$ and $y_1\geq x_1$ by \cref{eq:reservationYigreaterxi}, its gain is at most $y_1-\alpha\cdot (x_1+y_1)\leq \valname_2-\alpha\cdot (\valname_1+\valname_2)$ and its competitive ratio is at least
        \[
            \frac{\valname_3}{\valname_2-\alpha\cdot (\valname_1+\valname_2)}=\frac{1}{\valname_3-\alpha\cdot (\valname_1+\valname_2)}
        \]
        by \cref{eq:system2}, which, as we showed in \eqref{eq:reject1ratioholds}, is at least $\lowerbound$. 
    \item[\casename~\reservationCaseReserveZi.] In this situation, \alglower has reserved $x_1$, $y_1$, and $z_1$. 
        As in \casename~\reservationCaseRejectZi, none of these three items can be packed together and $z_1$ is the largest of the three. \alglower's gain is at most $z_1-\alpha\cdot (x_1+y_1+z_1)\leq \valname_3-\alpha\cdot (\valname_1+\valname_2+\valname_3)$, so its competitive ratio is at least 
        \[
            \frac{\valname_3}{\valname_3-\alpha\cdot (\valname_1+\valname_2+\valname_3)}
        \;.\]
        Since we know from \cref{eq:reject1ratioholds} in \casename~\reservationCaseRejectFinaliii, together with \cref{eq:system2}, that 
        \[
            \frac{\valname_3}{\valname_2-\alpha\cdot (\valname_1+\valname_2)}\geq \lowerbound
        \;,\]
        it suffices to show that
        \[
            \valname_3-\alpha\cdot (\valname_1+\valname_2+\valname_3)\leq \valname_2-\alpha\cdot (\valname_1+\valname_2)
        \;,\]
        which is equivalent to 
        \[
            (1-\alpha)\cdot \valname_3\leq \valname_2\;.
        \]
        Now recall that we just showed in \casename~\reservationCaseRejectFinaliii that \cref{eq:reject1condition} holds, \ie, 
        that $\valname_3\leq (1+\alpha)\cdot \valname_2$ or $\valname_3/(1+\alpha)\leq \valname_2$. All we need to verify is therefore that 
        \[
            1-\alpha\leq \frac{1}{1+\alpha}
        \;,\]
        which is clearly the case since $(1-\alpha)\cdot (1+\alpha)=1-\alpha^2<1$.
    \item[\casename~\reservationCasePackFinaliii.] In the final case we need to consider, \alglower has reserved $x_1$, $y_1$, packed $z_1$, removed $z_1$, and packed the final item of size $1$. 
        Its gain is then $1-\alpha\cdot (x_1+y_1)-\beta\cdot z_1\leq 1-\alpha\cdot (\valname_1+\valname_2)-\beta\cdot \valname_3$ and its competitive ratio is at least
        \[
            \frac{1}{1-\alpha\cdot (\valname_1+\valname_2)-\beta\cdot \valname_3}
        \;.\]
        Since we know from \eqref{eq:reject1ratioholds} in \casename~\reservationCaseRejectFinaliii that 
        \[
            \frac{1}{\valname_3-\alpha\cdot (\valname_1+\valname_2)}\geq \lowerbound
        \;,\]
        it suffices to show that
        \[
            1-\alpha\cdot (\valname_1+\valname_2)-\beta\cdot \valname_3\leq \valname_3-\alpha\cdot (\valname_1+\valname_2)
        \;,\]
        which is equivalent to 
        \[(1+\beta)\cdot \valname_3\geq 1
        \;.\]
        This however is simple to check: We know that $\beta\geq \tilde{\beta}(\alpha)\geq \tilde{\beta}(1/4)=3/4$, and that $\valname_3\geq \valname_2$. We also know that $\valname_2=\valname_2(\alpha)$ is decreasing in $\alpha$, so $\valname_2\geq\valname_2(\sqrt{2}-1)=2-\sqrt{2}>4/7$, and thus
        \[(1+\beta)\cdot \valname_3>(1+3/4)\cdot 4/7=1\;.\]
\end{description}
\end{proof}

\theoremLowerReservationMedium*
\begin{proof}
The proof is based on the lower bound for the reservation-only case for $\sqrt{2}-1<\alpha\leq \phi-1$ by \citefull{BBFHLR22}. Its structure is shown in \cref{fig:lbresmid}.
\begin{figure}
    \centering
    \begin{tikzpicture}[
    font=\large,
    >=Latex
]
    \def\rejdist{2.5cm}
    \def\resdist{3cm}
    \node[itemsettings] (x1)
        {$x_1 \da \frac{1}{2 + \alpha}$};
    \node[itemsettings, 
          ifpack=x1] (x2)
        {$1$};
    \node[itemsettings,
          ifreservedepth2=x1] (y1)
        {$y_1 \da 1 - x_1 + \varepsilon$};
    \node[itemsettings,
          ifpack=y1] (y2)
        {$1$};
    
    \node[leafsettings,
          ifreject=x1] (x1rej)
        {$c = \infty$};
    \node[leafsettings,
          ifreject=x2] (x2rej)
        {$c \geq \frac{1}{x_1}$};
    \node[leafsettings,
          ifpack=x2] (x2pack)
        {$c \geq \frac{1}{1 - \beta x_1}$};
    \node[leafsettings,
          ifreject=y1] (y1rej)
        {$c = \frac{y_1}{x_1 - \alpha x_1}$};
    \node[leafsettings,
          ifreserve=y1] (y1res)
        {$c = \frac{y_1}{y_1 - \alpha\cdot (x_1 + y_1)}$};
    \node[leafsettings,
          ifpack=y2] (y2pack)
        {$c \geq \frac{1}{1 - \alpha x_1 - \beta y_1}$};
    \node[leafsettings,
          ifreject=y2] (y2rej)
        {$c \geq \frac{1}{y_1 - \alpha x_1}$};
    \draw[rejarr] (x1.west) -| (x1rej.north);
    \draw[packarr] (x1.south) -- (x2.north);
    \draw[resarr] (x1.east) -| (y1.north);
    \draw[rejarr] (x2.west) -| (x2rej.north);
    \draw[packarr] (x2.south) -- (x2pack.north);
    \draw[rejarr] (y1.west) -| (y1rej.north);
    \draw[packarr] (y1.south) -- (y2.north);
    \draw[resarr] (y1.east) -| (y1res.north);
    \draw[packarr] (y2.south) -- (y2pack.north);
    \draw[rejarr] (y2.west) -| (y2rej.north);
\end{tikzpicture}
    \caption{Structure of the lower bound analyzed in \cref{thm:lowerbound_reservation_medium}. Downward arrows (green) represent the choice to pack an item, while leftward arrows (red, dashed) and rightward arrows (blue, dotted), respectively, represent rejecting or reserving an item.}
    \label{fig:lbresmid}
\end{figure}
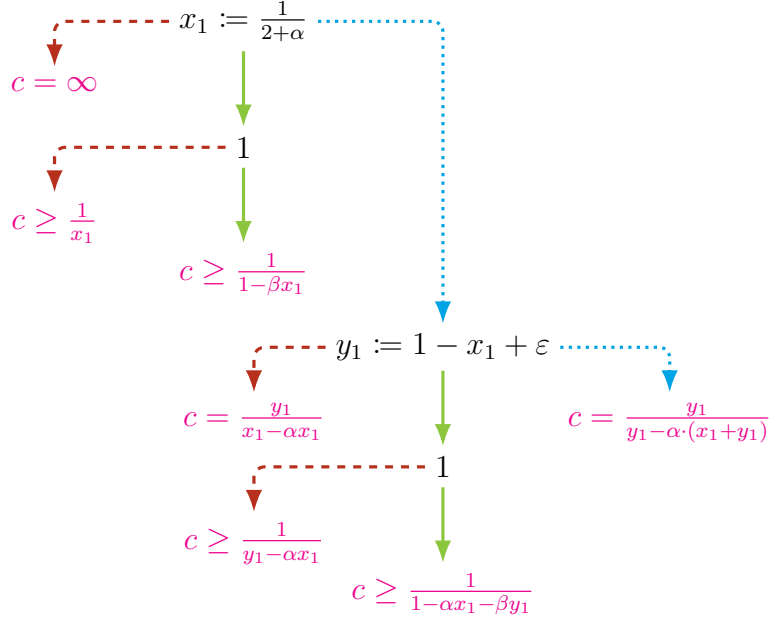
Consider any deterministic algorithm \alglower and let $\eps>0$ be sufficiently small. \alglower first receives an item $x_1\da 1/(2+\alpha)$. If it rejects this item, its competitive ratio is unbounded. 

If it packs $x_1$, it receives a final item of size $1$. Since there are no further items, reserving this item is clearly suboptimal. If \alglower rejects this item, its competitive ratio is $1/x_1=2+\alpha=\ell$. If it removes $x_1$ to pack it, its gain is
\begin{equation}\label{eq:lowerboundresmediumfinal}
    1-\beta\cdot x_1\leq 1- (1+\alpha)\cdot x_1=1-\frac{1+\alpha}{2+\alpha}=\frac{1}{2+\alpha}\;,
\end{equation}
so its competitive ratio is $2+\alpha=\ell$.

If it reserves $x_1$, it receives an item $y_1=1-x_1+\eps$. If it rejects this item, its competitive ratio is 
\begin{equation}\label{eq:lowerboundresmediumy}
    \frac{y_1}{x_1-\alpha\cdot x_1}\geq \frac{1-x_1}{(1-\alpha)\cdot x_1}=\frac{1+\alpha}{1-\alpha}\leq 2+\alpha
\;,\end{equation}
because of the equivalence
\[\frac{1+\alpha}{1-\alpha}\leq 2+\alpha\iff \alpha^2+2\alpha-1\leq 0\iff (\sqrt{2}-1-\alpha)\cdot (\sqrt{2}-1+\alpha)\leq 0\]
and the fact that $\alpha>\sqrt{2}-1$.

If \alglower reserves $y_1$, the instance stops. \alglower's gain is at most 
\[
    y_1-\alpha\cdot (x_1+y_1)\xrightarrow{\eps\to 0} (1-\alpha)\cdot (1-x_1)-\alpha\cdot x_1\;.
\]
Now we can check that
\[
        (1-\alpha)\cdot (1-x_1)-\alpha\cdot x_1=(1-\alpha)\cdot \frac{1+\alpha}{2+\alpha}-\alpha\cdot x_1=(1-\alpha^2)\cdot x_1-\alpha\cdot x_1< x_1-\alpha\cdot x_1\;.
\]
Therefore, by \cref{eq:lowerboundresmediumy}, its competitive ratio is at least $\ell$.

If \alglower packs $y_1$, it receives a final item of size $1$. Since there are no further items, reserving it is clearly suboptimal. If \alglower removes $y_1$ to pack this item, its gain is at most
\[1-\alpha\cdot x_1-\beta\cdot y_1\leq 1-\beta\cdot y_1\leq 1-\beta\cdot x_1\leq \frac{1}{2+\alpha}\]
by \cref{eq:lowerboundresmediumfinal}.
Finally, if \alglower rejects the final item, by \cref{eq:lowerboundresmediumy} its competitive ratio is at least
\[\frac{1}{y_1-\alpha\cdot x_1}\geq\frac{y_1}{y_1-\alpha\cdot x_1}\geq \ell\;.\]
\end{proof}
\section{Putting Things Together}\label{sec:assembly}

We are now ready to assemble all matching bounds into a proof of our main result, which we restate here.

\theoremMain*

\begin{proof}
    For the symbiosis region, we have proven the upper bound of $1/f(\alpha,\beta)$ in \cref{thm:alg_achieves_f}, using \algsymb. Each of the three branches of $f$ (see \cref{lem:f_piecewise}) is matched by its own lower bound: for $\beta \leq 1 - \alpha$, we have $f = \frac{1}{2}$, and \cref{thm:globaltwo} matches $\frac{1}{f} = 2$; for $\beta \leq \alpha + \frac{2\alpha - 1}{\alpha(1 - \alpha)}$, we have $f = 1 - \alpha$, and \cref{thm:lowerboundsymbiosistrivial} matches $1/f = 1/(1 - \alpha)$; and for $\beta > \max \left( 1 - \alpha, \alpha + \frac{2\alpha - 1}{\alpha(1 - \alpha)} \right)$, we have $f = \fo$, and \cref{thm:lowerboundf} matches $1/f = 1/\fo$.
        
    For the dominance regions, the case $\alpha = 0$ is trivial: if reservation does not induce any cost, an algorithm can reserve every item and afterwards optimally solve the resulting offline instance, \ie, $\cknapresrem = \cknapres = 1$. The case $\alpha \geq 1$ is again trivial: replacing every reservation by a rejection never decreases the gain since a reserved item $x$ contributes at most $x - \alpha x \leq 0$, even if it is packed later; hence $\cknapresrem = \cknaprem$. The upper bounds for the reservation-only case were proven by \citefull{BBFHLR22} and those for the removal-only case by \citefull{HKM2014}; see also \cref{obs:single_mechanism_upper}.  It remains to match $\min(\cknapres(\alpha), \cknaprem(\beta))$ from below for $0 < \alpha < 1$. Wherever $\cknapres(\alpha) = 2$ or $\cknaprem(\beta) = 2$ is the minimum, \cref{thm:globaltwo} matches it. Above $\Usymb$, \cref{thm:lowerbound_reservation_small} matches $\cknapres(\alpha)$ for $\frac{1}{4} < \alpha \leq \sqrt{2} - 1$, and \cref{thm:lowerbound_reservation_medium} matches it for $\sqrt{2} - 1 < \alpha \leq \varphi - 1$. Finally, the remaining regions are covered by \cref{thm:lower-equib}. Above the equilibrium border, where $\alpha > \varphi - 1$, its bound equals $\cknapres(\alpha) = 1/(1 - \alpha)$, because $(\alpha, \beta) \in \Ures$ gives $\cknaprem(\beta) \geq \cknapres(\alpha)$. Below the equilibrium border, where $\beta > 1/2$ and $\alpha < 1$, we have $\beta \leq \frac{3\alpha - \alpha^2 - 1}{1 - \alpha}$ --- for $\alpha \geq \varphi - 1$ this is the border formula \cref{eq:explicit_equilibrium_border} itself, and for $\alpha < \varphi - 1$ the region lies below $\Usymb$, whose lower boundary is $\max \left( \frac{1}{2}, \frac{3\alpha - \alpha^2 - 1}{1 - \alpha} \right)$ by \cref{def:Usymb-f} --- and this is equivalent to $\cknaprem(\beta) \leq 1/(1 - \alpha)$, as computed in the proof of \cref{lem:f_strictly_better}; hence the bound of \cref{thm:lower-equib} equals $\cknaprem(\beta)$.
\end{proof}

Note that the results from \cref{thm:results} are complemented by the results proven by \citefull{BGLMR23}, who analyzed the case of $\beta=0$. Together, this gives a complete picture of the competitive ratio for all possible combinations of cost parameters for reservation and removal.

\section{Conclusion}\label{sec:conclusion}

We have determined the exact competitive ratio of the online proportional knapsack problem with paid reservation and paid removal for every pair of cost parameters $(\alpha, \beta)$, with matching upper and lower bounds throughout the plane. The characterization separates the plane into three regimes: two in which a single mechanism is already optimal and the equilibrium border between them is sharp, and a symbiosis region in which reservation and removal together beat each of the single mechanisms.

This full characterization of the proportional knapsack problem, with respect to the power of paid reservation, paid removal, and any combination thereof, still leaves room for interesting generalizations. For example, considering the general online knapsack problem, in which items carry values independent of their sizes, the two mechanisms have so far only been studied in isolation: free removal alone does not make the general problem competitive, which led \citefull{IZ2010} to combine removal with resource augmentation. \citefull{BG2026} recently showed that reservation admits a tight competitive ratio of $2$ for every size-proportional cost factor, while under value-proportional costs no competitive algorithm exists once the cost factor exceeds $1/2$. To the best of our knowledge, the combination of the two mechanisms in the general setting is unexplored; in light of the symbiosis observed here, it is natural to ask whether paid removal can push the general problem below the ratio of $2$, or extend its competitive range under value-proportional reservation costs. A second direction is randomization. For the proportional knapsack without reservation or removal, a single random bit already reduces the competitive ratio from unbounded to $2$, and additional random bits do not help as shown by \citefull{BKKR2014}; whether randomization can beat the deterministic optimum established here for some, or even all, cost pairs $(\alpha, \beta)$ remains open.

\subsection*{Acknowledgments}

This work was partially supported by SNF Grant 10009262.

\subsection*{Declarations}

LLM support was used for proofreading the paper and polishing the presentation.  The authors assume responsibility for all content.

\small

\bibliographystyle{plainnat}
\bibliography{references}

\end{document}